\documentclass[%
 reprint,
 superscriptaddress,
 amsmath,amssymb,
 aps,
]{revtex4-2}

\usepackage{subcaption}
\usepackage{stackengine}
\usepackage{graphicx}% Include figure files
\usepackage{dcolumn}% Align table columns on decimal point

\usepackage{amsthm,bm}
\usepackage{braket}
\usepackage{physics}
\usepackage{tikz}
\usetikzlibrary{quantikz2}

\usepackage[ruled, vlined, linesnumbered]{algorithm2e}
\usepackage{hyperref}
\hypersetup{
  colorlinks=true,
  hypertexnames=false,
  linktocpage,
  colorlinks=true, 
  urlcolor=magenta!90!black,    
  linkcolor=blue!60!black, 
  citecolor=black!60 
}% add hypertext capabilities

\newcommand{\Asterisk}{\mathop{\scalebox{1.5}{\raisebox{-0.2ex}{$\ast$}}}}%
\usepackage[dvipsnames]{xcolor}

\newcommand{\Dom}[1]{{\color{magenta} [DH: #1]}}

\usepackage{xprintlen}
\usepackage{layouts}
\usepackage{comment}
\usepackage[capitalize]{cleveref}
\newtheorem{theorem}{Theorem}
\newtheorem{definition}[theorem]{Definition}

\newtheorem{lemma}[theorem]{Lemma}

\newtheorem{assumption}[theorem]{Assumption}

\newtheorem{apptheorem}{Theorem}[section]
\newtheorem{applemma}[apptheorem]{Lemma}
\newtheorem{appdefinition}[apptheorem]{Definition}
\newtheorem{appcorollary}[apptheorem]{Corollary}

\newtheorem{appexample}[apptheorem]{Example}
\newtheorem{appproperty}[apptheorem]{Property}

\AddToHook{env/theorem/begin}{\crefalias{theorem}{theorem}}
\AddToHook{env/lemma/begin}{\crefalias{theorem}{lemma}}
\AddToHook{env/definition/begin}{\crefalias{theorem}{definition}}
\AddToHook{env/corollary/begin}{\crefalias{theorem}{corollary}}
\AddToHook{env/proposition/begin}{\crefalias{theorem}{proposition}}
\AddToHook{env/conjecture/begin}{\crefalias{theorem}{conjecture}}
\AddToHook{env/example/begin}{\crefalias{theorem}{example}}
\AddToHook{env/property/begin}{\crefalias{theorem}{property}}
\AddToHook{env/assumption/begin}{\crefalias{theorem}{assumption}}

\AddToHook{env/apptheorem/begin}{\crefalias{apptheorem}{theorem}}
\AddToHook{env/applemma/begin}{\crefalias{apptheorem}{lemma}}
\AddToHook{env/appdefinition/begin}{\crefalias{apptheorem}{definition}}
\AddToHook{env/appcorollary/begin}{\crefalias{apptheorem}{corollary}}
\AddToHook{env/appproposition/begin}{\crefalias{apptheorem}{proposition}}
\AddToHook{env/appconjecture/begin}{\crefalias{apptheorem}{conjecture}}
\AddToHook{env/appexample/begin}{\crefalias{apptheorem}{example}}
\AddToHook{env/appproperty/begin}{\crefalias{apptheorem}{property}}
\begin{document}
%\printinunitsof{in}\prntlen{\textwidth}
% \title{Efficient Noise and Fidelity Estimation of Fault-Tolerant Quantum Circuits}
\title{In-situ benchmarking of fault-tolerant quantum circuits. I. Clifford circuits}
% \title{Exponential advantage in benchmarking fault-tolerant quantum circuits}
\author{Xiao Xiao}
\email{xiaoxiao@umd.edu}
 \affiliation{
 Joint Center for Quantum Information and Computer Science, NIST/University of Maryland,
College Park, Maryland 20742, USA.
}
\author{Dominik Hangleiter}
 \affiliation{Simons Institute for the Theory of Computing, University of California at Berkeley, Berkeley, California 94270, USA.}
  \affiliation{Institute for Theoretical Physics, ETH Z\"urich, 8093 Zürich, Switzerland.}
  
\author{Dolev Bluvstein}
\affiliation{
 Department of Physics, Harvard University, Cambridge, Massachusetts 02138, USA.
}
\author{Mikhail D. Lukin}
\affiliation{
 Department of Physics, Harvard University, Cambridge, Massachusetts 02138, USA.
}
\author{Michael J. Gullans}
\email{mgullans@umd.edu}
 \affiliation{
 Joint Center for Quantum Information and Computer Science, NIST/University of Maryland,
College Park, Maryland 20742, USA.
}
\affiliation{National Institute of Standards and Technology, Gaithersburg, MD 20899, USA.}
\date{\today}
\begin{abstract}
Benchmarking physical devices and verifying logical algorithms are important tasks for scalable fault-tolerant quantum computing. Numerous protocols exist for benchmarking devices before running actual algorithms. In this work, we show that both physical and logical errors of fault-tolerant circuits can even be characterized in-situ using syndrome data. To achieve this, we map general fault-tolerant Clifford circuits to subsystem codes using the spacetime code formalism and develop a scheme for estimating Pauli noise in Clifford circuits using syndrome data. We give necessary and sufficient conditions for the learnability of physical and logical noise from given syndrome data, and show that we can accurately predict logical fidelities from the same data. Importantly, our approach requires only a polynomial sample size, even when the logical error rate is exponentially suppressed by the code distance, and thus gives an exponential advantage against methods that use only logical data such as direct fidelity estimation. We demonstrate the practical applicability of our methods in various scenarios using synthetic data as well as the experimental data from a recent demonstration of fault-tolerant circuits by Bluvstein \textit{et al.} [Nature 626,
7997 (2024)]. 
Our methods provide an efficient, in-situ way of characterizing a fault-tolerant quantum computer to help gate calibration, improve decoding accuracy, and verify logical circuits.
\end{abstract}

\maketitle

\section{Introduction}
% \printlen[5][cm]{\textwidth}
Quantum error correction (QEC) is crucial for preserving quantum information in a noisy environment using redundant encodings.
% In order to perform computations while profiting from error correction, operations must be identified, which are compatible with the encoding while at the same time updating the logical information appropriately. 
% Such fault-tolerant operations are the basis for performing robust quantum computation. 
Recent experiments have demonstrated promising results \cite{bluvstein2024logical, acharya2024quantum,krinner2022realizing,ryan2024high,da2024demonstration,mayer2024benchmarking} that realize error-corrected quantum memory and computations. 
In these experiments, information is encoded in stabilizer codes. 
In such codes, certain Pauli operators are measured and the resulting bits of information known as the \emph{syndrome} are decoded to infer the (Pauli) errors that occurred. These errors can then be corrected.

When designing fault-tolerant circuits and developing decoders, the precise physical error model can have a significant impact on the performance of the fault-tolerant computation. 
Typically, the model is assumed to be simple, single-qubit depolarizing noise \cite{fowler2009high,fowler2012proof,chamberland2018flag,bravyi2024high}. 
However, this is an idealized model that does not adequately capture the scenarios present in different physical platforms. 
Understanding and learning error mechanisms and error rates occurring in a physical system is therefore critical in optimizing the performance of FTCs, specifically, for gate tune-up, code and circuit design, and improving decoders. 
A second, closely related challenge is to assess the quality of the outcome of an FTC in the presence of said physical noise channels. 
% While the theory of QEC predicts that logical noise rates can be suppressed exponentially, the logical error rate in any computation remains finite. 
% An assessment on the reliability of FTC therefore remains crucial, particularly in the near term. 
This is particularly important in the near term, when the accessible code distances limit the suppression of the logical error rate.

In this work, we present a practical and scalable method for learning physical and logical Pauli error rates of fault-tolerant computations that only uses the naturally available syndrome data. We show that the logical error rate of fault-tolerant computations can be learned using a number of samples that scales polynomially in the distance of the code. 
In contrast, direct methods that only use logical data require an exponential number of samples. 
We demonstrate the practicality of our methods by applying them to experimental data from the recent demonstration of fault-tolerant computation~\cite{bluvstein2024logical}.

To achieve this, we develop a new systematic framework for learning physically motivated Pauli errors in general subsystem codes, by further developing the ideas of \textcite{wagner2022pauli, wagner2023learning} which lay the groundwork for our method.
% We use a physically motivated Pauli error model, as in \cite{wagner2022pauli, wagner2023learning}.
% , which subsumes the inclusive model considered in other previous works.
% 
We then apply this framework to general, non-adaptive fault-tolerant Clifford circuits using the circuit-to-code mapping \cite{bacon2017sparse, gottesman2022opportunities, delfosse2023spacetime,pesah2025fault}. 
This mapping converts the problem of correcting errors in Clifford circuits into correcting errors in subsystem codes.

Our work provides a general framework for in-situ characterization of both physical and logical noise in fault-tolerant quantum computations, using syndrome data. 
These problems include fault-tolerant quantum memory experiments as studied in this work, but also magic-state distillation, and classically hard logical circuits such as Clifford operations augmented with magic-state inputs which we study in Part II \cite{XiaoWorkInProgress2026}.
Thus, it enables benchmarking and verifying circuits and algorithms using the natively available syndrome data.

% Moreover, our method can be applicable beyond the setting of non-adaptive Clifford circuits studied here. 
% \dom{In upcoming work, we show that}
% by incorporating the spacetime-code formalism, our approach enables benchmarking and verifying circuits and algorithms involving adaptive operations conditioned on prior logical measurement outcomes, encompassing classically hard logical circuits such as Clifford operations augmented with magic resources.
% \Dom{Should we not mention that we also do this in the upcoming work?}

\subsection{Contributions}

Our concrete contributions are as follows. For the noise model, we give a necessary and sufficient condition regarding the learnability of the total Pauli channel as well as individual Pauli channels. 
We prove for general stabilizer code/Clifford circuits that after grouping errors into syndrome classes, the total error rate of each class can be accurately learned with negligible second-order correction. 
% In this formulation, we find that our framework bridges the work of \textcite{wagner2022pauli, wagner2023learning} and the analytic solution suggested by \textcite{remm2025experimentally}. We are able to make an accuracy improvement over the analytic solution. 
% \Dom{What about the Wagner works?} \Xiao{they don't have a concrete learning algorithm}
Moreover, we show that under circuit-level Pauli noise, if the circuit is fault-tolerant with respect to the assumed noise model, then the logical error rates are learnable solely from the circuit’s syndrome data, enabling efficient estimation of the circuit’s logical fidelity.

We also analyze the sample complexity, under the local sparse noise model assumption, for learning both physical- and logical-level noise from syndrome data as a function of the circuit’s spacetime volume. For logical circuits corresponding to quantum low-density parity-check (qLDPC) codes, we prove that in the low physical-error-rate regime, the required sample complexity for physical error rate learning becomes constant. Finally, we show that in the below-threshold regime, learning the logical error rate from syndrome data achieves an exponential sample-complexity advantage over direct logical-level measurements.

We first demonstrate the efficacy and scalability of our noise-learning algorithm using numerical benchmarks on simulated data. We then demonstrate its practicality by applying it to the experimental data of Ref.~\cite{bluvstein2024logical} where a fault-tolerant protocol for logical GHZ state preparation was demonstrated. 
In particular, we show that the estimated logical error rates are consistent with independently obtained experimental results.

\subsection{Significance}

Our work studies and verifies the effectiveness of the physically motivated Pauli error model in the context of fault-tolerant logical circuits; taking advantage of readily available data from near-term experiments on fault-tolerance. Within this error model, we prove an exponential sampling advantage for learning the logical error rate from syndrome data in a regime of low error rates.
Finally, our work solves an important practical problem by providing a systematic framework for estimating logical and physical noise rates with improved accuracy using small amounts of syndrome data, without the need for independent logical measurements. 
In part II of this work, we show that these ideas can be extended to efficiently verifying and benchmarking classically hard circuits from syndrome data \cite{XiaoWorkInProgress2026}.
% \xiao{The formalism and results established here can be extended to a broader class of circuits, e.g., circuits with adaptive/feedforward operation. In part II, we explain how analogous results still hold. In particular, the learnability of logical error rate and its sample-complexity advantage results provide a method for efficiently verifying classical hard circuits.}

\subsection{Related work}
A large body of work has developed efficient protocols for benchmarking quantum devices and characterizing noise. Early scalable approaches, including randomized benchmarking \cite{emerson2005scalable}, interleaving benchmarking \cite{magesan2012efficient}, cycle benchmarking \cite{erhard2019characterizing} only estimate global fidelity parameters. More detailed noise-characterization methods have subsequently been developed, including Pauli channel estimation \cite{flammia2020efficient,harper2020efficient,harper2021fast}, cycle error reconstruction \cite{carignan2023error}, and averaged circuit eigenvalue sampling (ACES) \cite{flammia_averaged_2022}. These protocols aim to reconstruct richer descriptions of the underlying noise, often at the level of Pauli or cycle error rates. Cycle error reconstruction and ACES have also been applied in quantum error-correction settings \cite{fazio2026characterizing,hockings2025scalable,hockings2025improving}. However, in these applications, they still largely function as standalone benchmarking or noise characterization procedures and cannot be applied to syndrome data from fault-tolerant circuits.

Earlier works attempted to solve the physical noise learning task from syndrome data for specific codes and error detection/correction circuits \cite{fowler2014scalable, spitz2018adaptive,google2021exponential}, benchmarking circuits with a single parameter using detector data \cite{hesner2024using}, or for general Pauli noise but in a heuristic way \cite{huo2017learning}. More recently, a series of works \cite{takou2025estimating,remm2025experimentally,blume2025estimating} generalized the formalism of Ref.~\cite{spitz2018adaptive} to other codes, e.g., when the decoding graph can be a hypergraph. However, these works focused on a restricted Pauli error model commonly used in decoding graphs, often referred to as the \emph{detector error model}. This model assumes that each detector flip happens independently. 
This implies that the underlying Pauli errors triggering the detectors occur independently, a model known also as the \emph{inclusive error model} since multiple different Pauli errors can happen at the same time (see the discussion in Appendix E of \cite{chao2020optimization}). 
Arguably, a physically more natural model is to assign each noisy operation a local Pauli channel, where Pauli errors occur \emph{exclusively}, that is, only a single Pauli error can happen at any time. 
When translating the inclusive error model to the physically natural exclusive model, negative proabilities can appear. 
Recently, \cite{remm2025experimentally,blume2025estimating} provided analytical solutions to the general detector-error-model learning problem. 
In practice, however, purely analytical solutions can overfit when syndrome data are limited, leading to degraded learning precision.
Therefore, these results do not address how to find the right balance between model complexity and shot noise. 

In a parallel series of works, \textcite{wagner2021optimal, wagner2022pauli, wagner2023learning} rigorously investigated the sufficient conditions under which physical and logical error rates can be extracted from the syndrome of general stabilizer codes in the channel coding setting under a  general, local Pauli noise model.
However, the identifying conditions are not satisfied in many cases, and no efficient algorithms for estimating noise based on this framework have been developed for the general case. 
Our work builds upon the formalism developed in those works.

Finally, the above works focus on learning error rates from the syndromes produced during rounds of error correction. None of the prior literature focuses on the problem of learning physical- or logical-level Pauli noise in general fault-tolerant circuits encountered in an algorithm using syndrome data. As a result, a systematic, broadly applicable framework for exploiting syndrome information throughout a logical circuit or algorithm has been lacking prior to our work.

\subsection{Overview}

This paper is structured as follows. In \cref{section: model}, we introduce the Pauli noise model. 
In \cref{section: static code}, we apply the formalism to the static-code setting, where a single round of syndrome information is extracted. In this section, we present necessary and sufficient conditions for the learnability of physical Pauli noise and the sample complexity to achieve the learning task to a desired accuracy. 
In \cref{section: circuit level}, we apply these results to the problem of learning the noise in fault-tolerant circuits by using the mapping of those circuits to the spacetime code. We demonstrate the effectiveness of our method using synthetic and experimental data. 
Finally, in \cref{section: logical error rate and fidelity}, we consider the problem of learning logical error rates in static codes and fault-tolerant computations. We give learnability conditions, demonstrate the method, and show analytically and numerically that a polynomial number of samples is sufficient to learn the logical error rates in circuits that correspond to qLDPC spacetime codes under the local sparse Pauli error model. 
We conclude with a discussion. All proofs and details are deferred to the appendices. 

\tableofcontents

\section{Model}\label{section: model}
%\printinunitsof{in}\prntlen{\textwidth}
%\printinunitsof{in}\prntlen{\columnwidth}

We denote the Pauli group without phases over $n$ qubits as $\mathcal{P}_n=\{I, X, Y, Z\}^{\otimes n}$ and focus on learning a  Pauli error channel 
\begin{equation}
    \mathcal{N}(\rho)=\sum_{e\in \mathcal{P}_n} P(e)\,e\rho e^{\dagger},
\end{equation}
acting on an $n$-qubit state $\rho$. 
Here, $P(e)$ are the error rates of the Pauli errors $e\in\mathcal{P}_n$. 
In the most generic case, learning cannot be efficiently possible as there are $4^{n}-1$ independent parameters that need to be specified.
This is why in practice we focus on a \emph{local} error model in which high-weight errors are given by products of individual low-weight errors. 

To be more precise, we define $\Gamma$ as a set of supports such that each $\gamma\in\Gamma$ is a subset of qubits on which a Pauli channel $\mathcal{N}_{\gamma}$ acts nontrivially, i.e.,
\begin{equation}
    \mathcal{N}_{\gamma}(\rho) = \sum_{e\in \mathcal{P}_{n}|\text{supp}(e)\subset\gamma} P_{\gamma}(e) e\rho e^{\dagger}. 
\end{equation}
Here, $P_{\gamma}(e)$ are the error rates of the individual channels satisfying the local condition $P_{\gamma}(e)=0$ if $\text{supp}(e)\not\subset \gamma$. 
The total Pauli noise channel, now denoted as $\mathcal{N}_{\Gamma}$, is the composition of these individual Pauli error channels and acts as 
\begin{equation}
    \mathcal{N}_{\Gamma}(\rho)=\circ_{\gamma\in\Gamma}\mathcal{N}_{\gamma}(\rho).\label{equation: channel composition}
\end{equation}
Since all Pauli channels commute with each other, the order of the composition does not matter here.
An example of this Pauli noise model is illustrated in \cref{fig: conv_factor_graph}. 
Since in general $\Gamma$ could contain overlapping supports, the resulting probability distribution
\begin{equation}
    P(e)=\Asterisk_{\gamma \in \Gamma}P_{\gamma}(e),
    \label{equation: convolution factor}
\end{equation}
of the total Pauli channel $\mathcal{N}_{\Gamma}$ is the convolution over all individual error distributions $P_{\gamma}(e)$, denoted by $\Asterisk$. The convolution over the Pauli group is defined by $P_{\gamma_1}\ast P_{\gamma_2}(e)=\sum_{e'\in\mathcal{P}_{n}}P_{\gamma_1}(e')P_{\gamma_2}(e'e)$. 
The structure of the error rate is sometimes referred to as a ``convolutional factor graph"~\cite {mao2005factor, wagner2022pauli}. For analytical convenience, we make the following assumption about each individual channel.

\begin{assumption}\label{assumption: open set probability}
    Let a total Pauli error channel $\mathcal{N}_{\Gamma}$ be a composition of Pauli channels $\mathcal{N}_{\gamma}$ with support $\gamma\in\Gamma$. We assume for any individual channel $\mathcal{N}_{\gamma}$, we have $P_{\gamma}(I)>1/2$, and for any non-identity error $e$ in $\mathcal{N}_{\gamma}$ with $\text{supp}(e)\subseteq \gamma$, $P_{\gamma}(e)>0$.
\end{assumption}

\begin{figure}[t]
\centering
\hspace{-0.6cm}
\begin{subfigure}{0.2\textwidth}
    \phantomcaption
    \stackinset{l}{2pt}{t}{2pt}{\captiontext*}
    {\includegraphics{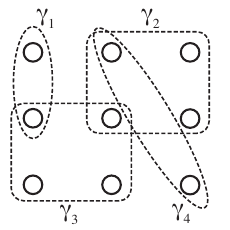}}
    \label{fig: conv_factor_graph}
\end{subfigure}
\hspace{0.5cm}
\begin{subfigure}{0.2\textwidth}
    \phantomcaption
    \stackinset{l}{2pt}{t}{2pt}{\captiontext*}
    {\includegraphics{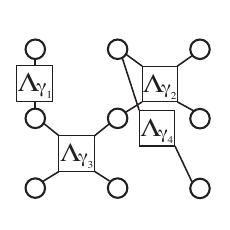}}
    \label{fig: factor_graph}
\end{subfigure}
    \caption{An example of the total Pauli error channel $\mathcal{N}_{\Gamma}$ as the composition of local Pauli channels $\mathcal{N}_{\gamma}$ (enclosed in dotted lines) acting on subsets $\gamma\in\Gamma$ of qubits (circles), where $\Gamma=\{\gamma_1,\gamma_2,\gamma_3,\gamma_4\}$. (a) The error rate $P$ is the convolution of the error rates $P_{\gamma}$ of all local error channels. (b) The factor graph of Pauli eigenvalues of the total channel $\mathcal{N}_{\Gamma}$ can be written as the product of eigenvalues $\Lambda_{\gamma}$ of the local channels after performing Walsh-Hadamard transformation of the error rates in (a).}
    \label{fig: conv_factor_graph_total}
\end{figure}

When considering sample complexity and classical computational cost, we focus on the case when individual channels are local. We emphasize that `local' here refers to the concept of $k$-locality rather than geometric locality. To be quantitative, we say that a Pauli channel $\mathcal{N}_{\Gamma}$ is a $(r_{\Gamma}, c_{\Gamma})$-Pauli channel if each local channel $\mathcal{N}_{\gamma}, \forall \gamma\in\Gamma$, acts on at most $r_{\Gamma}$ qubits and each qubit is only involved in $c_{\Gamma}$ channels. Moreover, we call $\mathcal{N}_{\Gamma}$ a \textit{local} \textit{sparse} Pauli channel if both $r_{\Gamma}$ and $c_{\Gamma}$ are $\mathcal{O}(1)$. In this setting, the noise channel $\mathcal{N}_{\Gamma}$ is fully described by at most $c_{\Gamma}4^{r_{\Gamma}}n=\mathcal{O}(n)$ parameters. In general, there are gauge degrees of freedom in the choice of local Pauli channels that give rise to a given total Pauli channel. We use two sets $\overline{\mu}$ and $\mu$ of irredundant quantities that parametrize the total channel $\mathcal{N}_{\Gamma}$ and the collection of local channels respectively, and show the gauge freedom of $\mu$ with respect to $\overline{\mu}$ in \cref{appendix: log u basis}. 

We use this model to describe circuit-level noise by using a formalism \cite{bacon2017sparse,gottesman2022opportunities,delfosse2023spacetime,pesah2025fault} that maps a Clifford circuit to a subsystem stabilizer code referred to as the spacetime code, see \cref{section: spacetime code mapping}. Qubits at each circuit layer are mapped to an additional layer of qubits in the spacetime code, and therefore each individual Pauli channel defined here should be thought of as acting on a region of the circuits in space and time. In the main text, we focus on the parametrization of the collection of local channels and the learnable quantities therein, instead of studying the total channel $\mathcal{N}_{\Gamma}$ (see \cref{subsection: syndrome class learnability}). This is because benchmarking is most naturally framed in terms of local noise processes. In the numerical part of this work, we only considered non-overlapping local channels when implementing the learning algorithm, in line with the standard local stochastic Pauli error model typically being used in the error correction and fault-tolerant literature \cite{dennis2002topological, fowler2009high, bravyi2024high}.

\section{Efficient Pauli Noise Learning on Static Code\label{section: static code}}
Our learning algorithm takes the framework of~\textcite{wagner2022pauli} as a starting point. 
They considered subsystem codes in the channel coding setting, where a single round of Pauli errors is applied to the error-correcting code followed by perfect syndrome measurement.
Their work \cite{wagner2022pauli} establishes sufficient conditions on the noise channels and code properties that guarantee the learnability of the total channel $\mathcal{N}_{\Gamma}$ from syndrome data. 
In the following subsection (\cref{sec:basic formalism}), we introduce this basic framework for learning Pauli noise. 
We then adapt this formalism and introduce our algorithm to estimate physical error rates of arbitrary stabilizer subsystem codes from syndrome data in the subsequent subsections. 
We refer to the noise learning problem in this channel coding setting as the \emph{static code} problem. We then show in \cref{section: circuit level,section: logical error rate and fidelity} how to learn \emph{circuit-level} Pauli noise in Clifford circuits by reducing it to a static code problem using the circuit-to-spacetime-code mapping. 

%\Dom{I think the distinction of ours from Wagner et al's results is unclear. Maybe we can state their key results as theorems attributed to them, and then contrast with ours?}\Xiao{}

\subsection{Basic formalism}
\label{sec:basic formalism}

We consider a system of $n$ qubits encoded in a subsystem code, which is a generalization of a stabilizer code. 
It is defined by the gauge group $\mathcal{G}$, a non-abelian subgroup of the Pauli group $\mathcal{P}_n$. 
The gauge group is the group of operators that preserve the code space and the encoded logical information. 
Considering general subsystem codes will be important when applying our formalism to circuit-level noise when using the spacetime code mapping in \cref{section: spacetime code mapping}, since in that case there are logically trivial errors in the circuit outside of the stabilizer group of the spacetime code.

For the purpose of learning physical error rates, it suffices to consider the stabilizer group
\begin{equation}
    \mathcal{S}=Z(\mathcal{G})=\{S\in\mathcal{G}|SG=GS, \forall G \in \mathcal{G}\},
\end{equation}
which is the center $Z(\mathcal G)$ of the gauge group $\mathcal{G}$ and coincides with $\mathcal{G}$ for stabilizer codes. 
An encoded state in the code, referred to as a codeword, is stabilized by the stabilizer group, i.e.,
\begin{equation}
    S\ket{\psi}=\ket{\psi},\forall S\in\mathcal{S}.
\end{equation}
We say a stabilizer subsystem code is a $(r_s, c_s)$-quantum low-density parity-check code (qLDPC) when there exists a set of stabilizer group generators such that each generator acts on at most $r_s$ qubits, and there are at most $c_s$ generators acting on each qubit.

The measurement outcomes of stabilizer operators $S \in \mathcal S$, known as the \emph{syndrome}, provide the information needed for decoding and error correction. To obtain a complete syndrome, that is, measurement outcomes for all $S\in\mathcal{S}$, a generating set of stabilizer operators needs to be measured.
Here, we consider the generic case in which only a subset of generators $\{M_1, M_2,...\}$ is measured. 
We denote the resulting measured subgroup as $\mathcal{M}=\langle M_1,M_2,... \rangle\leq \mathcal{S}$. A Pauli error $e\in\mathcal{P}_n$ is detectable if there exists $M\in\mathcal{M}$ that anticommutes with~$e$.

The learning algorithm makes use of the fact that the Pauli eigenvalue, defined as
\begin{equation}
    \Lambda(O)=\sum_{e\in\mathcal{P}_n}[[e, O]]P(e),
\end{equation}
of the channel for all Pauli operators $O\in\mathcal{P}_n$ is a Fourier transform (Walsh-Hadamard transform) of the probability distribution $P(e)$ over the Pauli group $\mathcal{P}_n$. Here we denote $[[A,B]]= 2^{-n}\Tr[ABA^{-1}B^{-1}]$ as the scalar commutator of two Pauli operators $A$ and $B$, which takes the value $1$ if $A$ and $B$ commute and $-1$ if they anticommute. 
Based on the local structure of the total error rate shown in \cref{equation: convolution factor} and the convolution theorem, the Pauli eigenvalue $\Lambda(O)$ of the total channel $\mathcal{N}_{\Gamma}$ has a structure of a factor graph (see \cref{fig: factor_graph})
\begin{align}
    \Lambda(O)&=\widehat{P}(O)=\prod_{\gamma\in \Gamma}\widehat P_{\gamma}(O)=\prod_{\gamma\in \Gamma}\Lambda_{\gamma}(O_\gamma).
    \label{equation: E factorization}
\end{align}
Here, $O_\gamma$ is the tensor product of components of $O \in \mathcal P_n$ restricted to the support $\gamma$, and identity operators on all qubits outside of $\gamma$. 

All the factors on the right-hand side of \cref{equation: E factorization} are contained within the set of all local Pauli eigenvalues, which provides a complete description of the unknown channel $\mathcal{N}_{\Gamma}$. 
Conversely, for any $M \in \mathcal{M}$, the left-hand side can be estimated from repeated error-correction experiments (comprising noise state preparation and syndrome measurement), under
the assumption that the same error channel $\mathcal{N}_{\Gamma}$ is applied each time the codeword is reinitialized. 
To see this, given the initial state $\rho$ as a codeword stabilized by $\mathcal{S}$, we have
\begin{align}
    \langle M\rangle_{\mathcal{N}(\rho)}&=\Tr{M\mathcal{N}(\rho)}\\
    &=\sum_eP(e)\Tr{M e\rho e^{\dagger}}\\
    &=\Lambda(M)\langle M\rangle_{\rho}\\
    &=\Lambda(M),
\end{align}
for any $M\in\mathcal{M}\leq\mathcal{S}$. 
In other words, the Pauli eigenvalue $\Lambda(M)$ is the expectation value of the syndrome measurement outcomes obtained from repeated measurements followed by reinitialization.

Given the estimated Pauli eigenvalues of all $M\in\mathcal{M}$ from the syndrome measurement data, one can attempt to recover the local Pauli eigenvalues $\Lambda_\gamma(M_\gamma)$ by solving for the factors on the right-hand side in \cref{equation: E factorization}, in order to reconstruct the overall Pauli channel $\mathcal{N}_{\Gamma}$. 
\textcite{wagner2022pauli} rigorously prove conditions of learnability of $\mathcal{N}_{\Gamma}$ from syndrome data, and we state their main result (adapted to our model in \cref{section: model}) here.
\begin{theorem}[\textcite{wagner2022pauli}]
     Let a Pauli error channel $\mathcal{N}_{\Gamma}$ be a composition of local Pauli channels with support $\gamma\in\Gamma$ as defined in \cref{equation: channel composition}. Furthermore, assume $P_{\gamma}(I)>\frac{1}{2}$ for all $\gamma$. 
     % We say $\gamma$ support error $e$ if $\mathrm{supp}(e)\subseteq \gamma$, and $e$ is detectable if $e=I$ or $\exists S\in \mathcal{S},[[e,S]]=-1$. 
     Then $\mathcal{N}_{\Gamma}$ is learnable from the syndrome statistics if any union $\gamma_1 \cup \gamma_2$ of $\gamma_1,\,\gamma_2\in\Gamma$ only supports detectable errors.
    \label{theorem: wagner}
\end{theorem}

Another way to phrase this condition is that the collection of local errors $\overline{\mathcal{E}}_{\Gamma}$ must have nontrivial and distinct syndromes. Since the product of two errors with the same syndrome must have a trivial syndrome, it suffices to check the product of pairwise local errors, thus the condition in \cref{theorem: wagner}. The notion of pure distance \cite{wagner2022pauli} is central to understanding the sufficient condition of learnability. It is defined to be the minimum weight of \emph{any} undetectable error
\begin{equation}
d_{\mathrm{pure}}=\min_{e\in N(\mathcal{S})\backslash I} \text{weight}(e), \label{equation: pure distance}
\end{equation}
where $N(\mathcal{S})$ is the normalizer of the stabilizer group. 
Recall that, in contrast, the distance is the minimum weight of an undetectable error which \emph{also} alters the logical information.
If the weight of the product of any two errors $e_1,e_2$ from the local channels is less than the pure distance, then $e_1e_2$ is detectable. Therefore, we see from \cref{theorem: wagner} that error channels with maximum support size up to $\lfloor \frac{d_{\mathrm{pure}}-1}{2} \rfloor$ are learnable from the syndrome expectations of a pure-distance $d_{\mathrm{pure}}$ code.

%Regarding the total Pauli channel $\mathcal{N}$, in \cref{appendix: Gauge degrees of freedom in local Pauli channels} we further analyze the gauge degrees of freedom in overlapping local channels and show an efficient algorithm in \cref{appendix: Marginal probability estimation} to estimate the marginal distribution of the total error rate $P(e)$ on a constant-size subset of qubits.

\subsection{Learning syndrome class total error rate for generic static codes\label{subsection: syndrome class learnability}}
%\Dom{Let's say clearly what we use from Wagner et al, and what is our own contribution.}\Xiao{}
In this section, we consider the general case when the learnability condition in \cref{theorem: wagner} may fail. We provide concrete, efficient algorithms that learn the sum of error rates for errors with the same syndrome and estimate physical error rates for any subsystem code, under extra constraints when necessary.

We first establish the notation used throughout this work. For each local Pauli channel defined in \cref{section: model} with support $\gamma\in\Gamma$, we denote the set of all non-identity Pauli errors of $\mathcal{N}_{\gamma}$ as $\mathcal{E}_{\gamma}$ (e.g., $\mathcal{E}_{\{2\}}=\{X_2, Y_2, Z_2\}$ for a single-qubit error channel on qubit $2$ and $X_2=I\otimes X\otimes I \otimes...$) and the set of all local errors of the total channel $\mathcal{N}_{\Gamma}$ as
\begin{equation}
\mathcal{E}_{\Gamma}=\bigcup_{\gamma\in\Gamma}\mathcal{E}_{\gamma}.
\end{equation}
When local channels overlap, there can be distinct local channels 
containing the same Pauli error. We distinguish elements that correspond to the same operator but from different channels, i.e., each element $e\in\mathcal{E}_{\Gamma}$ is also implicitly labeled by its associated local channel, denoted $\gamma_e$. We denote as $\overline{\mathcal{E}}_{\Gamma}$ the set that does not distinguish the same Pauli errors from different local channels. $\overline{\mathcal{E}}_{\Gamma}=\mathcal{E}_{\Gamma}$ if and only if the supports in $\Gamma$ do not overlap.

We define the vector of both individual and aggregated local error rates as $\vec{p}_{\gamma}=(P_{\gamma}(e)|e\in\mathcal{E}_{\gamma})$ and $\vec{p}=(P_{\gamma_e}(e)|e\in\mathcal{E}_{\Gamma})$. Similarly, we define the vector of individual and aggregated local Pauli eigenvalues as $\vec{\lambda}_{\gamma}=(\Lambda_{\gamma}(e)|e\in\mathcal{E}_{\gamma})$ and $\vec{\lambda}=(\Lambda_{\gamma_e}(e)|e\in\mathcal{E}_{\Gamma})$ respectively. We adopt a convention in which the entries of these vectors follow the tensor-product ordering of Pauli operators:
\begin{equation}
   (I, X, Y, Z)^{\otimes n}.\label{equation: order}
\end{equation}
For example, the vector $\vec{\lambda}_{\{2,5\}}$ has the following order: $\vec{\lambda}_{\{2,5\}}=(\Lambda(I_2X_5),...,\Lambda(Z_2Y_5),\Lambda(Z_2Z_5))$.

The collection of either local error rates or local Pauli eigenvalues over all local channels is a succinct description of $\mathcal{N}_{\Gamma}$. The two are related by the (inverse) Fourier transformation
\begin{equation}
    \vec{p}_{\gamma}=(W_{4^{|\gamma|}}^{-1})'\,(\vec{1}_{\gamma}-\vec{\lambda}_{\gamma}).\label{equation: local lambda to p}
\end{equation}
Here, $W_{4^k}$ is the Walsh-Hadamard matrix defined recursively as $W_{4^k}=W_{4}\otimes W_{4^{k-1}}$ where
\begin{equation}
    W_{4}=
\begin{pmatrix}
1 & 1 & 1 & 1 \\
1 & 1 & -1 & -1 \\
1 & -1 & 1 & -1 \\
1 & -1 & -1 & 1
\end{pmatrix}.
\end{equation}
$(W_{4^{|\gamma|}}^{-1})'$ refers to the submatrix of $W_{4^{|\gamma|}}^{-1}$ after eliminating the first row and the first column.
Assuming $P_{\gamma}(I)>\frac{1}{2}$ for all local support $\gamma$, we arrive at the matrix equation between expectation values of elements in the measured stabilizer group and the local channel Pauli eigenvalues:
\begin{equation}
    \log\vec{\Lambda}^{(\mathcal{M})}=A^{(\mathcal{M})}\log \vec{\lambda},
    \label{equation: log equation}
\end{equation}
 by taking the logarithm on both sides of \cref{equation: E factorization}. Here, $\vec{\Lambda}^{(\mathcal{M})}=(\Lambda(M))_{M\in\mathcal{M}}$ and the rows and columns of matrix $A^{(\mathcal{M})}$ are labeled by elements $M\in\mathcal{M}$ and $e\in\mathcal{E}_{\Gamma}$ respectively. The matrix elements indexed by row $M$ and column $e$, denoted $A^{(\mathcal{M})}[M,e]$, have the values
\begin{equation}
    A^{(\mathcal{M})}[M,e]=\begin{cases}
        1&\text{if }M_{\gamma_e}=e,\\
        0&\text{else.}
    \end{cases}
\end{equation}
$A^{(\mathcal{M})}$ is not invertible when the learnability condition in \cref{theorem: wagner} fails. For example, this happens when error-correcting codes contain weight-two stabilizer elements or logical operators, such as in the rotated surface code~\cite{horsman2012surface}. Spacetime codes constructed from most Clifford circuits typically contain low-weight gauge operators that also lead to a violation of this condition under simple circuit-level noise models, as will be discussed in \cref{section: circuit level}.

To understand what can be learned when the learnability condition in \cref{theorem: syndrome class learning} fails, we transform the set of local eigenvalues $\lambda$ into another set of parameters $\mu$ which we refer to as the transformed local eigenvalues. We show that $\mu$ approximates the local error rate, and any transformed local eigenvalue corresponding to a unique column of the resulting matrix equation can be identified.

We consider an inverse Walsh-Hadamard transform of $\log \vec{\lambda}$, where $\vec{\lambda}$ is the collection of all local channel Pauli eigenvalues. This results in an equivalent set (vector) of transformed eigenvalues $\mu$ ($\vec{\mu}$) describing all local channels:
\begin{equation}
    \log\vec{\mu}=-2\left(\bigoplus_{\gamma\in\Gamma}(W_{4^{|\gamma|}}^{-1})'\right)\log \vec{\lambda}.
\end{equation}
Since $\log \vec{\lambda}\approx \vec{\lambda}-\vec{1}$, from \cref{equation: local lambda to p} we have the approximation to the local error rates
\begin{equation}
    \log \vec{\mu}\approx -2\vec{p}.\label{equation: approximation of p from mu}
\end{equation}
The deviation is negligible when error rates are low, i.e., at the order of $10^{-2}$ or below, and remains insignificant compared to the shot noise arising from the finite sample size.

With the parametrization using the transformed Pauli eigenvalues, we arrive at a new matrix equation
\begin{equation}
\log\vec{\Lambda}^{(\mathcal{M})}=D^{(\mathcal{M})}\log \vec{\mu},\label{equation: log equation mu}
\end{equation}
where $D^{(\mathcal{M})}$ encodes the commutations between measured stabilizer elements $\mathcal{M}$ and the local errors $\mathcal{E}_{\Gamma}$ and is referred to as the syndrome matrix. Specifically, the entry at row $M\in\mathcal{M}$ and column $e\in\mathcal{E}_{\Gamma}$ is
\begin{equation}
D^{(\mathcal{M})}[M,e]=\langle M,e\rangle,\label{equation: D entries}
\end{equation}
and $\langle e,f \rangle$ equals 0 when $e$ and $f$ commute and $1$ otherwise. We systematically derive the syndrome matrix $D$ in \cref{appendix: log u basis} and show in \cref{appendix: A rank general} that distinct columns of $D$ will be linearly independent. Note that if two errors $e$ and $e'$ are in the same syndrome class, then the columns of $e$ and $e'$ of the matrix $D$ will be identical. Therefore, it is natural to partition the error set $\mathcal{E}_{\Gamma}$ into equivalence classes, which we refer to as ``syndrome error classes'' or simply ``syndrome classes'' and denote the set of classes as $\mathcal{C}$. Errors are in the same class $C\in\mathcal{C}$ if and only if they have the same syndrome. Furthermore, we distinguish the class $C_0$ of trivial syndrome from the set  $\mathcal{C}^{*}=\mathcal{C}\backslash \{C_0\}$ of classes of nontrivial syndromes. As a result, we denote as $\nu_C$ the learnable quantity for each syndrome class $C\in\mathcal{C}^*$, and refer to them as learnable Pauli eigenvalues:
\begin{align}
    \log\nu_C=\sum_{e\in C}\log\mu_{e}\approx -2\sum_{e\in C} p_{e}=-2P_C,\label{equation: learnables}
\end{align}
where we define $P_C=\sum_{e\in C} p_{e}$ as the syndrome class error rate for the class $C$. The learnable Pauli eigenvalues can be obtained by solving the equation
\begin{align}
    \log\vec{\Lambda}^{(\mathcal{M})}=D'^{(\mathcal{M})}\log \vec{\nu}.\label{equation: D' matrix main}
\end{align}
Here, $D'$ is the full-rank matrix eliminating the redundant columns of $D$ and $\vec{\nu}=(\nu_C)_{C\in\mathcal{C}^*}$. Given a syndrome class $C$, we denote $\Gamma_C=\{\gamma\in\Gamma|\mathcal{E}_{\gamma}\cap C\neq\emptyset\}$ as the collection of local channels that contribute to the error class $C$. We then denote $\vec{p}^{(\Gamma_C)}=(P_{\gamma_e}(e)|e\in\mathcal{E}_{\Gamma_C})$ as their error rates grouped into a vector. We have the following results.
%In the setting of a local sparse Pauli channel $\mathcal{N}_{\Gamma}$ on a qLDPC code, the syndrome classes $C$, the associated set of channels $\Gamma_C$, and the set of errors $\mathcal{E}_{\Gamma_C}$ all have constant size.
\begin{theorem}[Learnability conditions]
    Let a total Pauli error channel $\mathcal{N}_{\Gamma}$ be a composition of Pauli channels $\mathcal{N}_{\gamma}$ with support $\gamma\in\Gamma$. Under \cref{assumption: open set probability}, using the syndrome data we have:
    \begin{enumerate}
        \item The total channel $\mathcal{N}_{\Gamma}$ is learnable if and only if every error $e\in\overline{\mathcal{E}}_{\Gamma}$ has a unique and nontrivial syndrome.
        \item Each individual channel $\mathcal{N}_{\gamma}$ is learnable if and only if every $e\in\mathcal{E}_{\gamma}$ produces a syndrome that is nontrivial and unique among the syndromes of all errors in $\mathcal{E}_{\Gamma}$.
        \item For all nontrivial syndrome error classes $C\in\mathcal{C}^*$, the sum of syndrome error class error rates $P_C=\sum_{e\in C} p_{e}$ can be learned up to an $\mathcal{O}(\|\vec{p}^{(\Gamma_C)}\|^2)$ deviation.
    \end{enumerate}
    \label{theorem: syndrome class learning}
\end{theorem}
See the proof in \cref{appendix: physical proof}.

Note that the matrix $A^{(\mathcal{M})}$ and $D^{(\mathcal{M})}$ have exponentially ($|\mathcal{M}|=2^{\Theta(n)}$) many rows while there are only polynomial many (transformed) local Pauli eigenvalues $\lambda$ ($\mu$). We prove in \cref{appendix: A rank general} that the number of independent equations is exactly the number of nontrivial syndromes $|\mathcal{C}^*|$ and design an efficient algorithm (\cref{algorithm: reduced stabilizer group}) that finds $|\mathcal{C}^*|$ elements in $\mathcal{M}$ from the product of correlated generators of $\mathcal{M}$. 
We denote the reduced subset of $\mathcal{M}$ as $\mathcal{M}'$. The algorithm ensures that $\mathcal{M}'$ will generate $|\mathcal{C}^*|$ independent equations (see \cref{appendix: minimal equations}):
\begin{equation}
    \log\vec{\Lambda}^{(\mathcal{M}')}=A^{(\mathcal{M}')}\log \vec{\lambda},
    \label{equation: log equation reduced}
\end{equation}
or
\begin{equation}
    \log\vec{\Lambda}^{(\mathcal{M'})}=D'^{(\mathcal{M'})}\log \vec{\nu}.\label{equation: D matrix reduced}
\end{equation}

We conclude this subsection by noting a related quantity for each nontrivial syndrome class $C$, i.e., the detector error rate $P^{\mathrm{detect}}_C$ for a syndrome class $C$ defined as the probability of an odd number of errors in $C$ occurring. As a special case, if one assumes every single Pauli error is independent, then we can exactly learn the detector error rate as  (see \cref{appendix: physical proof})
\begin{equation}
    P^{\mathrm{detect}}_C=\frac{1-\nu_C}{2},\label{equation: detector error rate}
\end{equation}
where $\nu_C$ is the learnable Pauli eigenvalues from \cref{equation: D matrix reduced}.
\begin{figure*}
    \centering
    \includegraphics[width=0.98\linewidth]{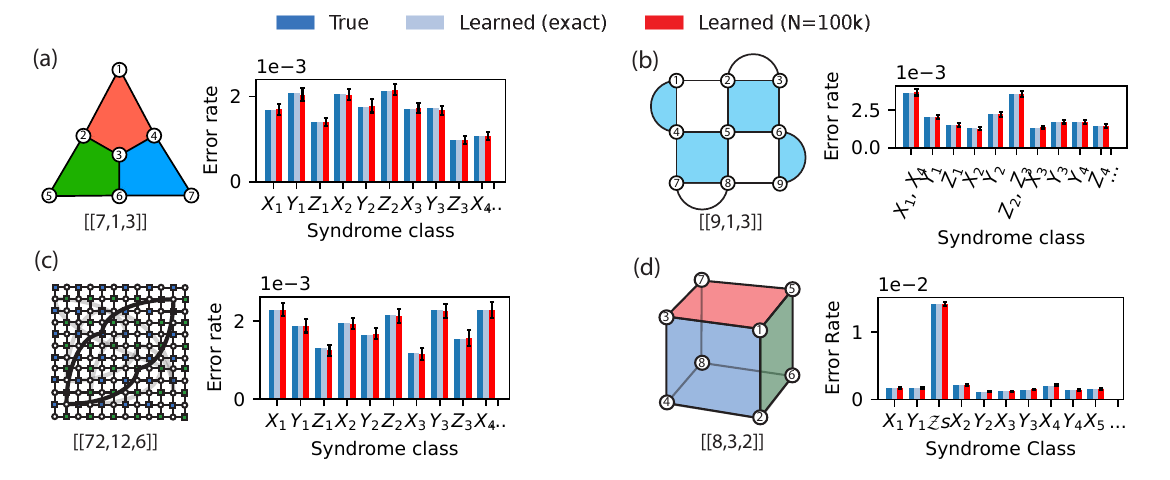}
    \caption{Simulation result on various codes of learning syndrome class total error rates (true error rates in blue) using exact syndrome expectation values (light blue) and expectation values (red) from 100k samples. Here, single-qubit Pauli error rates are randomly generated around the mean. Codes in (a) seven-qubit Steane code and (c) the $[[72,12,6]]$ bivariate bicycle codes with polynomial $A(x,y)=x^3+y+y^2$,\,$B(x,y)=y^3+x+x^2$, have $d_{\mathrm{pure}}\geq 3$, therefore all error rates are learnable. (b) The $3 \times 3$ rotated surface code has weight-2 $XX$ (blue) and $ZZ$ (white) generators on the boundary. (d) The $[[8,3,2]]$ 3D color code have $Z$ stabilizers on the faces and a single $X$ stabilizer generator $X^{\otimes 8}$. Therefore, all $Z$ errors, denoted as $\mathcal{Z}s$, are indistinguishable. Both codes in (b) and (d) have $d_{\mathrm{pure}}=2$, and the sum of error rates in each syndrome error class can be learned.}
    \label{fig: static code}
\end{figure*}
\subsection{Physical error rate estimation\label{section: physical error rate estimation}}
We can further estimate all individual error rates by incorporating a prior belief of the relative strength $r(e)$ of errors $e$ with a non-unique syndrome. For example, these could be obtained from offline benchmarking before running the logical circuit, or in the absence of any prior information, they could be simply uniform constraints  $r(e)=1,\forall e\in\mathcal{E}_{\Gamma}$.
We express this additional information as linear constraints on the error rates:
\begin{equation}
    B\vec{p}=\vec{0}.
    \label{equation: extra constraints}
\end{equation}
More specifically, for each nontrivial syndrome error class $C_i\in \mathcal{C}^{*}$ for $i \in [|\mathcal C^*|]$, we pick a representative error $e_{C_i} \in C$, and then add pairwise constraints between the error rate $p_{e_{C_i}}$ and all other rates $p_{e'}$ for $e'\in {C_i}\backslash e_{C_i}$. Each constraint between $e_{C_i}$ and $e'$ thus corresponds to a row of $B$ as
\begin{equation}
    B_{ij}=
    \begin{cases}
    1 , & \text{if } e_{j}=e_{C_i}, \\
    -\frac{r(e_{C_i})}{r(e_j)} , & \text{if } e_{j}=e', \\
    0,  & \text{else}.
\end{cases}
\label{equation: extra_nontrivial}
\end{equation}
We use the abbreviation $e_j \coloneqq (\vec{\mathcal{E}}_{\Gamma})_{j}$ where $\vec{\mathcal{E}}_{\Gamma}$ is the tuple of all local errors following the order given in \cref{equation: order}. On the other hand, each error $e_{C_0}$ in the trivial syndrome error class $C_0$ is independent of the syndrome data and we could declare them as unlearnable. Alternatively, we could estimate its error rates as the weighted average over the error rate of the set of detectable errors $\mathcal{E}^{*}_{\Gamma}:=\mathcal{E}_{\Gamma}\backslash C_0$ according to the prior belief $r(e)$:
\begin{equation}
    B_{ij}=
    \begin{cases}
    1, & \text{if } e_{j}=e_{C_0}, \\
    \frac{-r(e_{C_0})}{|\mathcal{E}^{*}_{\Gamma}|\,r(e_{j})} , & \text{if } e_{j}\in\mathcal{E}^{*}_{\Gamma}, \\
    0,  & \text{else}.
\end{cases}
\label{equation: extra_trivial}
\end{equation}
This adds $|C_0|$ rows to $B$ in addition to the $\sum_{C\in \mathcal{C}^*} (|C|-1)$ rows from \cref{equation: extra_nontrivial}.

For estimating general Pauli channels, we resort to the system of non-linear equations \cref{equation: log equation reduced,equation: D matrix reduced,equation: extra constraints}. In practice, the empirical expectation values $\tilde{\Lambda}(M)$ do not lie in the image of $A^{(\mathcal{M})}$ due to shot noise, and we implement least square optimization to solve the two sets of equations:
\begin{equation}
   \begin{aligned}
    &\underset{\vec{p}}{\text{minimize}}\quad\|\log\vec{\tilde{\Lambda}}^{(\mathcal{M}')}-A^{(\mathcal{M}')}\log \vec{\lambda}\|^2+\|B\vec{p}\|^2,\\
    &\text{subject to}\quad \vec{p}\geq 0.
\end{aligned}\label{equation: optimization}
\end{equation}
Alternatively, optimization can be constructed using the transformed eigenvalues $\vec{\mu}$. For low error rates or assuming the independent errors assumption (\cref{equation: detector error rate}), one could reduce the number of variables and solve a simple system of linear equations \cref{equation: D matrix reduced} to estimate syndrome class error rates, i.e., $|\mathcal{C}^{\ast}|$ parameters. Physical error rate can be estimated using \cref{equation: extra constraints} thereafter. For noisy empirical syndrome expectation values, a straightforward pseudoinverse of $D'^{(\mathcal{M}')}$ will give the least-squares estimate of the error rates, but may result in negative error rates for large shot noise. Alternatively, we solve the optimization:
\begin{equation}
   \begin{aligned}
    &\underset{\vec{p}}{\text{minimize}}\quad\|\log\vec{\tilde{\Lambda}}^{(\mathcal{M}')}-D'^{(\mathcal{M}')}\log \vec{\nu}\|^2,\\
    &\text{subject to}\quad \log \vec{\nu}\leq 0.
\end{aligned}\label{equation: optimization D'}
\end{equation}
Most numerical results in this work are generated via \cref{equation: optimization}, but \cref{equation: optimization D'} gives comparable results in the practically relevant error rate regime.

\subsection{Simulation on static code}
We simulate single qubit Pauli noise $\Gamma=\{\gamma_1=\{1\},\gamma_2=\{2\},...,\gamma_n=\{n\}\}$ on stabilizer codes of $n$ qubits and solve the optimization problem in \cref{equation: optimization}. Different codes with various parameters are chosen to show the generality of the algorithm. The simulated error rate of each single Pauli component in $\mathcal{E}_{\Gamma}$ follows a Gaussian distribution with mean $5\times 10^{-3}$ and standard deviation $10^{-3}$.

We begin by simulating the $[[7,1,3]]$ Steane code \cite{steane1996error}, a CSS code with both $X$ and $Z$ stabilizer generators on every plaquette shown in \cref{fig: static code}(a). The code has pure distance $d_{\mathrm{pure}}=3$ as defined in \cref{equation: pure distance} since it has stabilizers with minimum weight 4 and distance 3. This implies that all the single-qubit errors are learnable from the syndrome data. Therefore it has a empty trivial syndrome error class $C_0$, and all nontrivial classes being singlets: $\mathcal{C}^*=\{\{X_1\},...,\{\{Z_7\}\}$.

We first test the Pauli noise learning based on the exact syndrome expectation values for $M\in\mathcal{M}'$ from the simulated true noise $\vec{p}^{\,\text{true}}$:
\begin{align}
    \Lambda(M)=\prod_{i\in\text{supp}(M)}\Lambda_{\{i\}}(M_i)=\prod_{i\in\text{supp}(M)}(W_4\,\vec{p}^{\,\text{true}}_{\{i\}})_{M_i}
\end{align}
Next, we benchmark the learning algorithm with empirical syndrome expectation values $\tilde{\Lambda}(M)$ obtained from sampling syndromes $N=100,000$ times and repeat the learning 30 times with resampled noise rates. 
As shown in \cref{fig: static code}(a), there is a perfect agreement between the simulated true error rates and the learned error rates based on exact syndrome expectation values and good agreement with the learned error rate based on sampled expectation values.

We perform a similar test on the rotated surface code \cite{fowler2012surface}, a family of topological codes, defined on a $d \times d$ square lattice with parameters $[[d^2,1,d]]$. 
Note that its pure distance $d_{\mathrm{pure}}$ is always $2$ for any distance $d$ due to the stabilizer generators on the weight-2 plaquettes on the boundary.
Every single-qubit error is detectable so $C_0$ is empty, but there are indistinguishable errors in some of the nontrivial syndrome error classes.
In \cref{fig: static code}(b), we show our learned error rates on the $3 \times 3 $ rotated surface codes and compare with the actual true error rate we simulated. The classes of indistinguishable errors are $\{X_1, X_4\},\{Z_2, Z_3\},\{X_6, X_9\},\{Z_7, Z_8\}$. Theorem \ref{theorem: syndrome class learning} implies that the sum of errors of each syndrome class is learnable. With the exact syndrome expectation value, we observe that all syndrome class error rates are learned accurately. With 100,000 samples, the learned syndrome class error rates closely match the true error rates.

\cref{fig: static code}(c) and (d) show the simulation on the $[[72,12,6]]$ bivariate bicycle code, where every single qubit Pauli channel is learnable, and on the $[[8,3,2]]$ 3D color codes where the sum of all $Z$ errors and every other Pauli errors, is learnable.

\subsection{Sample complexity for physical error rate learning\label{subsection: physical sample complexity main}}

In this section, we study the number of samples $N$ required to learn the syndrome class error rates $P_C$, which are related to the learnable Pauli eigenvalues by $P_C\approx-\frac{\log \nu_C}{2}$. We focus on low error rates for which $-\log\nu_C\approx 1-\nu_C$. The precision in $\nu$ depends on the subset of the stabilizer elements chosen as well as how one solves \cref{equation: D matrix reduced} when we have an overdetermined system of linear equations. A recursive analytic solution is proposed in Ref.~\cite{remm2025experimentally} for learning the detector error rate $P^{\mathrm{detect}}_C$ in \cref{equation: detector error rate} in an error model where every local Pauli error is independent. We show in \cref{subsection: physical sample complexity} that this method can be generalized to solve the linear equations \cref{equation: D matrix reduced} for our more general error model. Furthermore, we show in \cref{theorem: syndrome class learning sample complexity} that the sample size required to achieve a constant additive precision $\epsilon$ of estimating syndrome class error rate can be shown to be independent of the system size for qLDPC codes and local sparse Pauli channel.
\begin{figure}[t]
\centering
\begin{subfigure}{0.95\columnwidth}
    \phantomcaption
    \hspace{-0.4cm}\stackinset{l}{-2pt}{t}{-3pt}{\captiontext*}
    {\includegraphics[width=\textwidth]{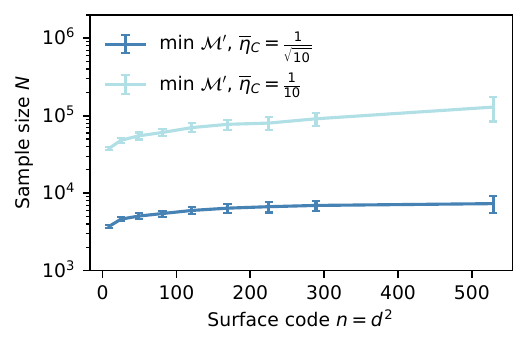}}
    \label{fig: surface code sample size}
\end{subfigure}\\
\vspace{0.5cm}
\begin{subfigure}{0.95\columnwidth}
    \phantomcaption
    \hspace{-0.3cm}\stackinset{l}{-2pt}{t}{-7pt}{\captiontext*}
    {\includegraphics[width=\textwidth]{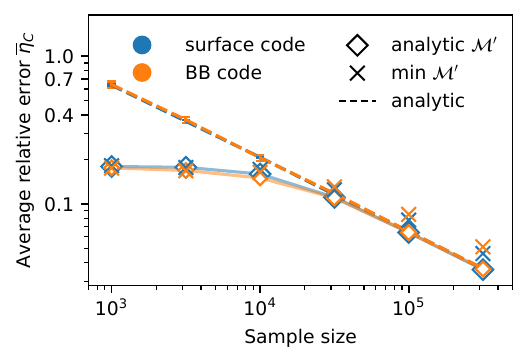}}
    \label{fig: precision of two examples}
\end{subfigure}
 \caption{Numerical results of relative error of learned single-qubit error rates vs sample size. The true error rates are sampled from $p\sim\mathcal{N}(\mu,\sigma^2)$, where $\mu=5/3\times10^{-3},\sigma=1/3\times10^{-3}$.(a) Estimated sample size needed for surface code with various code distances using the optimization \cref{equation: optimization} with uniform constraints $B$ within each syndrome class. This does not involve avoiding overfitting, and only the minimal stabilizer subset is used. Sample size is almost constant with respect to code size $n$. (b) Result of optimization while avoiding overfitting. We test the relative error of learning for surface code $d=23$ (blue) and the $[[72,12,6]]$ bivariate bicycle codes (orange). With the same stabilizer subset $\mathcal{M}'$, the optimization (diamonds) outperforms the analytical solution (dotted lines) in \cref{equation: recursive solving} at low sample sizes ($N\leq 30\,\mathrm{k}$) and is comparable to it at larger sample sizes. Using only a minimal stabilizer subset gives similar performance (crosses) at low sample size. The result is the average of 20 sets of error rates from the normal distribution $\mathcal{N}$.}
\label{fig: sample_complexity}
\end{figure}

Given a set of measured stabilizer generators $\{M_i\}^m_i$ for the measured stabilizer group $\mathcal{M}=\langle M_1,\dots, M_m\rangle$, we denote the indices of the stabilizer generators violated for each syndrome class $C\in\mathcal{C}$ as $J(C)\subseteq [m]$. We define $X_{C}=\log\nu_{C}$ and for any $S\subseteq [m]$, let $Y_S=\log\Lambda(\prod_{i\in S}M_i)$. The recursive solution in \cite{remm2025experimentally} for any syndrome class $C\in\mathcal{C}$ can be written as 
\begin{equation}
    X_C=-\sum_{\substack{C'\in\mathcal{C}\\J(C')\supset J(C)}}X_{C'}+\frac{1}{2^{(|J(C)|-1)}}\sum_{S\subseteq J(C)}(-1)^{(|S|-1)}Y_S.\label{equation: recursive solving}
\end{equation}
We show in \cref{appendix: sample complexity} that this solves \cref{equation: D matrix reduced} using Gaussian elimination. The subset $\mathcal{M}'$ of stabilizer elements involved here is the union of elements in the local stabilizer group, i.e., $\mathcal{M}'=\bigcup_{C\in\mathcal{C}}\{\prod_{i\in J'}M_i|J'\subseteq J(C)\}$. For a $(r_s,c_s)$-qLDPC code and a local sparse $(r_{\Gamma}, c_{\Gamma})$-total Pauli error channel $\mathcal{N}_{\Gamma}$, we have $|J(C)|\leq r_{\Gamma}c_s$ and $|\{C'\in\mathcal{C}|J(C')\supset J(C)\}|\leq r_{s}c_{\Gamma}$ for any syndrome class $C$. Each variable $X_C$, therefore, depends on a bounded number of $Y_S$'s, and each involved $Y_S$ has a bounded additive error from shot noise due to all elements in $\mathcal{M}'$ being constant-weight. With this observation, we show the sample complexity for learning syndrome class error rates below:

\begin{theorem}[Sample complexity for syndrome class learning\label{theorem: syndrome class learning sample complexity}]
    Consider a local sparse Pauli channel $\mathcal{N}_{\Gamma}$ on an $n$-qubit qLDPC code. Given $\epsilon>0$, there exists a constant $C_c>0$ such that if $\|\vec{p}\|_{\infty}^2<\epsilon/(2C_c)$, the sample size $N$ required to estimate the syndrome error class $P_C, \forall C\in\mathcal{C}^*$ to additive precision $\epsilon$ using \cref{equation: recursive solving} is $\tilde{\mathcal{O}}(\epsilon^{-2})$, independent of the system size $n$.
\end{theorem}
See the details of the proof in \cref{subsection: physical sample complexity} using Hoeffding's inequality. Note that applying Chernoff bound to the Bernoulli distribution (syndrome bit flip) gives a more precise scaling $N=\tilde{\mathcal{O}}(p_{\mathrm{max}}/\epsilon^2)$, where $p_{\mathrm{max}}=\|\vec p\|_{\infty}$.

\begin{figure*}[t]
\centering
    \includegraphics[width=0.95\textwidth]{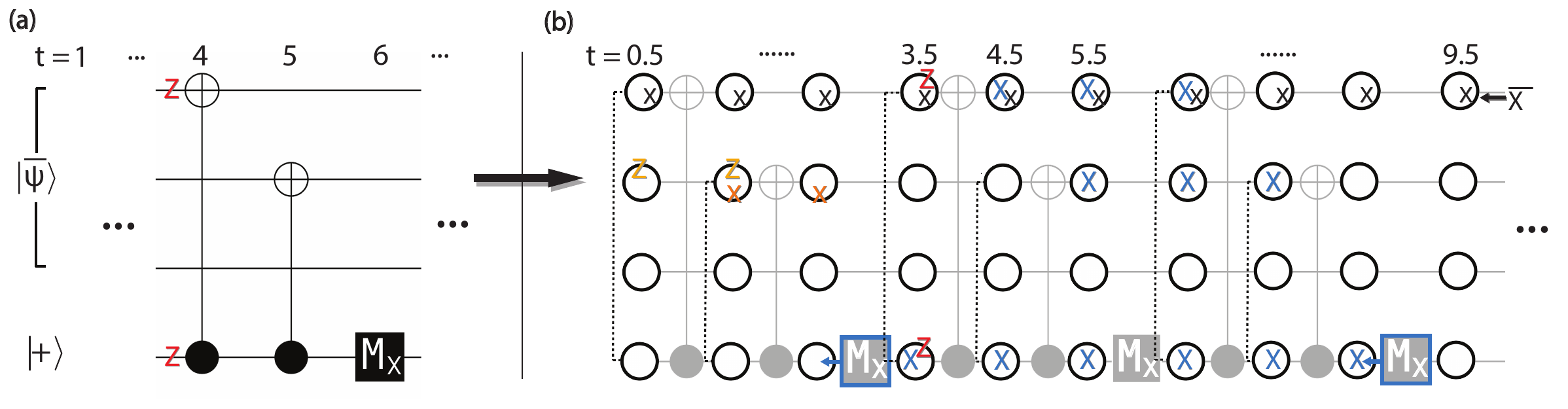}
    \caption{Circuit-to-code mapping illustrated using a syndrome extraction circuit for the 3-qubit repetition code. (a) $r$ rounds of 3-layer syndrome extraction circuits that measure $X_1X_2$ for the repetition code with stabilizer group $\mathcal{S}^{(\text{ini})}=\langle X_1X_2,X_2X_3\rangle$. The last qubit is the ancilla qubit for syndrome readout. (b) The spacetime code with $4\times(3r+1)$ qubits mapped from the circuit in (a). Pauli operator in blue is a measured stabilizer generator of the spacetime code due to the parity constraints on the first and third measurement results (blue squares). The stabilizer elements are local within the associated measurements in the temporal direction. One logical representative $\overline{X}^{\text{(ST)}}$ as the back-cumulant of $\overline{X}$ at the last layer is shown in black. It will be a measured spacetime logical operator if the first qubit is eventually measured. Correlated $ZZ$ error on the CNOT gate shown in (a) is mapped to $ZZ$ errors at $t=3.5$ (in red). Its syndrome and logical effect are given by the commutation with the spacetime stabilizer and logical operator respectively. The gauge group $\mathcal{G}^{(\text{ST})}$ contains errors with trivial syndrome and logical effect. $\mathcal{G}^{(\text{ST})}$ is non-abelian in general (see a pair of errors in light and dark orange).}
    \label{fig: spacetime code mapping}
\end{figure*}

We test the quality of learning as a function of sample size and system size using both the minimal $\mathcal{M}'$ syndrome data, i.e., the minimal subset $\mathcal{M}'$ with $|\mathcal{M}'|=\rank \mathcal{M}$ from \cref{algorithm: reduced stabilizer group}, and analytical $\mathcal{M}'$, i.e., the subset $\mathcal{M}'$ involved in the analytical solution \cref{equation: recursive solving}. We measure the quality of learning using the average relative error:
\begin{align}
    \overline{\eta}_{C}=\frac{1}{|\mathcal{C}^*|}\sum_{C\in\mathcal{C}^*}\frac{\epsilon_C}{P_C},
\end{align} 
which compares the learned and true syndrome-class error rates over all nontrivial syndrome classes $\mathcal{C}^*$. 

First, we test on the rotated surface code with system sizes $n=3^2,5^2,\dots, 23^2$ and sample sizes from 1000 to 6,000,000. Each qubit is subject to a single-qubit Pauli channel, and the error rate $p\sim\mathcal{N}(\mu,\sigma^2)$ of each Pauli component is sampled from a Gaussian distribution with mean $\mu=5/3\times10^{-3}$ and standard deviation $\sigma=1/3\times10^{-3}$. We clip negative values to $10^{-5}$. We tested the optimization method based on minimal $\mathcal{M}'$ in \cref{equation: optimization} with uniform constraints $B$ within each syndrome class. We fit a linear model of $\overline{\eta}_C$ as a function of sample size on the log scale and estimate the sample size needed to achieve a fixed precision. As shown in \cref{fig: surface code sample size}, the optimization using a minimal stabilizer subset requires nearly constant sample size.

We further compare the optimization method in \cref{equation: optimization D'} with the analytic solution \cref{equation: recursive solving} and observed two main factors influencing the precision. On the one hand, the minimal stabilizer subset $\mathcal{M}'$ used in \cref{equation: optimization D'} contains the least noisy syndromes due to their minimal weight, but also lacks redundancy. On the other hand, we observe that at smaller sample sizes, the system of equations \eqref{equation: recursive solving} tends to overfit, i.e., learning the shot noise rather than the actual noise. Therefore, the gradient-descent-based optimization on linear equations provides two advantages. On the one hand, one can tune the number of stabilizers used as long as the equation satisfies the full-rank condition. Moreover, one can in principle terminate the optimization before it starts to overfit. We adopt the simple approach of calibrating the termination point by simulating error rates drawn from the same distribution and recording the average number of iterations after which the average relative error starts to increase. 

We test the optimization \cref{equation: optimization D'} on a distance-$23$ rotated surface code and the $[[72,12,6]]$ bivariate bicycle codes. We first determine a terminal step number from one fixed error rate and apply this same number of optimization steps to learn 20 other sets of error rates from the same distribution $\mathcal{N}(\mu,\sigma^2)$. In \cref{fig: precision of two examples}, we observe that below $N=30\,\mathrm{k}$, optimization significantly outperforms the analytic solution, with the two agreeing once $N$ becomes large. Moreover, we observe that by optimizing while minimizing overfitting, using the minimal $\mathcal{M}'$ also has comparable performance in this low sample size regime compared with using the larger set from the analytic solution.

\section{\label{section: circuit level}Circuit-level noise learning}

To learn circuit-level Pauli noise for Clifford circuits, we utilize the circuit-to-code mapping, a formalism of converting the problem of correcting circuit-level Pauli errors into the problem of correcting Pauli errors on the corresponding \emph{spacetime code}~\cite{bacon2017sparse, gottesman2022opportunities, delfosse2023spacetime,pesah2025fault}. As we focus on non-adaptive fault-tolerant Clifford circuits here, the method introduced in \cref{section: static code} for static codes with a single round of syndrome extraction can then be applied to the spacetime code.

We begin by reviewing the spacetime code formalism and identifying the correct noise model on the spacetime code from the circuit-level error model in \cref{section: spacetime code mapping}. We then demonstrate our method using simulations of circuit-level noise in \cref{section: circuit level}, and apply it to experimental data in \cref{section: physical_learning_GHZ}.

\subsection{Mapping to the spacetime code\label{section: spacetime code mapping}}

We first introduce the setup of the spacetime code. Given a Clifford circuit with $n$ qubits and $T$ time steps, the spacetime code is a subsystem stabilizer code of $n^{\text{(ST)}}=n(T+1)$ qubits. We label each qubit by its spacetime location $(q,t)$ where $q$ denotes the qubit (i.e., the spatial label in the circuit) and $t$ denotes the time-coordinate index. The first and last layer of spacetime qubits are placed before the first gate layer and after the last gate layer respectively. The temporal labels $t$ of the spacetime qubits are $0.5, 1.5, ...,T+0.5$, where the $0.5$ shift here emphasizes that the spacetime qubits are placed before and after the gates. 
We illustrate the spacetime mapping for a simple 4-qubit Clifford circuit that extracts a syndrome of the 3-qubit repetition code in \cref{fig: spacetime code mapping}.

We now introduce the construction of the measured stabilizer and the logical group in the spacetime code, as these are relevant to the noise-learning algorithm. One way to capture the effect of circuit Pauli errors on syndrome and logical measurements is by checking the commutation relations of the circuit errors that are forward propagated to the syndrome and logical measurements at particular spacetime locations.
%\Dom{Why don't we want to do this?}
A more natural way from the perspective of the noise model, however, is to capture these commutation relations in the structure of the spacetime code, while errors stay local in their spacetime location, just like in the static code case. This is possible because the same commutation holds between the errors localized at their spacetime locations and the backward-propagated measurements to the relevant layer. 

Formally, the spacetime stabilizer and logical operators are objects referred to as \emph{back-cumulant} \cite{delfosse2023spacetime} (or \emph{spackle} \cite{bacon2017sparse, gottesman2022opportunities}) of some Pauli operator $O\in\mathcal{P}_{n(T+1)}$ and we denote it as $\overleftarrow{O}$.
For $t\in\{1,...,T+1\}$, the components $\overleftarrow O_{t} \in \mathcal P_n$ of $\overleftarrow{O} \in \mathcal P_{n(T+1)}$  
 at the layer $t$
are given by
\begin{align}(\overleftarrow{O})_{t-0.5} = \prod_{t'=T+1}^{t}U^{\dagger}_{t,t'}O_{t'-0.5}U_{t,t'}.
\end{align}
Here, $U_{t,t'} = \prod_{s = t}^{t'-1}U_t$ is a product of the gates $U_s$ from layer $t$ to $t'-1$ and we  let $U_{t,t'}=I$ if $t'\leq t$. 

%\Dom{Maybe in the following paragraphs on stabilizers, logicals and errors you can point expliticly to the figure again}
We begin by constructing the spacetime stabilizer operators. 
We assume the input to the Clifford circuit, up to a possible flip of the syndrome values, is a codeword of a stabilizer code with stabilizer group $\mathcal{S}^{(\text{ini})}$, i.e., a state $\ket \psi$ satisfying $S^{(\text{ini})}\ket{\psi}=\pm\ket{\psi}, \forall S^{(\text{ini})}\in\mathcal{S}^{(\text{ini})}$. We fictitiously add a layer of perfect syndrome measurement $S^{(\text{ini})}\in\mathcal{S}^{(\text{ini})}$ at $t=0$ to record the initial syndrome. 
This is then followed by the Clifford gates and Pauli measurements in the circuits. 
One can efficiently track how the stabilizer group and the group of logical operators change after each operation's action~\cite{aaronson2004improved}.

The spacetime stabilizer operators are then constructed from the linear constraints that the measurement results of a noise-free circuit satisfy. 
Given an initial stabilizer group with $s$ generators and $m$ measurements in the circuit, we denote the measured operators by $\vec E=\left(S^{(\text{ini})}_1,...,S^{(\text{ini})}_s,M^{(\text{circ})}_1,...,M^{(\text{circ})}_m\right)$. 
The measurement results of these operators are stored in a binary vector $\vec o\in\mathbb{Z}_2^{s+m}$, 
where we assign $o_i=0$ when the value of $E_i$ is $+1$ and $o_i=1$ if it is $-1$. In the absence of circuit errors, whenever a stabilizer element is measured, the measurement result is deterministic and depends on previous measurement outcomes. The measurement result $\vec o$ in the circuit therefore needs to satisfy consistency relations, or \emph{parity checks} $\mathcal{O}^{\perp}=\{\vec u\}$, in that every set of noise-free measurement outcomes $\vec o$ satisfies $\vec u\cdot\vec o=0$ for all $\vec u \in \mathcal O^\perp$. 
Circuit errors can be detected from flips of those parity checks, which coincides with the notion of detector \cite{fowler2014scalable,derks2024designing}.

The generators of the measured spacetime stabilizer group $\mathcal{M}^{\text{(ST)}}$ are then constructed from the parity constraints $\vec u\in\mathcal{O}^{\perp}$ of the circuit as
\begin{align}
    M^{\text{(ST)}}(\vec u)=\overleftarrow{\prod^{s+m}_{i=s+1}\kappa_{t_{E_i}-0.5}(E_i^{u_i})},\label{equation: spacetime stabilizer from u}
\end{align}
where $\kappa_t(O) \in \mathcal{P}_{n(T+1)}$ denotes the operator obtained by embedding $O \in \mathcal{P}_n$ at layer $t$ and tensor product with identity $I_n$ for all other layers. $t_{E_i}$ denotes the layer where $E_i$ is measured (for initial stabilizer generators $S^{(\text{ini})}$, $t_{S^{(\text{ini})}}=0$). One notable property of $M^{\text{(ST)}}(\vec u)$ is that it is bounded by the time involved in the parity constraint $\vec u\in\mathcal{O}^{\perp}$ (Lemma 8 in \cite{delfosse2023spacetime}), i.e., for all $(q,t)\in\text{supp}(M^{\text{(ST)}}(\vec u))$ we have
%\begin{equation}
%\{t|(q,t)\in\text{supp}(M^{\text{(ST)}}(u))\}\in\mathcal{T}(u)=\%{t_{\mathcal{B}_j}-0.5|u_j=1\}.\label{equation: local spacetime s property}
%\end{equation}
\begin{equation}
\min_{\,j\in[s+m]\;\mid\;u_j=1}\; t_{E_j} <t<\max_{\,j\in[s+m]\;\mid\;u_j=1}\; t_{E_j}.\label{equation: local spacetime s property}
\end{equation}

We now introduce the measured logical group of the spacetime code. We emphasize that we want to understand the logical errors only on the relevant logical measurements of the circuit. This makes the spacetime code a subsystem code since, in general, the set of undetectable yet logically trivial errors in the circuits corresponds to a non-abelian subgroup in $\mathcal{P}_{n(T+1)}$ for the spacetime code. See example of such errors in \cref{fig: spacetime code mapping}(b). These logically trivial errors generate the gauge group $\mathcal{G}^{(\text{ST})}$ of the spacetime code. See \cref{appendix: spacetime code} for the construction. 

In a subsystem code, logical Pauli operators that commute with the gauge group encode the logical information and are called the bare logical operators. In contrast, dressed logical operators commute with the stabilizer group but may anticommute with some gauge operators \cite{pastawski2015fault}; these are the logical errors of the code. Unless otherwise specified, we use “logical operators” to mean bare logical operators. In this work, we focus on a subgroup of spacetime logical operators $\mathcal{L}'^{\text{(ST)}}$ corresponding to the logical Pauli measurements in the circuit. In a Clifford circuit with $m_L$ logical measurements, we denote the logical measurement outcome as $\vec o^L\in\mathbb{Z}_2^{m_L}$. In a noiseless circuit, each logical Pauli measurement result $o^L_i$ will depend on a set of physical measurement results as $o^L_i=\vec u^L\cdot \vec o$ for some binary vector $\vec u^L\in\mathbb{Z}_2^{s+m}$. Here, the non-zero elements of $\vec u^L$ correspond to the physical measurements involved. $\vec u^L$ corresponds to a generator
\begin{align}
    L'^{\text{(ST)}}(\vec u^L)=\overleftarrow{\prod^{s+m}_{i=s+1}\kappa_{t_{E_i}-0.5}(E_i^{u^L_i})}
\end{align}
of the group of measured logical operators $\mathcal{L}'^{\text{(ST)}}$ in the spacetime code.
See \cref{fig: spacetime code mapping} for an example. The detailed generic algorithm of constructing $\mathcal{O}^{\perp}$, $\mathcal{M}^{\text{(ST)}}$, and $\mathcal{L}^{\text{(ST)}}$ is explained in \cref{appendix: spacetime code}. 

%The gauge group $\mathcal{G}^{\text{(ST)}}$ of the spacetime code in our work will be defined as the normalizer of the measured logical group $\mathcal{L}'^{\text{(ST)}}$. It encodes the indistinguishable and logically equivalent circuit errors. We do not need to explicitly construct them for the learning algorithm. We defer the gauge group construction to \cref{appendix: spacetime code}, see also Ref.~\cite{gottesman2022opportunities}.

Pauli errors in the circuit can also be mapped onto the spacetime code. A circuit error represented as $e^{(\text{circ})}=(e_t\in\mathcal{P}_n\mid t\in [T+1])$, i.e., inserting Pauli error $e_t$ before each circuit layer $t$, can be mapped to the error $e^{\text{(ST)}}\in\mathcal{P}_{n(T+1)}$ in the spacetime code as
\begin{equation}
    e^{\text{(ST)}}=\prod_{t\in[T+1]}\kappa_{t-0.5}(e_t).\label{equation: circuit to spacetime error mapping}
\end{equation}
\begin{comment}
\Dom{I'm not sure why the following paragraph is important. You've already said that you can forward or backward-propagate. I guess you want to talk about the errors but what exactly is it you want to say? }
In the circuit picture, the syndrome and logical effect of the circuit error $e^{(\text{circ})}$ are determined by first propagating errors in $e^{(\text{circ})}$ to the circuit layers containing the relevant measurements associated with a stabilizer/logical measurement, then calculating the commutation relation. This is equivalently reflected in the commutation relation of $e^{\text{(ST)}}$ with elements in the measured spacetime stabilizer $\mathcal{M}^{\text{(ST)}}$ or the logical operators $\mathcal{L}^{\text{(ST)}}$ (see example in \cref{fig: spacetime code mapping}) due to the identity (Proposition 3 in \cite{delfosse2023spacetime}) that for all $e^{\text{(ST)}},O^{\text{(ST)}}\in\mathcal{P}_{n(T+1)}$
\begin{equation}\label{equation: spacetime commutation identity}
    [[\overrightarrow{e^{\text{(ST)}}},O^{\text{(ST)}}]]=[[e^{\text{(ST)}},\overleftarrow{O^{\text{(ST)}}}]].
\end{equation}
Here, $\overrightarrow{O}$, referred to as the cumulants, are constructed similarly to \cref{equation: spacetime stabilizer from u}:
\begin{align}
    (\overrightarrow{O})_{t-0.5} = \prod_{t'=1}^{t}U_{t',t}O_{t'-0.5}U^{\dagger}_{t',t}.
\end{align}
\end{comment}
Throughout the remainder of the text, we identify circuit errors $e^{(\text{circ})}$ with their corresponding spacetime-code Pauli errors $e^{(\mathrm{ST})}\in\mathcal{P}_{n(T+1)}$ via equation \cref{equation: circuit to spacetime error mapping} and its inverse relation $e^{(\text{circ})}=(e^{(\text{ST})}_t\mid t\in[T+1])$.

We now define the circuit-level Pauli error model in a way that admits a natural representation as a Pauli noise model (\cref{section: model}) on the spacetime code by identifying each spacetime location $(q,t)$ in the circuit with the spacetime qubit labeled by $(q,t)$.
\begin{definition}[Circuit-level Pauli error model]\label{definition: Circuit-level Pauli error model}
    We define a circuit-level Pauli model by a set of spacetime supports $\Gamma^{(\text{ST})}=\{\gamma_i\}_i$. Each support $\gamma\in\Gamma^{(\text{ST})}$ is a set of spacetime locations $\{(q_1,t_1),(q_2,t_2),...\}$ in the circuit labeling the spacetime qubits. Each support is associated with a Pauli error distribution $P_{\gamma}:\mathcal{P}_{n(T+1)}\to [0,1]$ and $P_{\gamma}(e)=0$ if $\text{supp}(e)\not\subset \gamma$. The circuit-level Pauli error model is then a distribution of Pauli errors over the whole circuit given by the convolution of local distributions $P_{\Gamma^{(\text{ST})}}=\Asterisk_{\gamma\in\Gamma^{(\text{ST})}}P_{\gamma}$.
\end{definition}

Pauli error models on the spacetime code correspond to circuit-level noise models on the circuit with errors correlated in space and time, depending on the structure of the set of spacetime supports $\Gamma^{(\text{ST})}$. Therefore, a general Pauli noise model could, in principle, capture features of correlated errors, leakage errors, and non-Markovian errors, but the practical application for more complicated error models will be out of the scope of this paper. In the remaining numerical and experimental results, we instead consider the Pauli noise model on the spacetime code according to the standard circuit model in the fault-tolerance literature, where errors are only correlated between qubits involved in each circuit location defined below.
%\Dom{Say what a circuit location is.}
%\Dom{I think a figure would be helpful here. }
\begin{definition}[Circuit location]
    \label{definition: circuit location}
    We define circuit locations as the indivisible operations in a circuit, including idling, state preparation, gates, and measurements. Each location is associated with a time step and a set of qubits it acts upon. Locations at the same time step do not share any qubit.
\end{definition}
%\Dom{Define circuit location.}
\begin{definition}[Standard circuit-level Pauli error model]
    \label{definition: standard circuit Pauli error model}
     The standard circuit-level Pauli noise model is a restricted version of the circuit-level Pauli noise model in \cref{definition: Circuit-level Pauli error model}. In this mode, a $w$-qubit Pauli error channel is placed before every $w$-qubit circuit location except for state preparation locations, where the error channel is applied afterwards. The channel acts on the same $w$ qubits involved in the corresponding location.
\end{definition}

The standard circuit-level Pauli noise model, apart from the state preparations at the beginning, corresponds to a non-overlapping Pauli error channel $\mathcal{N}_{\Gamma}$ on the spacetime code. Specifically, $\mathcal{N}_{\Gamma}$ is the composition of $w$-qubit Pauli channels, each acting on the $w$ spacetime qubits at time $t-0.5$ corresponding to a circuit location involving those $w$ qubits at time $t$. For example, a noisy CNOT gate at circuit layer $t$ on qubit $i$ and $j$ corresponds to a local Pauli channel $\mathcal{N}_{\{(i,t-0.5),(j,t-0.5)\}}$ on the spacetime code.

%\Dom{Is there a freedom to choose whether noise is applied before or after the gate?}

\begin{figure*}[t]
% \printlen[5][pt]{\linewidth}
\centering
\begin{subfigure}{0.8\textwidth}
\phantomcaption
\stackinset{l}{-20pt}{t}{0pt}{\captiontext*}
    {\includegraphics[width=\textwidth]{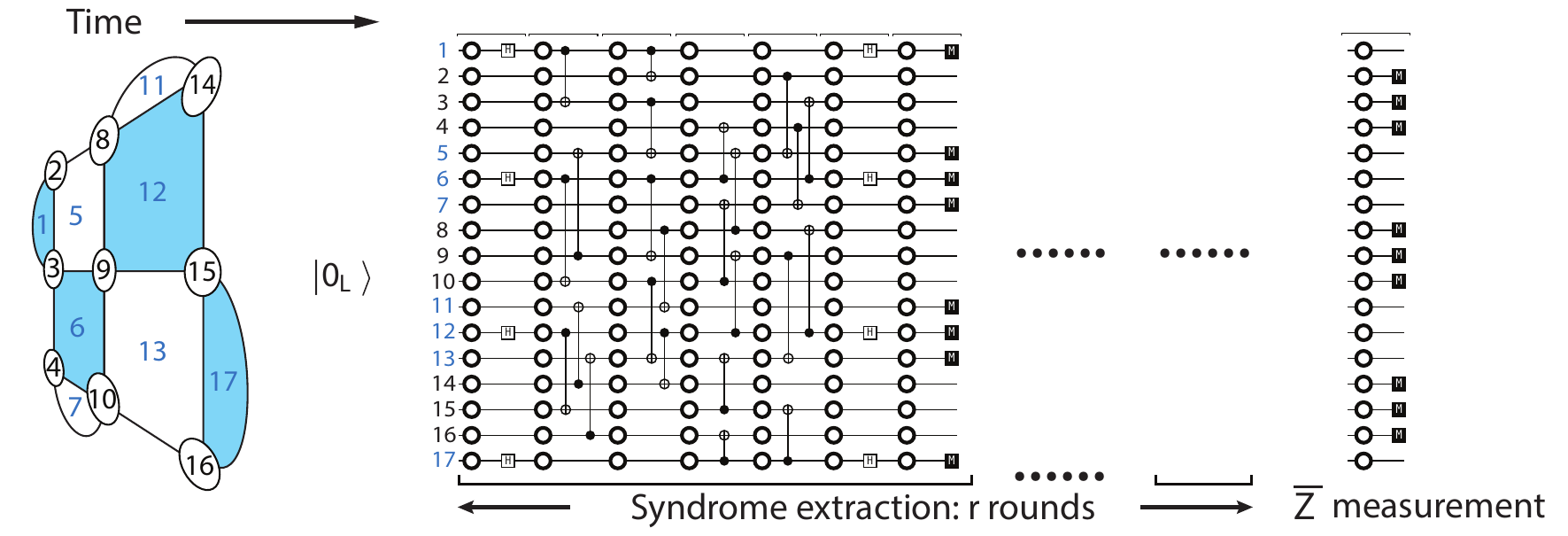}}
    \label{fig: spacetime_surface_code_circuit}
\end{subfigure}
\begin{subfigure}{0.9\columnwidth}
\phantomcaption
\stackinset{l}{0pt}{t}{0pt}{\captiontext*}
    {\hspace{0.2cm}\includegraphics[width=\textwidth]
    {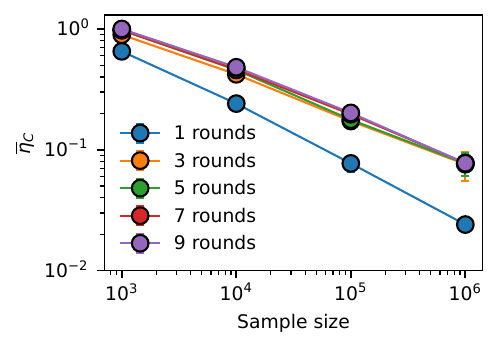}}
    \label{fig: spacetime_syndrome_class_relative}
\end{subfigure}
\hspace{0.5in}
\begin{subfigure}{0.9\columnwidth}
\phantomcaption
\stackinset{l}{-5pt}{t}{-2pt}{\captiontext*}
    {\includegraphics[width=\textwidth]{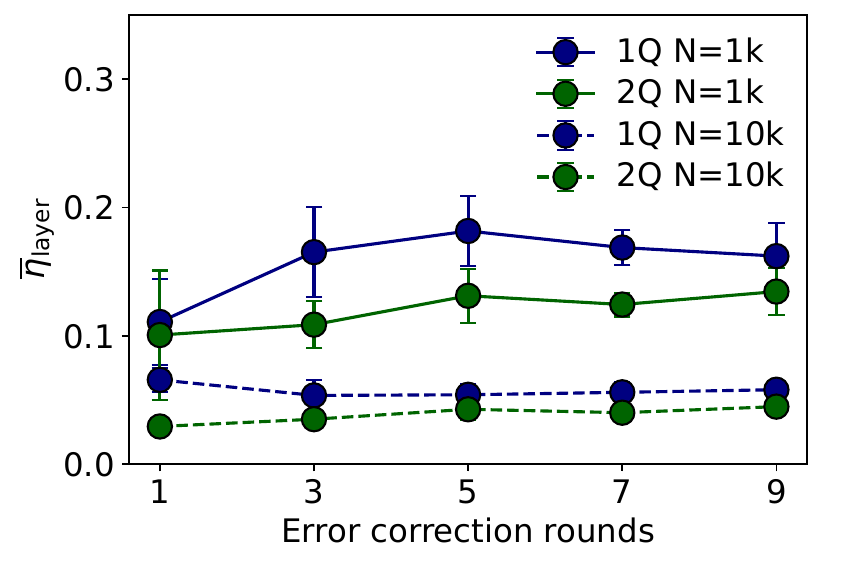}}
    \label{fig: spacetime_layer_relative}
\end{subfigure}
 \caption{(a) Error correction circuit for the $3 \times 3$ rotated surface code and its corresponding spacetime code. The surface code starts from its logical $\ket{\overline 0}$ state and we simulate $r$ rounds of syndrome extractions followed by a transversal $Z$ measurement. The learning algorithm converts the whole circuit into a spacetime code where spacetime qubits (empty circles) are placed before every layer of gates. (b) The relative error $\overline{\eta}_{C}$ of the learned syndrome class total error rate averaged over all classes. We simulated for different rounds of syndrome extraction. (c) relative error $\overline{\eta}_{\mathrm{layer}}$ of the estimated average single-qubit (1Q) and two-qubit (2Q) error rates per layer. The true circuit-level error rates are sampled from a Gaussian distribution ($\mu=0.001/3,\sigma=0.0002/3$ and $\mu=0.01/15,\sigma=0.002/15$ for the Pauli error components of single-qubit and two-qubit locations respectively).}
\label{fig: spacetime surface code}
\end{figure*}

\subsection{\label{subsection: simulation}Simulation on circuit-level noise learning}

We first test our noise learning algorithm on the spacetime code of the Clifford circuits by performing simulations on surface code quantum memory circuits, constructed as follows. We initialize in the initial logical state of the $3 \times 3$ rotated surface code in the logical $\ket{\overline 0}$ state, and place one ancilla qubit in the state $\ket{0}$ for each plaquette of the rotated surface code. 
A single round of syndrome measurement has seven layers of gates, consisting of performing Hadamard gates on the $X$ syndrome ancillas, using CNOT gates to entangle ancilla and data qubits, and measuring ancilla qubits in the correct basis to read out the syndrome bits, see \cref{fig: spacetime_surface_code_circuit}. 
Here, the order of CNOTs is chosen to maximize parallelism and ensure fault tolerance \cite{google2023suppressing}. 
We perform several rounds of syndrome extraction before a final transversal measurement in the computational basis to perform a fault-tolerant logical $\overline{Z}$ measurement. 
For simplicity, we do not perform ancilla qubit reset between rounds of syndrome extraction, but simply keep track of all measurement records \cite{geher2024reset}. The spacetime stabilizer measurement values (detectors) are calculated based on these records.

%\Dom{What do you mean by "one could"? Do we, or do we not? I can't follow the following explanation, but I think it makes a very important claim. Maybe an example and a sketch would help? You also claim that this generalizes to general LDPC codes, but I don't understand how the backwards light cone of the stabilizers will always cancel out. And then you say that the locality is $8r+4$, which is not LDPC, no?} 
Given a Clifford circuit, there are different ways of selecting a set of generators for the measured spacetime stabilizer group $\mathcal{M}^{\text{(ST)}}$ corresponding to different bases that implement the same parity check matrix $\mathcal{O}^{\perp}$ on the measurement outcomes. With the property in \cref{equation: local spacetime s property}, we optimize the set of generators such that the weight of each $M^{(\text{ST})}(\vec u)$ is constant by constructing a generating set of $\{\vec u\}$ for $\text{span}(\mathcal{O}^{\perp})$ that involve measurements in a local time window. Specifically, for the surface code syndrome extraction circuit in \cref{fig: spacetime_surface_code_circuit}, we define the local spacetime stabilizer generators based on the following parity constraints: For the first two rounds ($r=1,2$) of syndrome extractions, we set each $u$ to encode the parity of an ancilla measurement result and the corresponding initial stabilizer value. For syndrome measurements at round $r>2$, we set each $u$ to encode the parity of an ancilla result at round $r$ and round $r-2$. Lastly, the transversal $Z$ measurement at the end also creates $Z$ type spacetime generators. For each $Z$ stabilizer generator $S_z$ of the surface code, we set an additional parity check $u$ as the parity of the final data qubit bit strings for $S_z$ and the corresponding ancilla measurement at the last two rounds.
%\Dom{I think I'm not following where the parity checks come from.} 
We construct the local spacetime generators with these parity checks $u$ according to \cref{equation: spacetime stabilizer from u}. A $r$-round $3 \times 3$ surface code EC circuit has $8r+4$ local spacetime stabilizer generators. We feed these generators into \cref{algorithm: reduced stabilizer group} to find the minimal set of equations in \cref{equation: log equation reduced}. Under the standard circuit-level Pauli noise model for a 9-round syndrome extraction circuit, we obtain $961$ independent equations corresponding to $|\mathcal{C}^*|=961$ nontrivial syndrome error classes.

We simulate the standard circuit-level Pauli error model where each Pauli component of single-qubit and two-qubit channels has an error rate following a Gaussian distribution centered at near-threshold error rates $0.001/3$ and $0.01/15$, fluctuating with standard deviation $0.0002/3$ and $0.002/15$ respectively. After converting the circuit into a spacetime code, we apply our learning algorithm \cref{equation: optimization} for static code under various extra constraints $B$ using sample sizes $N=10^3,10^4,10^5,10^6$. 

To evaluate the quality of circuit noise learning, we first study how well we can learn the total error rates of the syndrome class of the corresponding spacetime code. Errors within each syndrome class correspond to circuit-level errors that violate the same parity constraints (detectors) $\forall u\in\mathcal{O}^{\perp}$. Here, we use the optimization method \cref{equation: optimization} instead of the approximation method in \cref{equation: optimization D'} as we observe the result from optimization is slightly better, although the latter is more efficient in terms of classical computational cost. We use the simple extra assumption $B$ which assumes uniform noise for all local Pauli errors, i.e., $r(e)=1, \forall e\in \mathcal{E}_{\Gamma}$ in \cref{equation: extra_nontrivial} and \cref{equation: extra_trivial}. In \cref{fig: spacetime_syndrome_class_relative}, we show the average relative error $\overline{\eta}_{C}=\frac{1}{|\mathcal{C}^*|}\sum_{C\in\mathcal{C}^*}\eta_C$ of the predicted syndrome class total error rates to the true error rates. With around $N=1000$ samples, the learning algorithm can learn the orders of magnitude of the syndrome class total error rates. By increasing the sample size to $N = 10^4$ to $10^5$, we achieve better relative precision, and the estimation quality remains stable as the number of syndrome extraction rounds increases. We note that we expect the accuracy can be substantially improved for sample size $N=1000 \text{ to } 10000$ if techniques of reducing shot noise are used, according to the comparison of circles and dotted lines shown in the \cref{fig: precision of two examples} in the code capacity setting.

We also estimate the average error rate for single-qubit locations and two-qubit locations (CNOT gates) per layer, and we expect it to require a smaller sample size to reach the same relative precision compared with estimating syndrome class total error rates. We implement the optimization \cref{equation: optimization} with $B$ following the average noise ratio of single-qubit and two-qubit gates, i.e., $r(e^{(1Q)})=1,r(e^{(2Q)})=10$. We consider $\overline{\eta}_{\mathrm{layer}}^{(1Q)}$ and $\overline{\eta}_{\mathrm{layer}}^{(2Q)}$ which represent the relative error of our estimated single-qubit and two-qubit error rates per layer averaged over all layers, and they are shown in \cref{fig: spacetime_layer_relative}. With $1000$ samples under this noise setup, the learning algorithm can estimate single-qubit and two-qubit gate average error rates per layer within $0.2$ relative error.

\begin{figure}[t]
\centering
\hspace{0.0cm}
\begin{subfigure}{0.95\columnwidth}
    \phantomcaption
    \stackinset{l}{0pt}{t}{-5pt}{\captiontext*}
    {\includegraphics[width=\textwidth]{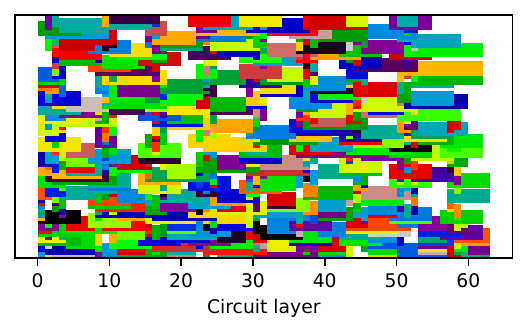}}
    \label{fig: surface temporal span}
\end{subfigure}
\begin{subfigure}{0.95\columnwidth}
    \phantomcaption
    \hspace{0cm}\stackinset{l}{0pt}{t}{-5pt}{\captiontext*}
    {\includegraphics[width=\textwidth]{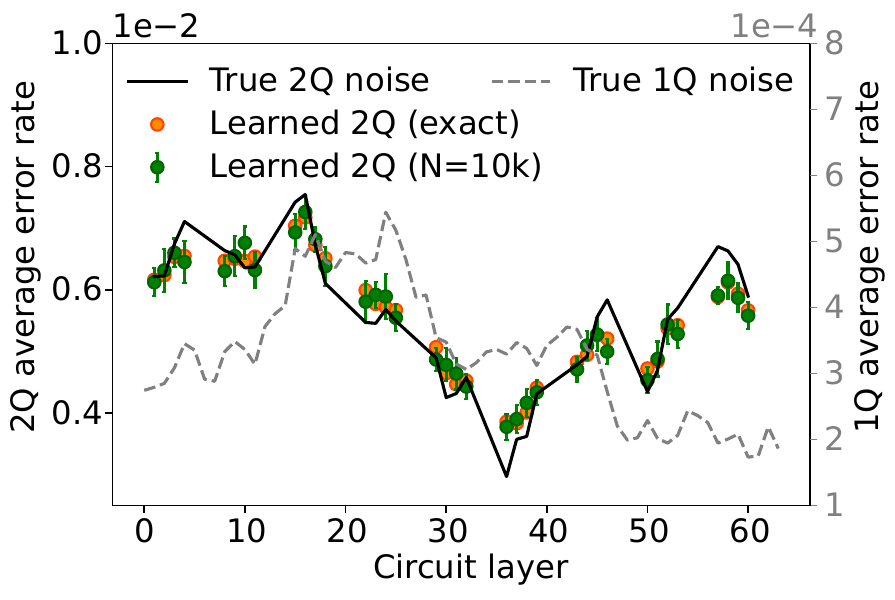}}
    \label{fig: two_qubit_drift}
\end{subfigure}
\caption{(a) The temporal span of 961 nontrivial syndrome error classes $\mathcal{C}^*$ for the 9-round $3 \times 3$ surface code EC circuit (\cref{fig: spacetime_surface_code_circuit}) with circuit level noise. The vertical position is arbitrary as we focus on the resolution in time, and the colors are only for distinguishing different classes. 331 classes are contained in a single layer and 297 classes involve only two consecutive layers. (b) The estimated average two-qubit (2Q) error rates per layer using the simulated syndrome data. We apply a time-varying error rate that follows a random walk for 1Q and 2Q errors independently. The extra constraints encoded in $B$ (\cref{equation: optimization}) is set according to the initial ratio of $p^{(1Q)}_0$ and $p^{(2Q)}_0$. We observe that the learned error rates both from the exact syndrome (orange dots) and from $N=10,000$ samples (green dots) match the true error rates and serve as good estimates of the time-dependent noise.}
\label{fig: time-dependent noise}
\end{figure}

\begin{figure*}
\centering
\begin{subfigure}{0.75\textwidth}
\phantomcaption
\stackinset{l}{-10pt}{t}{0.5pt}{\captiontext*}
    {\includegraphics[width=\textwidth]{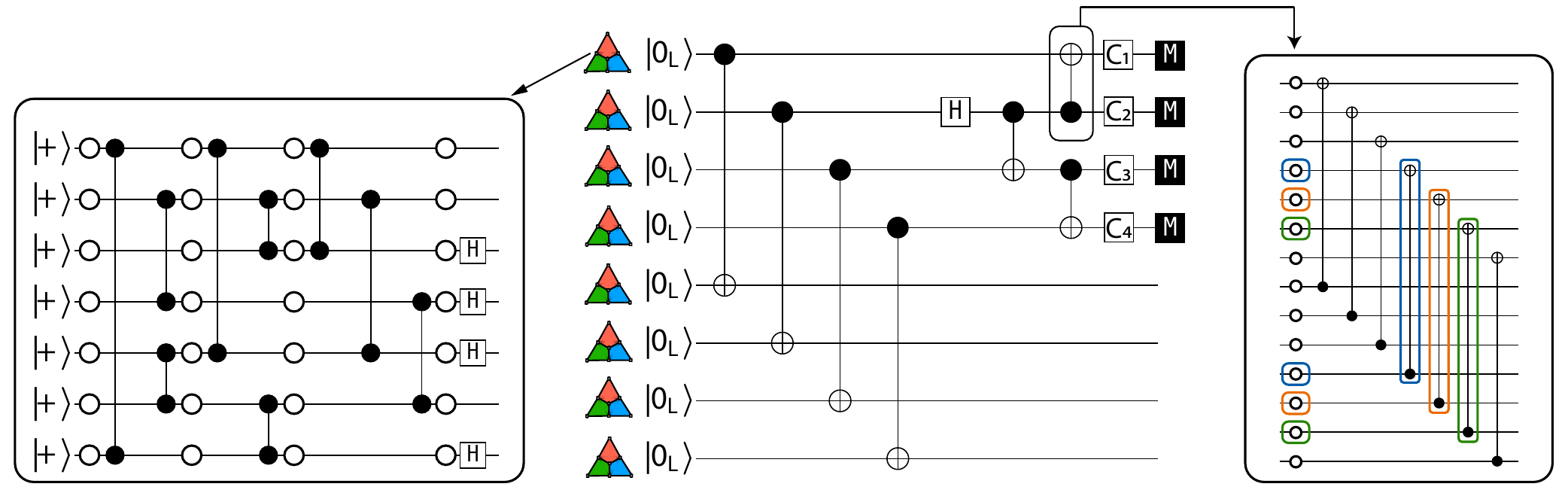}}
    \label{fig: logical circuit}
\end{subfigure}
\begin{subfigure}{0.75\textwidth}
    \phantomcaption
    \stackinset{l}{-10pt}{t}{0.5pt}{\captiontext*}
    {\includegraphics[width=\textwidth]{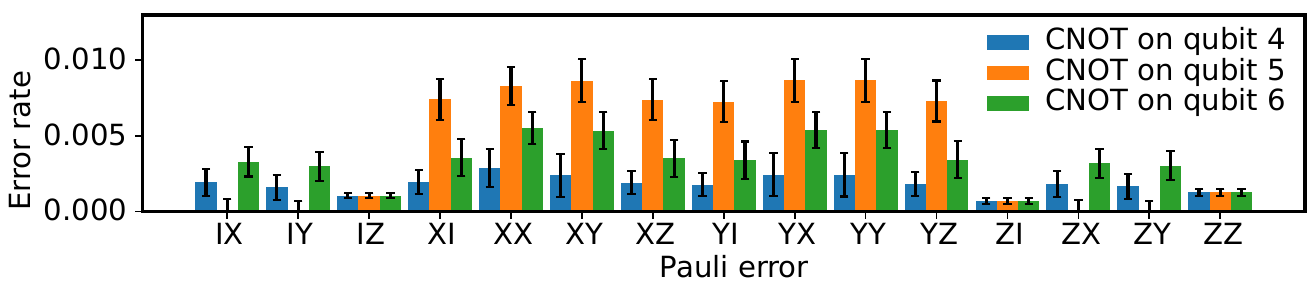}}
    \label{fig: physical circuit learned}
\end{subfigure}
\caption{Logical GHZ state preparation experiment and physical error rate learning. (a) The box on the left shows the encoding circuit for logical $\ket{\overline 0}$ of the $[[7,1,3]]$ Steane code. A series of logical circuits for preparing the 4-qubit encoded GHZ state and logical stabilizer measurements are shown in the middle. Here, logical Hadamard and CNOT gates are transversal gates acting on code blocks. The lower four code blocks fault-tolerantly verify the prepared $\ket{\overline 0}$. The last two layers implement direct fidelity estimation in experiments where logical Clifford gates are applied and then qubits are measured transversally in the Z basis. The box on the right shows the physical circuit of the transversal CNOT gates of the enclosed logical CNOT gate. Empty circles shown in the physical circuits are qubits of the spacetime code, susceptible to errors from gates to their right in the diagram. (b) Learned Pauli error rates of the CNOT gates enclosed in the middle logical circuit for the logical $\overline{ZZZZ}$ measurement circuits are shown. The learning algorithm is applied based on the syndrome data of the spacetime stabilizer (detectors) constructed without the knowledge of the initial logical state of the data blocks.}
\label{fig: logical circuits and learning}
\end{figure*}

Furthermore, we investigate how well the learning algorithm could perform under time-dependent noise. In time series analysis, one could model discrete time-dependent noise as an average over a series of Markov chains $p_t=\frac{1}{M}\sum^{M}_{m=1}p^m_t,\,t=0,.., T$. Here $p^m_t$ is of the form \cite{granger1980long,erland2007constructing}:
\begin{equation}
    p^m_t=\theta^m p^m_{t-1}+\sigma\eta_t,\text{ }t=1,\ldots,T.
\end{equation}
where $\eta_t$ are i.i.d random variables following the standard normal distribution and $\theta^m$ is a parameter from $0$ to $1$. One could approximate the noise signal with spectral density $S(f)\sim 1/f^{\alpha}$ if $\theta^m$ following a certain distribution parametrized by $\alpha$ (Eq. (15) in \cite{erland2007constructing}). With $\theta^m=0$, we have non-correlated noise, also known as white noise, whereas when $\theta^m=1$, we have noise exhibiting the 1D random walk ($S(f)\sim 1/f^{2}$). For simplicity, we simulate the $3 \times 3$ surface code circuit with unbiased Pauli errors that only distinguish single qubit operations (including measurement and idling noise) and two quibit gates with initial total error rate $p^{(1Q)}_0=3\times 10^{-4}$ and $p^{(2Q)}_0=5\times 10^{-3}$, respectively. We model 1Q and 2Q noise separately, each being uniform in every layer and evolving as a random walk with parameter $M=1, \theta=1, \frac{\sigma}{p_0}=0.1$.

Similarly, we use the gate-based extra constraints in the $B$ matrix for the optimization, i.e., $r(e^{(1Q)})/r(e^{(2Q)})=p^{(1Q)}_0/p^{(2Q)}_0$. We note that we have worse predictions of average error rates per layer for less noisy processes, e.g., single-qubit gates, when error rates are highly time-dependent. On the other hand, the algorithm can capture the error rates for the noisy process, i.e., two-qubit gates in this case. In \cref{fig: two_qubit_drift}, we show the result of the learned two-qubit gate errors from both exact syndrome expectations and syndrome expectations from $N=10K$ samples. The learned error rates averaged at every layer are consistent with the true error rates, even at this small sample size. This exceeds our naive expectation that noise only resolved between different syndrome extraction rounds, as we only have one round of measurement at the end of the depth-$7$ syndrome extraction circuit. The resolution can be explained by the short temporal span of the syndrome error classes, as shown in \cref{fig: surface temporal span}. Out of 961 syndrome error classes, 331 classes are contained in a single layer and 297 classes involve only two consecutive layers, whereas a single EC cycle includes 7 layers of gates.
%\Dom{Can you explain this?}

\subsection{Learning from experimental data\label{section: physical_learning_GHZ}}

We consider learning the circuit-level noise model of a fault-tolerant circuit performed on the Harvard neutral atom platform \cite{bluvstein2024logical} that prepares a logical 4-qubit GHZ state, as shown in \cref{fig: logical circuit}. The logical circuit is encoded in four blocks of seven-qubit Steane code by the non-fault-tolerant $\ket{\overline 0}$ preparation circuit and then fault-tolerantly verified by the other four blocks of logical ancillas in $\ket{\overline 0}$. Then, the logical GHZ state is prepared using logical Hadamard and CNOT gates, which are transversal for the Steane code. To correct errors and verify logical fidelity, the direct fidelity estimation \cite{flammia2011direct} on the logical level was performed by measuring all
% applying logical Clifford gates at the end of the circuit 
 % and we do logical $\bar{Z}$ measurement, i.e. transversally measure $Z$ on each qubit. 
 % This effectively measures all 
elements of the logical stabilizer group $\langle\overline{ZZII}, \overline{IZZI}, \overline{IIZZ}, \overline{XXXX}\rangle$ of the logical GHZ state. The logical measurements in the $\overline{Z}$ basis, i.e., $\langle\overline{ZZII}, \overline{IZZI}, \overline{IIZZ}\rangle$, can be done simultaneously with a single logical circuit; therefore, 9 logical circuits are executed in total.
%\Dom{Explain "logical stabilizer"} 
The initial stabilizer generators of the circuit are
$X_1, X_2, X_5, X_6, X_3X_4,X_3X_7$, which propagate to the stabilizer generators of the Steane code for each data code block, along with all single-qubit X operators for the ancilla blocks. These are provided as input to \cref{algorithm: spacetime stabilizer} to construct the measured spacetime stabilizer group of each logical circuit.

To get a reasonable estimate of gate error rates, we impose extra constraints $B$ using previously reported gate fidelities \cite{evered2023high} as a reference. 
Specifically, for errors within each syndrome class of the spacetime code, we differentiate them by the prior belief $r(e)$ \cref{equation: extra_nontrivial,equation: extra_trivial}, i.e., a relative strength between errors of different circuit location type, as $r=1,3,10,20$ 
%\Dom{Remind me what $r$ is.}\Xiao{done}
for idle, single-qubit, two-qubit, and measurement errors respectively. Each logical circuit is repeated and postselected on successful initialization, resulting in a sample size $N$ of less than 3000 times, and these sampled measurement bitstrings are used to obtain the expectation values of the spacetime code syndrome, i.e., the empirical Pauli eigenvalues. Part of the learned logical CNOT gate error rates are shown in \cref{fig: physical circuit learned}. 
%\Dom{Polish the following.}\Xiao{}
To estimate the uncertainty of the learned noise, we simulate sampling syndrome with the same number of samples as in the experiment and apply the learning procedure to these simulated syndrome data. The resulting error bars represent the standard deviation of the learned error rates over 30 repetitions. We can distinguish the relative noise strength of different gates and Pauli components with a certain confidence. This can be useful to identify potential circuit defects to perform gate tune-up. One of the features we observe in common with the learned error rates in this collection of circuits is that the gate fidelity tends to diminish over time. We hypothesized that this is caused by atom loss. See \cref{appendix: more plots}, \cref{fig: logical GHZ physical learning appendix,fig: logical GHZ physical learning relative appendix}, for the complete estimated error rates for the $\overline{ZZZZ}$ measurement circuit.
%\Dom{You often join different thoughts with an "and" that are sequential thoughts. Make short sentences! }\Xiao{}

\section{\label{section: logical error rate and fidelity}Learning logical error rates and logical fidelity from syndrome data}

In the previous section, we showed how to explicitly characterize physical circuit-level errors in fault-tolerant circuits using syndrome data. In this section, we ask if that same data can also be used to benchmark the quality of the implemented \emph{logical circuits}. We extend the results of \textcite{wagner2023learning} from the channel-coding setting to non-adaptive Clifford circuits, allowing arbitrary stabilizer code codewords as inputs and logical measurements. Given a Clifford circuit that is fault-tolerant to the assumed circuit-level Pauli error model (\cref{definition: Circuit-level Pauli error model,definition: fault tolerant}), we first prove that with respect to the set of logical measurements, the total error rate of the classes of logically equivalent errors with a given syndrome is learnable using the syndrome expectation values. We emphasize that this is a decoder-independent quantity.

Consequently, with a decoder, we show that logical error rates for the set of logical measurements are learnable. In practice, we take the estimated physical error rates that satisfy the syndrome expectation, as described in \cref{section: physical error rate estimation}, and perform Monte Carlo simulations using any relevant decoder to learn the logical error rates. Notably, the input state of the circuit can be any codeword (e.g., non-stabilizer state) since errors of logical measurements in Clifford circuits can be characterized from the commutation relation between circuit errors and the spacetime code logical operators (see \cref{appendix: spacetime code}).
%\Dom{What is "the Pauli error model"?}

\subsection{Logical error rate learnability\label{section: logical learnability}}
We consider an ideal logical circuit with a set of $m_L$ logical measurements $\mathcal{M}_L=\{\overline{M}_i\}_i$ implemented by a physical fault-tolerant Clifford circuit. For example, for the logical GHZ preparation circuit shown in \cref{fig: logical circuit} with $Z$ basis measurements, the set of logical measurements is $\mathcal{M}_L=\{\overline{ZIII}, \overline{IZII}, \overline{IIZI},\overline{IIIZ}\}$ at the final layer. The ideal logical circuits can be thought of as a quantum-to-classical channel. Given an input state, the channel can be described by a distribution over the output classical bitstrings $\vec o^{(\text{ideal})}_L\in\mathbb{Z}_2^{m_L}$ corresponding to the ideal logical measurement outcome. When there is a circuit-level Pauli error channel, the logical output is obtained by first applying additional Pauli correction to the circuit based on the observed syndrome using a decoder and then determining the corrected logical outcome. The errors and the subsequent correction introduce an additional bit-flip channel on the ideal logical measurement outcomes, which can be described by a distribution $P_{\mathrm{fail}}$ over the possible logical errors, i.e., bit flips $\vec f_L\in\mathbb{Z}_2^{m_L}$. The final logical output is then given by $\vec o^{(\text{ideal})}_L\oplus \vec{f_L}$. We now show that given any decoder, the bit-flip channel on logical measurement outcome, described by the logical error rate $P_{\mathrm{fail}}(\vec f_L)$, is learnable from syndrome data under a fault-tolerant condition.

Given the physical error rates, we quantify the probability of a logical error happening for logical measurements using the spacetime code picture introduced in \cref{section: spacetime code mapping}. We map the physical circuit into an $n^{(\text{ST})}$-qubit spacetime subsystem code. In particular, we have the parity checks $\mathcal{O}^{\perp}$ in the physical measurements mapped to the measured spacetime stabilizer group $\mathcal{M}^{(\text{ST})}$, and the logical measurements $\mathcal{M}_L$ mapped to the group of measured spacetime logical operators $\mathcal{L}'^{(\text{ST})}$ using \cref{algorithm: spacetime stabilizer}. The effect of any circuit Pauli error on the syndrome and logical measurement can be captured by the commutation between the corresponding error $e\in\mathcal{P}_{n^{(\text{ST})}}$ in the spacetime code and the spacetime stabilizer and logical elements (see \cref{lemma: Equivalence of syndrome in circuit and spacetime code} and \cref{lemma: Equivalence of logical error in the spacetime code} in \cref{appendix: spacetime code}). As a result, an error $e\in\mathcal{P}_{n^{(\text{ST})}}$ is undetectable and logically trivial to the logical measurements if and only if $e$ is in the gauge group $\mathcal{G}'^{\text{(ST)}}=N(\mathcal{L}'^{\text{(ST)}})$ of the spacetime code. Therefore, any two Pauli errors that differ by elements of the spacetime code gauge group are indistinguishable and logically equivalent. For each Pauli error $e$ in the spacetime code, we define its \emph{logical error class} as the coset $e\mathcal{G}'^{\text{(ST)}}$ in $\mathcal{P}_{n^{(\text{ST})}}$ and we denote it as $\overline{C}_{\vec{s},\vec{\ell}}$. Here, $\vec{s}\in\mathbb{Z}_2^{r}$ and $\vec{\ell}\in\mathbb{Z}_2^{m_L}$ are the syndrome and logical component of the error $e$ respectively:
\begin{align}
    s_i(e)&=\langle e,M^{(\text{ST})}_i\rangle\\
    \ell_i(e)&=\langle e,L'^{(\text{ST})}_i\rangle,
\end{align}
where $\{M_i^{(\text{ST})}\}^r_i$ generates the measured stabilizer group. $\{L_i'^{(\text{ST})}\}^{m_L}_i$ is the set of measured spacetime logical operators from the logical measurements. Furthermore, we denote the probability of each logical error class as
\begin{equation}
\overline{P}\left(\vec{s},\vec{\ell}\right)=\sum_{e\in\overline{C}_{\vec{s},\vec{\ell}}}P(e)=\sum_{g\in\mathcal{G}'^{(\text{ST})}}P(e_{\vec s,\vec \ell}\,g),
\end{equation}
where $e_{\vec s,\vec \ell}$ is any error satisfying $\vec{s}(e)=\vec{s}$ and $\vec{\ell}(e)=\vec{\ell}$.
Note that the number of degrees of freedom of $\overline{P}$ is the numbers of cosets of $\mathcal{G}'^{\text{(ST)}}$, i.e., reduced by a factor of $|\mathcal{G}'^{\text{(ST)}}|$ compared with the physical error rate $P(e)$. In contrast to generic physical error rates that is determined by all Pauli eigenvalues in $\mathcal{P}_{n^{\text{(ST)}}}$, one only needs Pauli eigenvalues $\Lambda_{\mathcal{L}'}$ of the logical operators in $\mathcal{L}'^{\text{(ST)}}$ to determine $\overline{P}$. In fact, they are related by a Walsh-Hadamard transform $W$:
\begin{equation}
    \overline{P} \xleftrightarrow{\quad W \quad} \Lambda_{\mathcal{L}'}.\label{equation: logical_from_inverse_WH}
\end{equation}
See \cite{wagner2023learning} and \cref{appendix: circuit logical learning} for details. \textcite{wagner2023learning} showed the sufficient condition of the learnability of $\overline{P}$ in the channel coding setting.
\begin{theorem}[\textcite{wagner2023learning}]
     Let a Pauli error channel $\mathcal{N}_{\Gamma}$ be a composition of Pauli channels with support $\gamma\in\Gamma$ as defined in \cref{equation: channel composition}. The probabilities $\overline{P}$  of logical error classes are learnable from the syndrome data if any union $
     \gamma_1 \cup \gamma_2$ of two channel supports $\gamma_1, \gamma_2\in\Gamma$ only supports trivial logical operators.
    \label{theorem: wagner_logical}
\end{theorem}
We emphasize that the values of $\overline{P}(\vec s, \vec \ell)$ are decoder-independent quantities that depend only on the logical equivalence relation in the code, i.e., the gauge group and the physical error rates. To define the logical error rate $P_{\mathrm{fail}}$ of the logical measurements, we need a decoder. A decoder takes the syndrome $\vec{s}$ of the circuit as input and returns a logical class $\vec{\ell}_{\text{dec}}(\vec{s})$. Then, the probability of a logical error corresponds to the bit-flip $\vec f_L\in\mathbb{Z}_2^{m_L}$ of the ideal logical measurement outcome average over all possible syndromes is
\begin{equation}\label{equation: logical error rate}
    P_{\mathrm{fail}}\left(\vec f_L\right)=\sum_{\vec{s}} \overline{P}\left(\vec{s},\vec f_L \oplus \vec{\ell}_{\text{dec}}(\vec{s})\right).
\end{equation}
We show that syndrome expectation values in the circuit suffice to determine the logical error class probability, and hence the logical error rates, of fault-tolerant Clifford circuits (in the sense of \cref{definition: fault tolerant}) under the assumed circuit-level Pauli noise model.
\begin{definition}[Fault tolerance]
    Consider a circuit-level Pauli error model defined by a set of spacetime supports $\Gamma^{(\text{ST})}$ as in \cref{definition: Circuit-level Pauli error model}. We denote the set of individual spacetime Pauli errors of the model by $\mathcal{E}_{\Gamma^{(\text{ST})}}$. We say a non-adaptive Clifford circuit is fault tolerant to a circuit-level Pauli error model for a set of logical measurements $\mathcal{M}_L$ if and only if for any undetectable error $e$, either in $\mathcal{E}_{\Gamma^{(\text{ST})}}$ or as a product $e=e_1e_2$ for any $e_1,e_2\in\mathcal{E}_{\Gamma^{(\text{ST})}}$, the error $e$ on the circuit does not change the distribution of logical measurements outcome in $\mathcal{M}_L$ for any input codeword to the circuit.\label{definition: fault tolerant}
\end{definition}

\begin{theorem}[Learnability of logical errors]\label{theorem: logical error rate learnability}
    Given a non-adaptive Clifford circuit with a set of logical measurements $\mathcal{M}_L$ under a circuit-level Pauli noise model (satisfies \cref{assumption: open set probability} on the spacetime code), the logical error class probabilities $\overline{P}$ for $\mathcal{M}_L$ can be learned from the syndrome expectation value if and only if the circuit is fault-tolerant (\cref{definition: fault tolerant}) to the noise model for $\mathcal{M}_L$. As a consequence, given a decoder of the circuit, the associated logical error rates $P_{\mathrm{fail}}$ for the set of logical measurements $\mathcal{M}_L$ can be learned from the syndrome data if the circuit is fault-tolerant to the noise model for $\mathcal{M}_L$.
    \label{theorem: logical learning}
\end{theorem}
The proof using the spacetime code mapping is in \cref{appendix: circuit logical learning}. A naive way of learning the logical error rate would be to solve for $\Lambda_{\mathcal{L}'}$ from the syndrome expectation values and then performing the inverse Walsh-Hadamard transformation in \cref{equation: logical_from_inverse_WH} involving an exponentially large matrix. In practice, we first learn any set of physical error rates compatible with the syndrome expectations and then perform Monte Carlo simulations to estimate the logical error rate.

In addition to benchmarking logical measurements, estimating the fidelity of a prepared logical state is also of interest. We show that given the learned logical noise rates, we can estimate the logical fidelity \cite{hangleiter2024fault, krinner2022realizing, combes2017logical, bluvstein2024logical} of the noisy output state $\tilde{\rho}$ with the target output state $\rho=\proj{\overline{\psi}}$, which is assumed to be pure. We define the logical fidelity $F_L$ as the physical fidelity after correcting the output state $\tilde{\rho}$ back to the code space, see \cite{hangleiter2024fault} for a detailed discussion. We denote by $\mathcal{S}$ the stabilizer group of the output state $\rho$, and by $\mathcal{L}=N(\mathcal{S})$ the group of logical operators. We denote the set of representatives of each coset of $\mathcal{S}$ in $\mathcal{L}$ as $\hat{\mathcal{L}}$. In the case when the logical circuits are performing direct fidelity estimation by measuring logical operators in different bases \cite{flammia2011direct}, we have
\begin{align}
F_L&=\Tr\{R(\tilde{\rho})\rho\}\\
&=\frac{1}{2^n}\sum_{O\in\mathcal{P}_n}\langle O\rangle_{R(\tilde{\rho})}\langle O\rangle_{\rho}\\
&=\frac{1}{2^k}\sum_{L'\in\hat{\mathcal{L}}}\langle L'\rangle_{R(\tilde{\rho})}\langle L'\rangle_{\rho}
\label{equation: direct fidelity estimation}.
\end{align}
Here, $R$ is the CPTP map of syndrome measurement followed by error correction, and we denote by $\langle O\rangle_{\rho}$ the expectation value of $O$ with respect to the target output state $\rho$. We can learn the logical fidelity $F_L$ from the syndrome data of the family of circuits. This is because $P_{\mathrm{fail}}$ can be learned according to \cref{theorem: logical error rate learnability}, and $\langle L'\rangle_{R(\tilde{\rho})}$ depends on the Walsh-Hadamard transformation of $P_{\mathrm{fail}}$ and the ideal expectation value $\langle L'\rangle_{\rho}$ as
\begin{equation}
    \langle L'\rangle_{R(\tilde{\rho})}=\left(\sum_{\vec f_L \in\mathbb{Z}_2^{m_L}}(-1)^{\vec \alpha(L')\cdot \vec f_L}P_{\mathrm{fail}}(\vec f_L)\right)\langle L'\rangle_{\rho}.
\end{equation}
where $\vec \alpha(L')\in\mathbb{Z}_2^{m_L}$ is the binary vector whose nonzero entries indicate which measurement outcomes are used to compute $\langle L' \rangle$. When $\rho$ is a stabilizer state, denote its full stabilizer group as $\mathcal{S}_{\rho}>\mathcal{S}$ and denote by $\hat{\mathcal{S}}_{\rho}$ a set of representatives of the cosets of $\mathcal{S}$ in $\mathcal{S}_{\rho}$ then the logical fidelity reduces to $F_L=\frac{1}{2^k}\sum_{S\in\hat{\mathcal{S}}_{\rho}}\langle S\rangle_{R(\tilde{\rho})}$. In \cref{subsection: logical error rates in experiments}, we compared the learned logical fidelity from syndrome data with the experimentally measured fidelity (\cref{fig: logical circuit}) of the logical GHZ state from logical direct fidelity estimation.

When there is no sufficient set of logical measurement bases implemented at the end, syndrome constraints will not be sufficient to uniquely identify the logical fidelity. One can still estimate the logical fidelity by using additional constraints (\cref{equation: extra_trivial}) to first estimate all physical error rates and simulate the noisy circuit with a final round of perfect syndrome measurement and recovery. This is still expected to be accurate, particularly when several rounds of error correction are performed to accumulate sufficient syndrome data, except for the final round.

\begin{figure}[t]
\centering
\hspace{-0.8cm}
\begin{subfigure}{0.95\columnwidth}
    \phantomcaption
    \stackinset{l}{0pt}{t}{-5pt}{\captiontext*}
    {\includegraphics[width=\textwidth]{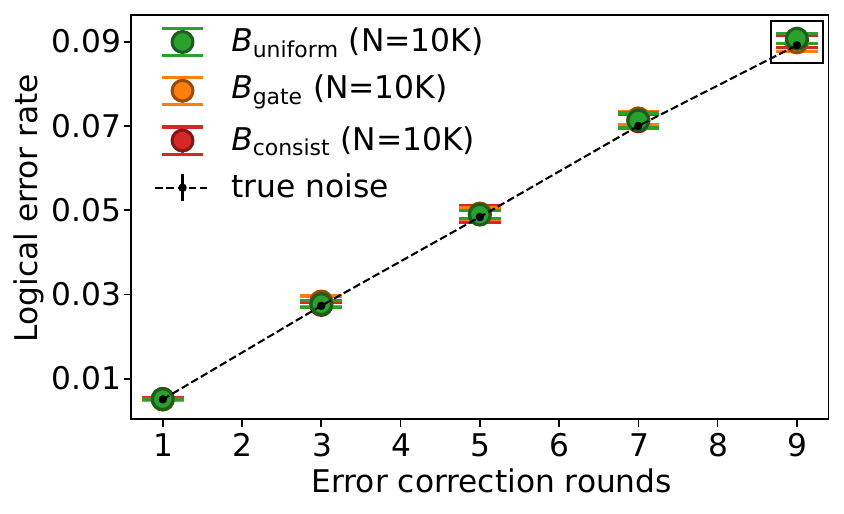}}
    \label{fig: surface_circuit_logical}
\end{subfigure}
\begin{subfigure}{0.95\columnwidth}
    \phantomcaption
    \hspace{-0.5cm}\stackinset{l}{0pt}{t}{-5pt}{\captiontext*}
    {\includegraphics[width=\textwidth]{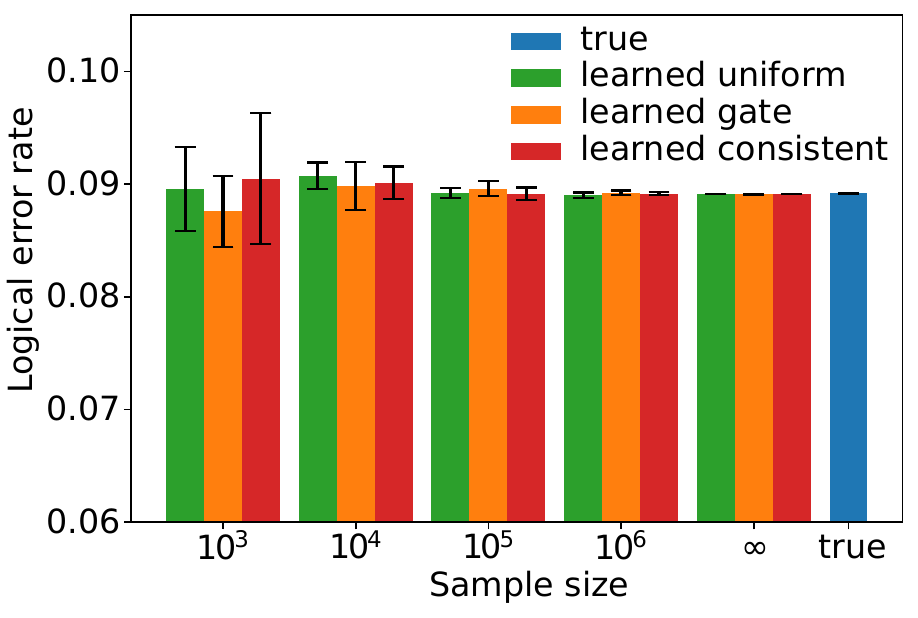}}
    \label{fig: r9_logical_sample_error}
\end{subfigure}
\caption{(a) Logical error rate estimation from simulation of surface code error correction circuits. We apply two-qubit Pauli error channels on two-qubit gates and single-qubit Pauli channels on other locations with near-threshold error rates fluctuating (20\%) around 0.01 and 0.001, respectively. Logical error rates are estimated by Monte Carlo sampling. (b) A zoomed-in plot of the learned logical error rates corresponds to 9 rounds of syndrome extraction (box in (a)) for different extra constraints $B$ and sample sizes. Learned logical error rates for different $B$ constraints match the true value.}
\label{fig: surface_logical}
\end{figure}

\begin{figure*}[t]
\centering
\hspace{-1cm}
\begin{subfigure}{0.31\textwidth}
    \phantomcaption
    \stackinset{l}{2pt}{t}{-10pt}{\captiontext*}
    {\hspace{0.25cm}\vspace{1cm}\includegraphics[width=\textwidth]{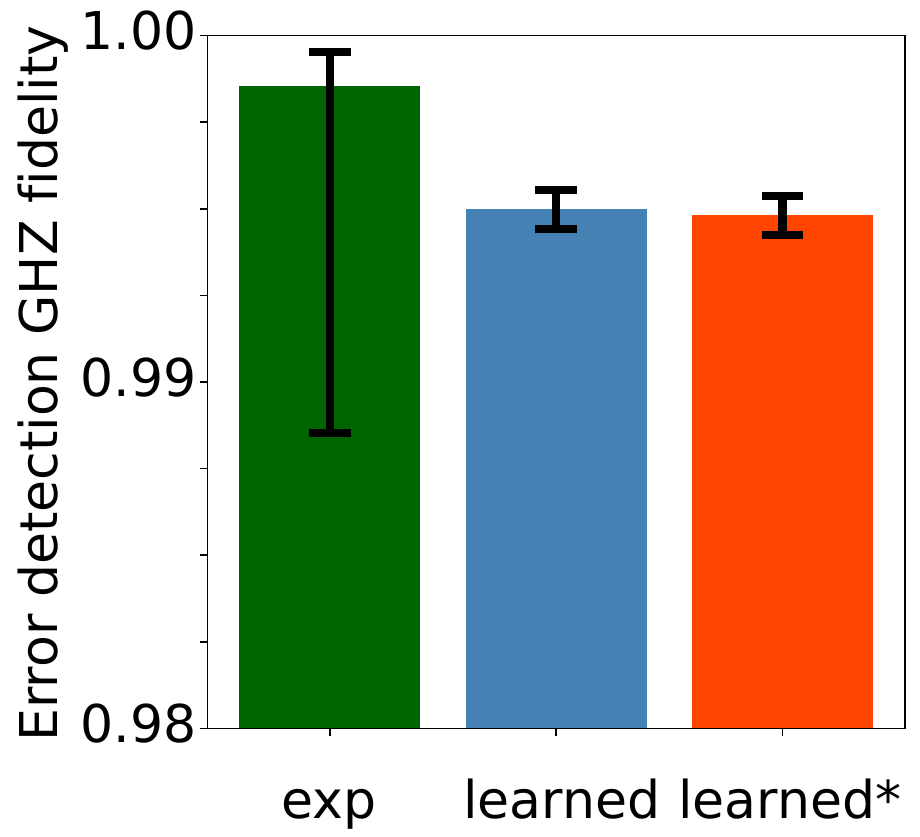}}
    \label{fig: post_select}
\end{subfigure}
\hspace{0.2cm}
\begin{subfigure}{0.29\textwidth}
\phantomcaption
\stackinset{l}{2pt}{t}{-5pt}{\captiontext*}
    {\hspace{0.25cm}\includegraphics[width=\textwidth]{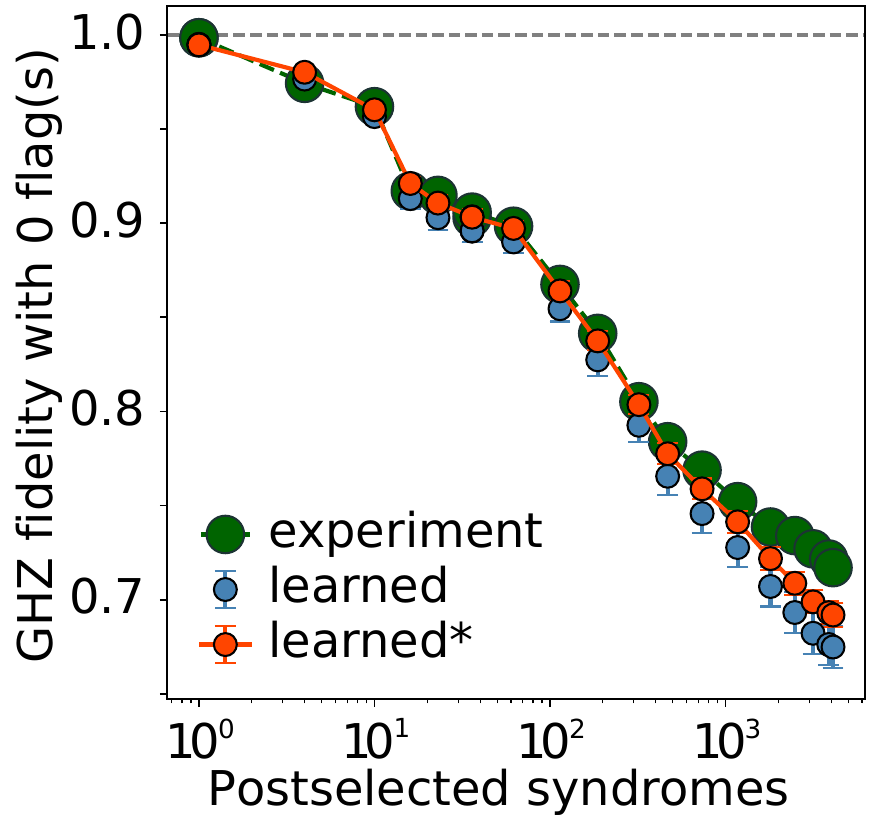}}
    \label{fig: flag_0}
\end{subfigure}
\hspace{0.2cm}
\begin{subfigure}{0.285\textwidth}
    \phantomcaption
    \stackinset{l}{2pt}{t}{-5pt}{\captiontext*}
    {\hspace{0.25cm}\includegraphics[width=\textwidth]{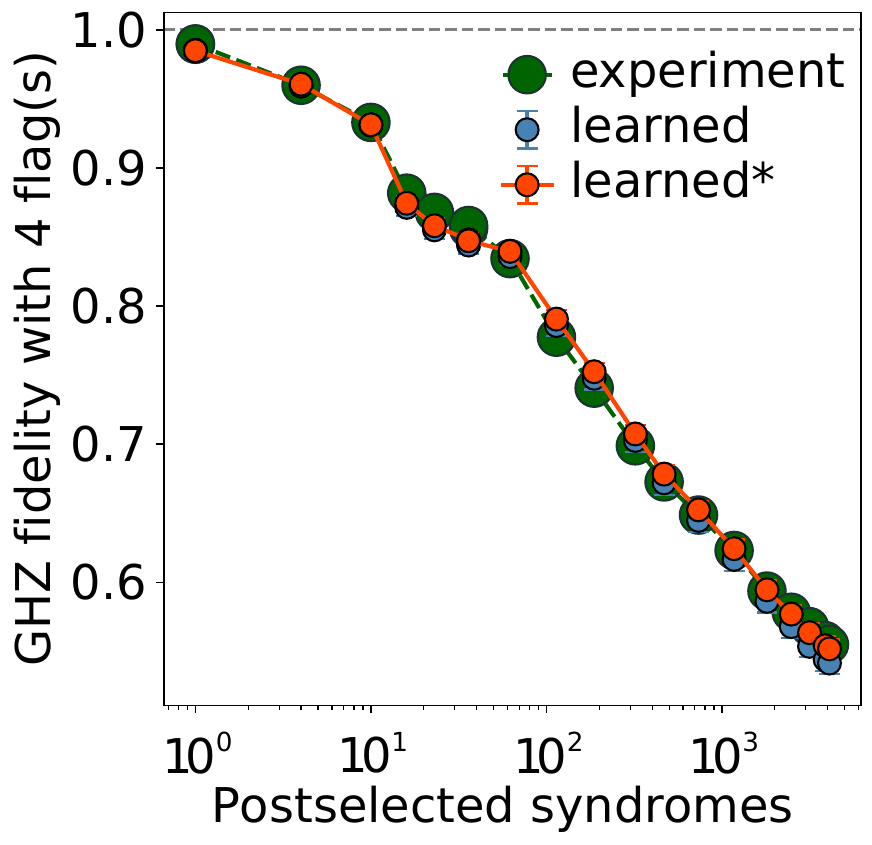}}
    \label{fig: flag_3}
\end{subfigure}
 \caption{Comparison of logical direct fidelity estimation between experimental results and learned noise models from syndromes. (a) Error-detection results for the logical GHZ fidelity, from experimental measurement, from the learned logical fidelity, and from the learned fidelity with extra intraclass uniformity constraints (denoted $^*$). (b),(c) The same comparison under partial postselection on the syndrome, conditioned on the number of ancilla blocks having a nontrivial syndrome (flag). See \cref{fig: logical fidelity learning (full)} for results with various numbers of flags.}
\label{fig: logical fidelity learning}
\end{figure*}

\subsection{Benchmarks}
We first test the logical error rate learning with the data from various simulated surface code syndrome extraction circuits in \cref{subsection: simulation}, ranging from a single round $r=1$ up to $r=9$ rounds of syndrome extractions. The logical error rate at the end of the 9 circuits is evaluated from the simulation of the syndrome measurement, decoding, and final correction of the whole circuit with Stim and Pymatching \cite{higgott2022pymatching,gidney2021stim}. The physical error rates are estimated using \cref{equation: optimization}, where we use different $B$ matrices for the extra constraints. We denote by  $B_{\text{gate}}$ the one that assumes a fixed ratio between pairs of errors in each syndrome error class consistent with the error rates without added Gaussian fluctuation. The other matrix asserts uniform noise within each syndrome error class, denoted as $B_{\text{uniform}}$. Finally, we use $B_{\text{consist}}$, which asserts constraints consistent with the true noise. As shown in \cref{fig: surface_logical}, the learned circuit noise predicts the logical error rate well with uncertainty $\epsilon_{P_{\mathrm{fail}}}\leq\pm 3 \cdot 10^{-4}$ for a single round syndrome extraction, and $\epsilon_{P_{\mathrm{fail}}}\leq\pm 2.2 \cdot 10^{-3}$ for nine rounds of syndrome extractions.
Importantly, notice that the logical failure rate is independent of the extra constraints used to fix the error rates of individual Pauli errors in the same syndrome classes. 
%\Dom{I still don't get the setting --- when are you correcting?}

\subsection{Logical error rates in experiments\label{subsection: logical error rates in experiments}}
We now use our algorithm to learn the logical error rates from the syndrome data of the experiments reported in Ref.~\cite{bluvstein2024logical}. We estimate the logical fidelity from the learned physical error rates (using \cref{equation: optimization}) of all 9 logical GHZ preparation circuits with different logical measurement bases shown in \cref{fig: logical circuit}, \cref{section: physical_learning_GHZ}. We perform Monte Carlo simulations on stim \cite{gidney2021stim} to sample measurement bitstrings from our estimated physical error rates of all the circuits.
%\Dom{language}
The simulated logical measurement results, corrected using the correlated decoding technique \cite{bluvstein2024logical,cain2024correlated}, are extracted. We note that these simulated logical measurement results encode the learned logical error rates, and the logical bitstrings from the experimental data are never used. In addition, we implemented ``sliding-scale error detection'', where we rank data block syndromes based on the probability of their most likely error and perform partial postselection based on a cutoff threshold, see \cite{bluvstein2024logical} for details. To mitigate shot noise from the syndrome expectation value due to the small sample size, we also implemented the optimization in \cref{equation: optimization} with added intraclass uniformity constraints (\cref{appendix: Intraclass uniformity}) denoted by '$^*$' in \cref{fig: flag_0,fig: flag_3}.

Comparing our learned GHZ state logical fidelity from the circuit syndrome with the experimental result obtained from direct fidelity estimation, we find them in good agreement, as shown in \cref{fig: logical fidelity learning}. For postselection on the trivial syndrome, the estimated logical fidelity is within the experimental error bars. The protocol also gives a good estimation in the case of error correction with partial postselection, as most experimental results fall in the range error bar of the learned data, and we see better agreement for the learned data with extra uniformity constraints. We speculate that the deviation from the experimental result is due to shot noise from the limited sample size and our simplified noise model. Nevertheless, the overall good agreement shows that the local Pauli error model is a valid error model in the setting of error-corrected logical quantum information processing for these experiments.

\subsection{Exponential sample complexity advantage at very low logical error rates}

\begin{figure}[t]
\centering
\hspace{-0.3cm}
\begin{subfigure}{0.25\textwidth}
    \phantomcaption
    \stackinset{l}{10pt}{t}{-5pt}{\captiontext*}
    {\includegraphics[width=\textwidth]{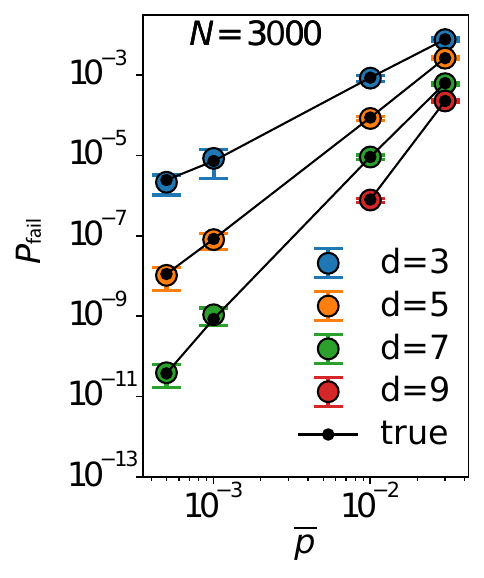}}
    \label{fig: logical_error_N3000}
\end{subfigure}
%\hspace{-0.2cm}
\begin{subfigure}{0.235\textwidth}
    \phantomcaption
    \stackinset{l}{0pt}{t}{-5pt}{\captiontext*}
    {\includegraphics[width=\textwidth]{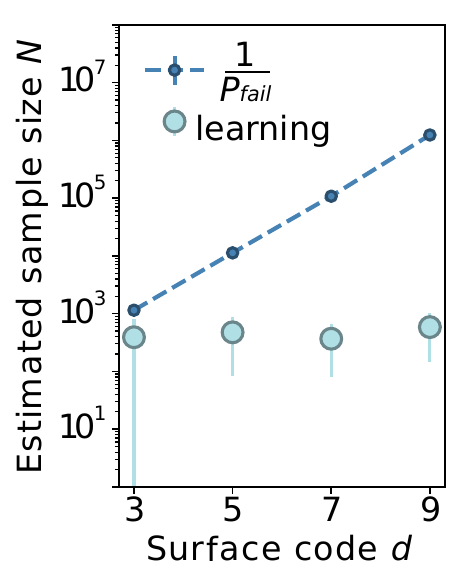}}
    \label{fig: logical_sample_surface}
\end{subfigure}
\label{fig: sample_advantange_logical}
\caption{(a) Rotated surface code logical error rate $P_{\mathrm{fail}}$ initialized in $\ket{\overline 0}$ for various distances $d$ and physical single-qubit error rates in the code capacity setting (average error rate $\overline p$ with $\sigma=0.3\,\overline p$). %\Dom{What are the values of $\mu$ you chose? Is $\mu = \overline p$?}
The black markers are the true logical error rates from the minimum-weight perfect-matching decoder. The colored markers are the learned values from $N = 3000$ samples of syndrome data. (b) Estimated sample size needed for direct fidelity estimation at the logical level (dark blue) versus the one needed (from interpolation of results shown in (a) with different $N$ at average physical error rate $10^{-2}$) for the logical error rate learning algorithm (light blue) if one requires $\eta=1$ relative error.}
\end{figure}

The threshold theorem of fault-tolerant quantum computation \cite{dennis2002topological,fowler2012proof,fowler2012surface} shows that there exist families of codes with growing distance $d$, such that if the physical error rate $p$ is below some threshold $p_{\text{th}}$ then one can exponentially suppress the logical error rate as a function of the distance of the code. Under the uniform single-qubit error model assumption and in the asymptotically low physical error rate regime, the logical error rate follows the formula \cite{bravyi2013simulation, fowler2012analytic}:
\begin{equation}
    P_{\mathrm{fail}}= A_d \left(\frac{p}{p_{\text{th}}}\right)^{\frac{d+1}{2}},\label{equation:  logical scaling}
\end{equation}
where $A_d=\mathrm{poly}(d)$ is some constant independent of $p$. 
When estimating such exponentially suppressed logical error rates directly at the logical level, for example, using direct fidelity estimation, an exponential number of samples is required to resolve the signal. Quantitatively, sampling from a binomial distribution with probability $p_{\text{bin}}$, the sample size $N$ needed to obtaining an estimation of $p_{\text{bin}}$ with relative error $\eta$ is given by
\begin{equation}
    N=\frac{1-p_{\text{bin}}}{\eta^2 p_{\text{bin}}}\approx\frac{1}{\eta^2 p_{\text{bin}}}.\label{equation: binomial scaling}
\end{equation}
Since multiplicative errors scale linearly under exponentiation, the asymptotic formula in \cref{equation:  logical scaling} suggests that achieving a logical-level relative error within $\eta_L$ only requires improving the physical-level multiplicative precision by a factor proportional to the code distance. This is much more desirable since the physical error rates typically do not decrease with increasing number of qubits, and we have
\begin{equation}
    N_{\text{physical}}\approx \frac{d^2}{\eta_L^2 p}.
\end{equation}
%\Dom{Why is this just a heuristic relation? And how does it avoid an exponential sample complexity? I am exponentiating an approximation by $d$, right?}\Xiao{the physical error rate is a constant so $1/p_{physical}$ does not scale with n}

In this section, we make the statement precise for Clifford circuits under the local sparse Pauli error model described in \cref{section: model}.
\begin{theorem}[Sample complexity of logical error estimation\label{theorem: logical sample complexity}]
    Consider a non-adaptive logical Clifford circuit with $n$ physical qubits and of physical depth $T$ that corresponds to a qLDPC spacetime code. We assume the circuit is fault-tolerant to a circuit-level noise that corresponds to a local sparse Pauli channel $\mathcal{N}_{\Gamma}$ on the spacetime code, for a set of logical measurements $\mathcal{M}_L$. Denote the vector of true error rates by $\vec{p}$ and the minimal syndrome class error rate by $P_{C\mathrm{min}}$. Then we can learn the logical error rates $P_{\mathrm{fail}}$ for $\mathcal{M}_L$ with relative error within $\eta_L$ using a sample size
    \begin{equation}
        N=\tilde{\mathcal{O}}\left(\frac{n^2T^2}{\eta_L^2P^2_{C\mathrm{min}}}\right),
    \end{equation}
    if the positivity condition in \cref{lemma: existence of solution} is satisfied and
    \begin{equation}\label{equation: p upper bound}
    \|\vec{p}\|_{\infty}\leq \min\left(\frac{2}{L_H K_{G_B}},\quad\left(\frac{\eta_L P_{C\mathrm{min}}}{4|\Gamma|C_c}\right)^{\frac{1}{2}}\right),
    \end{equation}
       for constants $L_H, K_{G_B}, C_c$ defined in \cref{equation: define K_{G_B},equation: define L_H,equation: define C_c}.
\end{theorem}
See \cref{subsection: logical sample complexity} for details of the proof. Again, applying the Chernoff bound gives a more precise scaling $N=\tilde{\mathcal{O}}(\frac{p_{\mathrm{max}}}{P^2_{C\mathrm{min}}})$, where $p_{\mathrm{max}}=\|\vec p\|_{\infty}$, with respect to the physical error rates. It reduces to $\mathcal{O}(1/p)$ under the uniform error rate case as expected.

The logical error rate estimation method considered in the proof makes use of the solution in
\cref{equation: recursive solving} and approximates the syndrome-class error rate to first order via $P_C = -\tfrac{1}{2}\log \nu_C + \mathcal{O}(\|\vec p\|^2)$.
The true error rates $\vec p$ requires a positivity condition (\cref{lemma: existence of solution}) that can be satisfied, for example, if there exists some $p_{\mathrm{max}}=O(1)$ such that $p_{\mathrm{max}}^2\lesssim p \leq p_{\mathrm{max}}$ for all entries $p$ of $\vec p$. Since better physical error rate learning precision is required as spacetime volume increases in the proof, we require the upper bound on $\|\vec p\|_\infty$ to decrease as $\sim 1/(nT)$. This is reflected in the number of local channels $|\Gamma|$ in the denominator of the upper bound in \cref{equation: p upper bound}.
%\Dom{What is the $\Gamma$ factor? Which upper bound?}
Nevertheless, for a fixed sufficiently low physical error rate, the theorem still yields an
exponential sampling advantage for finite-size systems. This first-order approximation can be avoided by using optimization methods to solve \cref{equation: optimization}. Therefore, we believe the sampling advantage holds more generally in this case.

In the regime where the physical error rate is well below the fault tolerance threshold, the sample complexity scaling approaches $\mathcal{O}(d^2)$ instead of $\mathcal{O}(n^2T^2)$. This scaling arises because the fault patterns are dominated by the minimum weight uncorrectable errors in the circuit, which require $\mathcal{O}(d)$ independent error events happening in the set of local errors $\mathcal{E}_{\Gamma}$.

\begin{figure}[t]
\centering
\hspace{-0.3cm}
\begin{subfigure}{0.27\textwidth}
    \phantomcaption
    \stackinset{l}{10pt}{t}{-5pt}{\captiontext*}
    {\includegraphics[width=\textwidth]{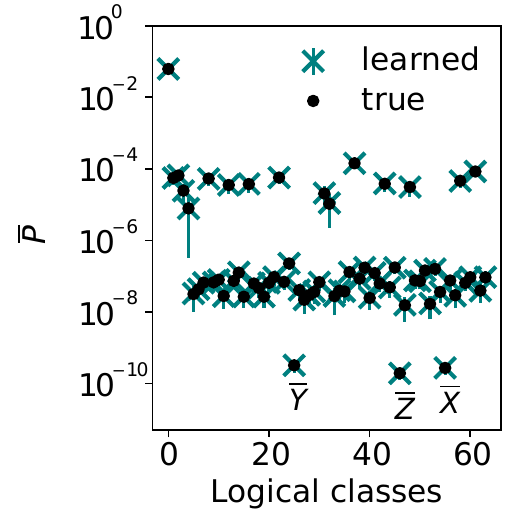}}
    \label{fig: five_qubit_logical_class_panel}
\end{subfigure}
%\hspace{-0.2cm}
\begin{subfigure}{0.2\textwidth}
    \phantomcaption
    \stackinset{l}{0pt}{t}{-5pt}{\captiontext*}
    {\includegraphics[width=\textwidth]{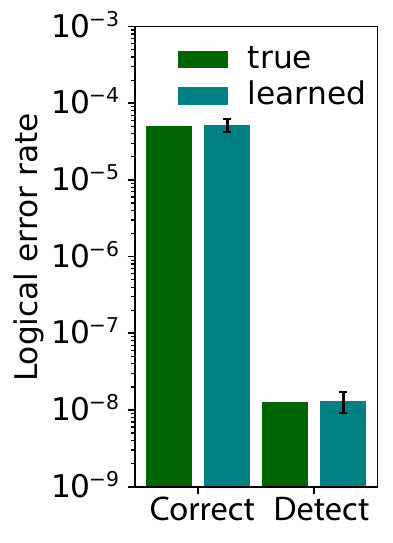}}
    \label{fig: five_qubit_correct_detect}
\end{subfigure}
    \caption{(a) Learned $64$ logical class error rate of the $[[5,1,3]]$ code. Each true single-qubit Pauli error rate is sampled from $p\sim\mathcal{N}(\mu,\sigma^2)$, where $\mu=10^{-3},\sigma=5\times10^{-4}$. With $N=8000$ samples of syndrome data, every logical class error rate can be estimated close to the true value in multiplicative precision. (b) The logical error rate of both error correction and error detection in (a) can be estimated to high precision.}
    \label{fig: five_qubit_logical_class}
\end{figure}
To test if the advantage holds in practice, we perform numerical simulations to study the accuracy of logical error rate learning on rotated surface codes of distance $d=3,5,7,9$ in the code capacity setting. We sample single-qubit error rates from a gaussian distribution $\mathcal{N}\sim(\mu=10^{-3},\sigma=\frac{1}{2}\times10^{-3}) $. As shown in \cref{fig: sample_advantange_logical}, with $N=3000$ samples, one can estimate the logical error rate with good relative precision at the order of $10^{-11}$ if the code/circuit is fault-tolerant to the noise model. We simulate learning error rates with different sample sizes and different surface code distances and perform a linear fit of the multiplicative precision with the sample size for each code on the log scale. In \cref{fig: logical_sample_surface}, we show the estimated samples needed for relative error $\eta=1$ and see a clear advantage compared to $1/P_{\mathrm{fail}}$ as in \cref{equation: binomial scaling}.

Finally, learning the probability of different logical error classes corresponding to different Pauli errors with a certain syndrome is also of interest. This is because post-selection on syndromes is also widely used in many protocols and is a general means of improving logical fidelity in the near term~\cite{bluvstein2024logical,bluvstein2025architectural,gidney2024magic,Rosenfeld2025magic,sales2025experimental}.
All logical error class probabilities can be learned with good relative precision since they can all be written as a polynomial in the physical error rates, similarly to the logical error rate averaged over syndromes. We observed from numerical simulations, as shown in \cref{fig: five_qubit_logical_class}, that all $64$ logical class error rates $\overline{P}$ of the $[[5,1,3]]$ of different orders of magnitude can be learned with $N=8000$ samples to good relative precision. This allows one to estimate the logical error rate conditioned on any subset of syndromes one postselects on, using syndrome data from a small number of samples.

\subsection{Correlated noise from rare events}

In this section, we discuss scenarios that violate the model in \cref{section: model} or the learnability condition in \cref{theorem: logical error rate learnability}. With some additional assumptions, we discuss potential methods for estimating both practically relevant and theoretical models.

In practice, one such violation arises from events that cause correlated noise across a large region. Examples of experimentally observed events are ionizing radiation caused by cosmic rays \cite{vepsalainen2020impact,wilen2021correlated,mcewen2022resolving,harrington2025synchronous,li2025cosmic,larson2025quasiparticle,binney2026distinguishing} and mechanical vibrations~\cite{binney2026distinguishing,kono2024mechanically}. It is of high interest and an active area of research to prevent such events from happening, as the elimination of these highly correlated events is necessary for reliable fault-tolerant computation. In the case of cosmic rays, for example, several groups have developed engineering/mitigation strategies for superconducting qubits \cite{cardani2021reducing,iaia2022phonon,bratrud2025measurement,mcewen2024resisting,mcjunkin2026demand,pinckney2026characterization,kurilovich2025correlated}.

The error mechanism of high-energy radiation is commonly attributed to producing phonons and quasiparticles that cause $T_1$ decay across many qubits in the chip~\cite{vepsalainen2020impact,wilen2021correlated,martinis2021saving}. Conditioned upon such an event happening, after Pauli twirling or randomized compiling, the resulting error channel can still be described by a local-sparse Pauli channel, albeit with a much higher error rate over an extended period of time \cite{baity2026characterizing}. Without any postselection of such events, the error channel will be a classical mixture of low error rate and high error rate local-sparse channels, thus falling outside of the model assumption in \cref{section: model}. However, such events are easily detected, and even characterized in real time from the time series of syndrome histories by looking at the syndrome expectation in a local time window average over large regions of the system \cite{google2021exponential,google2023suppressing,tan2024resilience,nayak2026iterative,vallero2026radiation,poster2026reloqate}.

Another scenario, motivated more theoretically, is the violation of the condition for learnability of the logical error rate: the correlated Pauli error with a weight larger than $\lfloor \frac{d-1}{2} \rfloor$, where $d$ is the code distance. Here, we consider a simplified scenario where a depolarizing channel $\mathcal{N}_{R}(\rho)=(1-p_{R})\rho + \frac{p_{R}}{4^{|R|}-1}\sum_{e\in\mathcal{P}_{R}} e\rho e^{\dagger}$ , concataenated with a local-sparse Pauli channel, acts on a subset $R\subseteq [n]$ of an $n$-qubit qLDPC code. To break the condition in \cref{theorem: logical error rate learnability}, we need $|R|=\Omega(d)$.  We define $\mathcal{G}_{R}$ as the set of stabilizer generators that overlap with the region $R$, and define the measurement result $X_g\in\{+1, -1\}$ for $g\in\mathcal{G}_{R}$, then conditioned on the event that a nontrivial error happens in $\mathcal{N}_R$ we have
\begin{equation}
    E(X_g)=-e^{-\Omega(|R|)}\approx 0 ,\quad \forall g\in\mathcal{G}_{R},
\end{equation}
which deviates significantly from the noiseless value of 1.
Therefore, detection is possible by calculating the average syndrome expectation $M=\frac{1}{m}\sum_{g}^{\mathcal{G}} X_g$. We show in \cref{appendix: single-shot detection} that the signal-to-noise ratio for such detection increases as we increase the system size if $|R|=\omega(\sqrt{n})$.

Furthermore, if the region $R$ contains a logical operator, we show that $\mathcal{N}_R$ has a logical-error contribution of order
$p_R$, and one can bound the logical error rate given the ability to detect such high-weight Pauli errors and the ability to learn the logical error rate $P^{(bg)}_{\mathrm{fail}}$ due to the background local Pauli channel (\cref{theorem: logical error rate learnability}) by postselecting syndrome data.
\begin{lemma}\label{lemma: logical error rate bound}
     If the support $R$ of the depolarizing channel $\mathcal{N}_R$ has total error rate $p_R$ and contains some (dressed) logical operator $L$, then up to some expontitally small correction in $|R|$, the logical error rate is bounded by:
\begin{equation}
    (1-p_R)P^{(bg)}_{\mathrm{fail}} + \frac{1}{2} p_R \leq P_{\mathrm{fail}}\leq(1-p_R)P^{(bg)}_{\mathrm{fail}} + p_R.
\end{equation}   
\end{lemma}
See the proof in \cref{appendix: lower bound logical error rate regional depolarizing}.
As a result, the errors from $\mathcal{N}_R$ can be reliably detected for $|R|=\omega(\sqrt{n})$ and in this case the logical error rate can either be learned when $d=\omega(\sqrt{n})$ and $|R|=\mathcal{O}(\sqrt{n})$, or bounded when $|R|=\omega(\sqrt{n})$ according to \cref{lemma: logical error rate bound}.

On the other hand, for a region $|R|$ not containing any logical operator, the logical error rate would highly depend on the decoding methods and the potential use of erasure decoding. We leave the details of this case for future work. Finally, the two types of correlated events can in principle be distinguished by examining the correlation among the syndrome data \cite{remm2025experimentally}.

\section{Discussion}

In this paper, we study the problem of learning Pauli noise in fault-tolerant Clifford circuits from syndrome data.  To this end, we analyze the noise-learning problem for general subsystem codes and map Clifford circuits to spacetime (subsystem) codes. We rigorously show that the total error rates of every equivalence class of circuit-level Pauli noise, i.e., the syndrome class where errors have the same syndrome, can be learned from the syndrome data of a fault-tolerant Clifford circuit. In fact, we are able to prove that local sparse Pauli channels in the circuit can be learned from a constant number of samples when the spacetime code of the circuit is qLDPC. 

Complementing these analytical results, we design a general algorithm to estimate circuit-level Pauli error rates of Clifford circuits given syndrome data. We do so by turning the noise learning problem into a least square optimization on a system of noisy non-linear equations (without approximation)/linear equations (with approximation), which allows us to adaptively adjust the redundancy as well as terminate early to reduce overfitting.
To improve the accuracy of our learning algorithm, one could design better optimization schemes and potentially dynamically adjust the effective degrees of freedom in the noise model to further reduce overfitting from shot noise.

%\Dom{I thought we concluded that we are not actually competitive with randomized benchmarking. It would be helpful to have that discussion here. }
Our work provides a complementary method for benchmarking gate errors compared to randomized benchmarking, gate-set tomography, and cycle error reconstruction \cite{knill2008randomized,merkel2013self,nielsen2021gate,carignan2023error,fazio2026characterizing}. These methods effectively magnify physical error rates by increasing the number of gate layers. In the context of our work, this is essentially trading measurements with gates, thus holding an advantage if measurements are slower than gate implementations. However, compared to these methods, our approach requires no pre-calibration and provides in situ estimates of the gate error rate for the target computational task. In addition, our method is able to capture non-Markovian or time-dependent effects such as drift, while offline benchmarking methods typically make a Markovian or time-independent assumption about the noise.
%\Dom{But why did we not recover the randomized benchmarking performance? Because of the intermediate measurements, right? }\Xiao{I think with uniform assumption then we do recover roughly the same accuracy with the same sample size, just that experimentally takes more time due to our method trades gates with measurements}

We further show that the logical error rate of the circuit is determined by the syndrome data if the circuit is fault-tolerant to the noise model we are learning. This allows us to learn logical error rates and estimate the logical fidelity of circuits without information from a logical measurement. This includes logical circuits with Clifford gates encoded in any stabilizer code, subsystem codes, and dynamical codes containing rounds of fault-tolerant syndrome extraction. 
%\Dom{The following discussion belongs in the previous section. } We do note that if one does not end with logical measurements, then there will typically be some undetectable errors at the last few layers of the last round of syndrome extraction. In this case, we could still learn the logical error rate up to the start of the last round from the syndrome and get a good estimate for the whole circuit up to this boundary effect. 
%As our title result,
Notably, we show that our method of extracting logical error rates from syndrome information provides an exponential sample advantage compared with tomography/fidelity estimation at the logical level only. This becomes important when logical error rates are suppressed by larger code distances. We do so using numerics and give a proof in the regime of low error rates. We note that an important problem remains unsolved: developing efficient classical algorithms to estimate the logical error rate of larger systems given the learned physical error rates. A series of work based on the splitting method and Metropolis sampling algorithms can potentially be such a complementary tool to our work \cite{bravyi2013simulation,mayer2025rare,beverland2025fail}.

Many works that estimate error rates from syndrome data use the detector error model (DEM) \cite{spitz2018adaptive,remm2025experimentally,takou2025estimating,blume2025estimating}, typically in the decoding community. It is then natural to ask what the relationship is between the detector error model and the Pauli channel error model in our work. The use of the Pauli channel arises naturally when randomized compiling is used for circuit execution \cite{wallman2016noise,hashim2020randomized}. Detector error model, being a concatenation of Bernoulli Pauli errors, is a subset of the general Pauli channel. Note that our sufficient and necessary condition of the learnability results \cref{theorem: syndrome class learning} and \cref{theorem: logical error rate learnability} applies to the generic Pauli channel and therefore subsumes the detector error model. To make both the sample and computational costs scalable, we introduce the local and sparse condition for the Pauli channel, obtained by concatenating individual channels. On the other hand, the detector error model also needs to satisfy certain sparsity conditions. We leave the detailed comparison and inclusion relation to the sparse DEM model, e.g., used in \cite{blume2025estimating} in the efficient case, to future work.

Our work opens up several directions for future research. Learning circuit-level noise directly benefits decoding since one could obtain a more accurate prior for the decoder. The improvement for the optimal decoder will be optimal within the learned noise model, since we uniquely determine the total Pauli error rates for each logical error class. On the other hand, one could develop a framework beyond the standard Pauli noise model to account for more realistic and specific noise models, such as coherent errors and crosstalk, where syndrome data may not provide enough information but will still be an important input to the learning protocol. Moreover, we point out the possibility of estimating logical eror rate under a toy model of regional depolarizing noise outside the code's correctability which violate the condition for \cref{theorem: logical error rate learnability}. We note that several open questions remain. One important regime is the borderline case in which the affected region has size $\mathcal{O}(\sqrt{n})$, where the simple syndrome-average detection criterion is not sufficient to give a clear asymptotic separation. Another direction is to relax the regional depolarizing-noise assumption and develop guarantees for more general correlated noise models, possibly with analyses tailored to specific code families.

%\Dom{Should we say that we already have an upcoming paper on that? Otherwise, maybe we shouldn't give the explicit ideas already? }
An exciting future direction is to go beyond non-adaptive Clifford circuits and estimate physical and logical errors in circuits with quantum advantage or circuits enabling universal quantum computation. In our upcoming work, Part II, \cite{XiaoWorkInProgress2026}, we progress towards this goal by adapting the noise-learning algorithm to verify logical Clifford circuits with magic state injection and extending Pauli noise learning to non-Clifford circuits. Investigating these future directions would potentially lead to a more comprehensive noise learning algorithm and an efficient verification of quantum advantage in the context of encoded logical circuits.

Note: After finalizing this work, we became aware of Ref.~\cite{ZhengLogicalNoiseWorkInProgress2026}, where the authors study learning the logical error rate from syndrome data.
\begin{acknowledgments}
We thank Steven Flammia, Martin Kliesch, and Scott Glancy for helpful discussions. 
This research was supported in part by NSF QLCI grant OMA-2120757, the Defense Advanced Research Projects Agency (DARPA) under Agreement No. HR00112490357, and ARO grant W911NF-23-1-0258.
DH is grateful for support from a Simons postdoctoral fellowship through DOE QSA and NSF QLCI Grant No.\ 2016245, and from the Swiss National Science Foundation through Ambizione Grant No.\ 223764.
We acknowledge the use of large language models to improve clarity and grammar of original text.
\end{acknowledgments}

\section*{Data availability}
The code to produce the findings of this work is publicly available in Ref.~\cite{xiao2026code}.

%\nocite{*}
\bibliography{references}% Produces the bibliography via BibTeX.

@article{bluvstein2024logical,
  title={Logical quantum processor based on reconfigurable atom arrays},
  author={Bluvstein, Dolev and Evered, Simon J and Geim, Alexandra A and Li, Sophie H and Zhou, Hengyun and Manovitz, Tom and Ebadi, Sepehr and Cain, Madelyn and Kalinowski, Marcin and Hangleiter, Dominik and others},
  journal={Nature},
  volume={626},
  number={7997},
  pages={58--65},
  year={2024},
  publisher={Nature Publishing Group UK London}
}

@inproceedings{flammia_averaged_2022,
  title={Averaged Circuit Eigenvalue Sampling},
  author={Flammia, Steven T.},
  booktitle={17th Conference on the Theory of Quantum Computation, Communication and Cryptography (TQC 2022)},
  series={Leibniz International Proceedings in Informatics (LIPIcs)},
  volume={232},
  editor={Le Gall, Fran\c{c}ois and Morimae, Tomoyuki},
  publisher={Schloss Dagstuhl -- Leibniz-Zentrum f{\"u}r Informatik},
  address={Dagstuhl, Germany},
  pages={4:1--4:10},
  year={2022},
  doi={10.4230/LIPIcs.TQC.2022.4},
  url={https://drops.dagstuhl.de/entities/document/10.4230/LIPIcs.TQC.2022.4},
  eprint={2108.05803},
  archivePrefix={arXiv}
}

@article{acharya2024quantum,
  title={Quantum error correction below the surface code threshold},
  author={Acharya, Rajeev and Aghababaie-Beni, Laleh and Aleiner, Igor and Andersen, Trond I and Ansmann, Markus and Arute, Frank and Arya, Kunal and Asfaw, Abraham and Astrakhantsev, Nikita and Atalaya, Juan and others},
  journal={Nature},
  year={2025},
  volume={638},
  pages={920--926},
  doi={10.1038/s41586-024-08449-y}
}

@article{huo2017learning,
  title={Learning time-dependent noise to reduce logical errors: real time error rate estimation in quantum error correction},
  author={Huo, Ming-Xia and Li, Ying},
  journal={New Journal of Physics},
  volume={19},
  number={12},
  pages={123032},
  year={2017},
  publisher={IOP Publishing}
}

@article{fowler2014scalable,
  title={Scalable extraction of error models from the output of error detection circuits},
  author={Fowler, Austin G and Sank, D and Kelly, J and Barends, R and Martinis, John M},
  journal={arXiv preprint arXiv:1405.1454},
  year={2014}
}

@article{spitz2018adaptive,
  title={Adaptive Weight Estimator for Quantum Error Correction in a Time-Dependent Environment},
  author={Spitz, Stephen T and Tarasinski, Brian and Beenakker, Carlo WJ and O'Brien, Thomas E},
  journal={Advanced Quantum Technologies},
  volume={1},
  number={1},
  pages={1800012},
  year={2018},
  publisher={Wiley Online Library}
}

@article{wagner2022pauli,
  title={Pauli channels can be estimated from syndrome measurements in quantum error correction},
  author={Wagner, Thomas and Kampermann, Hermann and Bru{\ss}, Dagmar and Kliesch, Martin},
  journal={Quantum},
  volume={6},
  pages={809},
  year={2022},
  publisher={Verein zur F{\"o}rderung des Open Access Publizierens in den Quantenwissenschaften}
}

@article{wagner2023learning,
  title={Learning logical Pauli noise in quantum error correction},
  author={Wagner, Thomas and Kampermann, Hermann and Bru{\ss}, Dagmar and Kliesch, Martin},
  journal={Physical review letters},
  volume={130},
  number={20},
  pages={200601},
  year={2023},
  publisher={APS}
}

@article{bacon2017sparse,
  title={Sparse quantum codes from quantum circuits},
  author={Bacon, Dave and Flammia, Steven T and Harrow, Aram W and Shi, Jonathan},
  journal={IEEE Transactions on Information Theory},
  volume={63},
  number={4},
  pages={2464--2479},
  year={2017},
  publisher={IEEE}
}

@article{gottesman2022opportunities,
  title={Opportunities and challenges in fault-tolerant quantum computation},
  author={Gottesman, Daniel},
  journal={arXiv preprint arXiv:2210.15844},
  year={2022}
}

@article{delfosse2023spacetime,
  title={Spacetime codes of Clifford circuits},
  author={Delfosse, Nicolas and Paetznick, Adam},
  journal={arXiv preprint arXiv:2304.05943},
  year={2023}
}

@article{dennis2002topological,
  title={Topological quantum memory},
  author={Dennis, Eric and Kitaev, Alexei and Landahl, Andrew and Preskill, John},
  journal={Journal of Mathematical Physics},
  volume={43},
  number={9},
  pages={4452--4505},
  year={2002},
  publisher={American Institute of Physics}
}

@article{fowler2009high,
  title={High-threshold universal quantum computation on the surface code},
  author={Fowler, Austin G and Stephens, Ashley M and Groszkowski, Peter},
  journal={Physical Review A—Atomic, Molecular, and Optical Physics},
  volume={80},
  number={5},
  pages={052312},
  year={2009},
  publisher={APS}
}

@article{bravyi2024high,
  title={High-threshold and low-overhead fault-tolerant quantum memory},
  author={Bravyi, Sergey and Cross, Andrew W and Gambetta, Jay M and Maslov, Dmitri and Rall, Patrick and Yoder, Theodore J},
  journal={Nature},
  volume={627},
  number={8005},
  pages={778--782},
  year={2024},
  publisher={Nature Publishing Group UK London}
}

@article{mao2005factor,
  title={On factor graphs and the Fourier transform},
  author={Mao, Yongyi and Kschischang, Frank R},
  journal={IEEE Transactions on Information Theory},
  volume={51},
  number={5},
  pages={1635--1649},
  year={2005},
  publisher={IEEE}
}

@article{flammia2020efficient,
  title={Efficient estimation of Pauli channels},
  author={Flammia, Steven T and Wallman, Joel J},
  journal={ACM Transactions on Quantum Computing},
  volume={1},
  number={1},
  pages={1--32},
  year={2020},
  publisher={ACM New York, NY, USA}
}

@article{abbeel2006learning,
  title={Learning factor graphs in polynomial time and sample complexity},
  author={Abbeel, Pieter and Koller, Daphne and Ng, Andrew Y},
  journal={The Journal of Machine Learning Research},
  volume={7},
  pages={1743--1788},
  year={2006},
  publisher={JMLR. org},
  url={https://www.jmlr.org/papers/v7/abbeel06a.html}
}

@article{horsman2012surface,
  title={Surface code quantum computing by lattice surgery},
  author={Horsman, Dominic and Fowler, Austin G and Devitt, Simon and Van Meter, Rodney},
  journal={New Journal of Physics},
  volume={14},
  number={12},
  pages={123011},
  year={2012},
  publisher={IOP Publishing}
}

@article{hesner2024using,
  title={Using Detector Likelihood for Benchmarking Quantum Error Correction},
  author={Hesner, Ian and Het{\'e}nyi, Bence and Wootton, James R},
  journal={Physical Review A},
  year={2025},
  volume={111},
  pages={052452},
  doi={10.1103/PhysRevA.111.052452}
}

@article{steane1996error,
  title={Error correcting codes in quantum theory},
  author={Steane, Andrew M},
  journal={Physical Review Letters},
  volume={77},
  number={5},
  pages={793},
  year={1996},
  publisher={APS}
}

@article{fowler2012surface,
  title={Surface codes: Towards practical large-scale quantum computation},
  author={Fowler, Austin G and Mariantoni, Matteo and Martinis, John M and Cleland, Andrew N},
  journal={Physical Review A—Atomic, Molecular, and Optical Physics},
  volume={86},
  number={3},
  pages={032324},
  year={2012},
  publisher={APS}
}

@article{aaronson2004improved,
  title={Improved simulation of stabilizer circuits},
  author={Aaronson, Scott and Gottesman, Daniel},
  journal={Physical Review A—Atomic, Molecular, and Optical Physics},
  volume={70},
  number={5},
  pages={052328},
  year={2004},
  publisher={APS}
}

@article{flammia2011direct,
  title={Direct fidelity estimation from few Pauli measurements},
  author={Flammia, Steven T and Liu, Yi-Kai},
  journal={Physical review letters},
  volume={106},
  number={23},
  pages={230501},
  year={2011},
  publisher={APS}
}

@article{evered2023high,
  title={High-fidelity parallel entangling gates on a neutral-atom quantum computer},
  author={Evered, Simon J and Bluvstein, Dolev and Kalinowski, Marcin and Ebadi, Sepehr and Manovitz, Tom and Zhou, Hengyun and Li, Sophie H and Geim, Alexandra A and Wang, Tout T and Maskara, Nishad and others},
  journal={Nature},
  volume={622},
  number={7982},
  pages={268--272},
  year={2023},
  publisher={Nature Publishing Group UK London}
}

@article{da2024demonstration,
  title={Demonstration of logical qubits and repeated error correction with better-than-physical error rates},
  author={Da Silva, MP and Ryan-Anderson, C and Bello-Rivas, JM and Chernoguzov, A and Dreiling, JM and Foltz, C and Gaebler, JP and Gatterman, TM and Hayes, D and Hewitt, N and others},
  journal={arXiv preprint arXiv:2404.02280},
  year={2024}
}

@article{krinner2022realizing,
  title={Realizing repeated quantum error correction in a distance-three surface code},
  author={Krinner, Sebastian and Lacroix, Nathan and Remm, Ants and Di Paolo, Agustin and Genois, Elie and Leroux, Catherine and Hellings, Christoph and Lazar, Stefania and Swiadek, Francois and Herrmann, Johannes and others},
  journal={Nature},
  volume={605},
  number={7911},
  pages={669--674},
  year={2022},
  publisher={Nature Publishing Group UK London}
}

@article{google2023suppressing,
  title={Suppressing quantum errors by scaling a surface code logical qubit},
  journal={Nature},
  volume={614},
  number={7949},
  pages={676--681},
  year={2023},
  publisher={Nature Publishing Group UK London},
  author={{Google Quantum AI}}
}

@article{derks2024designing,
  title={Designing fault-tolerant circuits using detector error models},
  author={Derks, Peter-Jan HS and Townsend-Teague, Alex and Burchards, Ansgar G and Eisert, Jens},
  journal={arXiv preprint arXiv:2407.13826},
  year={2024}
}

@article{wagner2021optimal,
  title={Optimal noise estimation from syndrome statistics of quantum codes},
  author={Wagner, Thomas and Kampermann, Hermann and Bru{\ss}, Dagmar and Kliesch, Martin},
  journal={Physical review research},
  volume={3},
  number={1},
  pages={013292},
  year={2021},
  publisher={APS}
}

@article{gidney2021stim,
  title={Stim: a fast stabilizer circuit simulator},
  author={Gidney, Craig},
  journal={Quantum},
  volume={5},
  pages={497},
  year={2021},
  publisher={Verein zur F{\"o}rderung des Open Access Publizierens in den Quantenwissenschaften}
}

@article{cain2024correlated,
  title={Correlated decoding of logical algorithms with transversal gates},
  author={Cain, Madelyn and Zhao, Chen and Zhou, Hengyun and Meister, Nadine and Ataides, J and Jaffe, Arthur and Bluvstein, Dolev and Lukin, Mikhail D},
  journal={Physical Review Letters},
  year={2024},
  volume={133},
  pages={240602},
  doi={10.1103/PhysRevLett.133.240602}
}

@article{hangleiter2024fault,
  title={Fault-Tolerant Compiling of Classically Hard Instantaneous Quantum Polynomial Circuits on Hypercubes},
  author={Hangleiter, Dominik and Kalinowski, Marcin and Bluvstein, Dolev and Cain, Madelyn and Maskara, Nishad and Gao, Xun and Kubica, Aleksander and Lukin, Mikhail D and Gullans, Michael J},
  journal={PRX Quantum},
  year={2025},
  volume={6},
  pages={020338},
  doi={10.1103/PRXQuantum.6.020338}
}

@article{combes2017logical,
  title={Logical randomized benchmarking},
  author={Combes, Joshua and Granade, Christopher and Ferrie, Christopher and Flammia, Steven T},
  journal={arXiv preprint arXiv:1702.03688},
  year={2017}
}

@article{ryan2024high,
  title={High-fidelity and Fault-tolerant Teleportation of a Logical Qubit using Transversal Gates and Lattice Surgery on a Trapped-ion Quantum Computer},
  author={Ryan-Anderson, C and Brown, NC and Baldwin, CH and Dreiling, JM and Foltz, C and Gaebler, JP and Gatterman, TM and Hewitt, N and Holliman, C and Horst, CV and others},
  journal={Science},
  year={2024},
  volume={385},
  pages={1327--1331}
}

@article{geher2024reset,
  title={To reset, or not to reset--that is the question},
  author={Geh{\'e}r, Gy{\"o}rgy P and Jastrzebski, Marcin and Campbell, Earl T and Crawford, Ophelia},
  journal={npj Quantum Information},
  year={2025},
  volume={11},
  pages={39},
  doi={10.1038/s41534-025-00998-y}
}

@article{higgott2022pymatching,
  title={Pymatching: A python package for decoding quantum codes with minimum-weight perfect matching},
  author={Higgott, Oscar},
  journal={ACM Transactions on Quantum Computing},
  volume={3},
  number={3},
  pages={1--16},
  year={2022},
  publisher={ACM New York, NY}
}

@article{mayer2024benchmarking,
  title={Benchmarking logical three-qubit quantum Fourier transform encoded in the Steane code on a trapped-ion quantum computer},
  author={Mayer, Karl and Ryan-Anderson, Ciar{\'a}n and Brown, Natalie and Durso-Sabina, Elijah and Baldwin, Charles H and Hayes, David and Dreiling, Joan M and Foltz, Cameron and Gaebler, John P and Gatterman, Thomas M and others},
  journal={arXiv preprint arXiv:2404.08616},
  year={2024}
}

@article{granger1980long,
  title={Long memory relationships and the aggregation of dynamic models},
  author={Granger, Clive WJ},
  journal={Journal of econometrics},
  volume={14},
  number={2},
  pages={227--238},
  year={1980},
  publisher={Elsevier}
}

@article{erland2007constructing,
  title={Constructing 1/ $\omega$ $\alpha$ noise from reversible Markov chains},
  author={Erland, Sveinung and Greenwood, Priscilla E},
  journal={Physical Review E—Statistical, Nonlinear, and Soft Matter Physics},
  volume={76},
  number={3},
  pages={031114},
  year={2007},
  publisher={APS}
}

@article{fowler2012proof,
  title={Proof of finite surface code threshold for matching},
  author={Fowler, Austin G},
  journal={Physical review letters},
  volume={109},
  number={18},
  pages={180502},
  year={2012},
  publisher={APS}
}

@article{chamberland2018flag,
  title={Flag fault-tolerant error correction with arbitrary distance codes},
  author={Chamberland, Christopher and Beverland, Michael E},
  journal={Quantum},
  volume={2},
  pages={53},
  year={2018},
  publisher={Verein zur F{\"o}rderung des Open Access Publizierens in den Quantenwissenschaften}
}

@article{knill2008randomized,
  title={Randomized benchmarking of quantum gates},
  author={Knill, Emanuel and Leibfried, Dietrich and Reichle, Rolf and Britton, Joe and Blakestad, R Brad and Jost, John D and Langer, Chris and Ozeri, Roee and Seidelin, Signe and Wineland, David J},
  journal={Physical Review A—Atomic, Molecular, and Optical Physics},
  volume={77},
  number={1},
  pages={012307},
  year={2008},
  publisher={APS}
}

@article{merkel2013self,
  title={Self-consistent quantum process tomography},
  author={Merkel, Seth T and Gambetta, Jay M and Smolin, John A and Poletto, Stefano and C{\'o}rcoles, Antonio D and Johnson, Blake R and Ryan, Colm A and Steffen, Matthias},
  journal={Physical Review A—Atomic, Molecular, and Optical Physics},
  volume={87},
  number={6},
  pages={062119},
  year={2013},
  publisher={APS}
}

@article{nielsen2021gate,
  title={Gate set tomography},
  author={Nielsen, Erik and Gamble, John King and Rudinger, Kenneth and Scholten, Travis and Young, Kevin and Blume-Kohout, Robin},
  journal={Quantum},
  volume={5},
  pages={557},
  year={2021},
  publisher={Verein zur F{\"o}rderung des Open Access Publizierens in den Quantenwissenschaften}
}

@article{carignan2023error,
  title={The error reconstruction and compiled calibration of quantum computing cycles},
  author={Carignan-Dugas, Arnaud and Dahlen, Dar and Hincks, Ian and Ospadov, Egor and Beale, Stefanie J and Ferracin, Samuele and Skanes-Norman, Joshua and Emerson, Joseph and Wallman, Joel J},
  journal={arXiv preprint arXiv:2303.17714},
  year={2023}
}

@article{takou2025estimating,
  title={Estimating decoding graphs and hypergraphs of memory quantum error-correction experiments},
  author={Takou, Evangelia and Brown, Kenneth R},
  journal={Physical Review A},
  year={2025},
  volume={112},
  pages={052414},
  doi={10.1103/qcz4-nx4r}
}

@article{remm2025experimentally,
  title={Experimentally Informed Decoding of Stabilizer Codes Based on Syndrome Correlations},
  author={Remm, Ants and Lacroix, Nathan and B{\"o}deker, Lukas and Genois, Elie and Hellings, Christoph and Swiadek, Fran{\c{c}}ois and Norris, Graham J and Eichler, Christopher and Blais, Alexandre and M{\"u}ller, Markus and others},
  journal={Physical Review Research},
  year={2026},
  volume={8},
  pages={013044},
  doi={10.1103/z1ng-wg3k}
}

@article{xinghua1999convergence,
  title={Convergence of Newton’s method and inverse function theorem in Banach space},
  author={Xinghua, Wang},
  journal={Mathematics of Computation},
  volume={68},
  number={225},
  pages={169--186},
  year={1999}
}

@article{gidney2024magic,
  title={Magic state cultivation: growing T states as cheap as CNOT gates},
  author={Gidney, Craig and Shutty, Noah and Jones, Cody},
  journal={arXiv preprint arXiv:2409.17595},
  year={2024}
}

@article{rosenfeld2025magic,
  title={Magic state cultivation on a superconducting quantum processor},
  author={Rosenfeld, Emma and Gidney, Craig and Roberts, Gabrielle and Morvan, Alexis and Lacroix, Nathan and Kafri, Dvir and Marshall, Jeffrey and Li, Ming and Sivak, Volodymyr and Abanin, Dmitry and others},
  journal={arXiv preprint arXiv:2512.13908},
  year={2025}
}

@article{bluvstein2025architectural,
  title={A fault-tolerant neutral-atom architecture for universal quantum computation},
  author={Bluvstein, Dolev and Geim, Alexandra A and Li, Sophie H and Evered, Simon J and Ataides, J and Baranes, Gefen and Gu, Andi and Manovitz, Tom and Xu, Muqing and Kalinowski, Marcin and others},
  journal={Nature},
  year={2026},
  volume={649},
  pages={39--46},
  doi={10.1038/s41586-025-09848-5}
}

@article{sales2025experimental,
  title={Experimental demonstration of logical magic state distillation},
  author={Sales Rodriguez, Pedro and Robinson, John M and Jepsen, Paul Niklas and He, Zhiyang and Duckering, Casey and Zhao, Chen and Wu, Kai-Hsin and Campo, Joseph and Bagnall, Kevin and Kwon, Minho and others},
  journal={Nature},
  volume={645},
  number={8081},
  pages={620--625},
  year={2025},
  publisher={Nature Publishing Group UK London}
}

@article{blume2025estimating,
  title={Estimating detector error models from syndrome data},
  author={Blume-Kohout, Robin and Young, Kevin},
  journal={arXiv preprint arXiv:2504.14643},
  year={2025}
}

@article{chao2020optimization,
  title={Optimization of the surface code design for Majorana-based qubits},
  author={Chao, Rui and Beverland, Michael E and Delfosse, Nicolas and Haah, Jeongwan},
  journal={Quantum},
  volume={4},
  pages={352},
  year={2020},
  publisher={Verein zur F{\"o}rderung des Open Access Publizierens in den Quantenwissenschaften}
}

@article{google2021exponential,
  title={Exponential suppression of bit or phase errors with cyclic error correction},
  journal={Nature},
  volume={595},
  number={7867},
  pages={383--387},
  year={2021},
  publisher={Nature Publishing Group UK London},
  author={{Google Quantum AI}}
}

@article{bravyi2013simulation,
  title={Simulation of rare events in quantum error correction},
  author={Bravyi, Sergey and Vargo, Alexander},
  journal={Physical Review A},
  year={2013},
  volume={88},
  pages={062308},
  doi={10.1103/PhysRevA.88.062308}
}

@article{mayer2025rare,
  title={Rare Event Simulation of Quantum Error-Correcting Circuits},
  author={Mayer, Carolyn and Ganti, Anand and Onunkwo, Uzoma and Metodi, Tzvetan and Anker, Benjamin and Skryzalin, Jacek},
  journal={arXiv preprint arXiv:2509.13678},
  year={2025}
}

@article{beverland2025fail,
  title={Fail fast: techniques to probe rare events in quantum error correction},
  author={Beverland, Michael E and Carroll, Malcolm and Cross, Andrew W and Yoder, Theodore J},
  journal={arXiv preprint arXiv:2511.15177},
  year={2025}
}

@article{fowler2012analytic,
  title={Analytic asymptotic performance of topological codes},
  author={Fowler, Austin G},
  journal={Physical Review A},
  year={2013},
  volume={87},
  pages={040301(R)},
  doi={10.1103/PhysRevA.87.040301}
}

@article{tan2024resilience,
  title={Resilience of the surface code to error bursts},
  author={Tan, Shi Jie Samuel and Pattison, Christopher A and McEwen, Matt and Preskill, John},
  journal={Physical Review A},
  year={2026},
  volume={113},
  pages={022450},
  doi={10.1103/vwtl-hwfx}
}

@article{harrington2025synchronous,
  title={Synchronous detection of cosmic rays and correlated errors in superconducting qubit arrays},
  author={Harrington, Patrick M and Li, Mingyu and Hays, Max and Van De Pontseele, Wouter and Mayer, Daniel and Pinckney, H Douglas and Contipelli, Felipe and Gingras, Michael and Niedzielski, Bethany M and Stickler, Hannah and others},
  journal={Nature Communications},
  volume={16},
  number={1},
  pages={6428},
  year={2025},
  publisher={Nature Publishing Group UK London}
}

@unpublished{XiaoWorkInProgress2026,
  author   = {Xiao, Xiao and
              Hangleiter, Dominik and
              Bluvstein, Dolev and
              Lukin, Mikhail D. and
              Gullans, Michael J.},
  title    = {In-situ benchmarking of fault-tolerant quantum circuits. {II}. Classically-hard circuits},
  note = {in preparation},
  date     = {2026},
  year={2026}
}

@misc{ZhengLogicalNoiseWorkInProgress2026,
  author   = {Zheng, Han and
              Chu, Andy and
              Chen, Senrui and
              Manes, Argyris Giannisis and
              Lee, Su-un and
              Zhou, Sisi and
              Jiang, Liang},
  title    = {Efficient learning of logical noise from syndrome data},
  year={2026},
  eprint={2601.22286},
  archivePrefix={arXiv},
  primaryClass={quant-ph}
}

@article{pastawski2015fault,
  title={Fault-tolerant logical gates in quantum error-correcting codes},
  author={Pastawski, Fernando and Yoshida, Beni},
  journal={Physical Review A},
  volume={91},
  number={1},
  pages={012305},
  year={2015},
  publisher={APS}
}

@article{emerson2005scalable,
  title={Scalable noise estimation with random unitary operators},
  author={Emerson, Joseph and Alicki, Robert and {\.Z}yczkowski, Karol},
  journal={Journal of Optics B: Quantum and Semiclassical Optics},
  volume={7},
  number={10},
  pages={S347--S352},
  year={2005}
}

@article{magesan2012efficient,
  title={Efficient measurement of quantum gate error by interleaved randomized benchmarking},
  author={Magesan, Easwar and Gambetta, Jay M and Johnson, Blake R and Ryan, Colm A and Chow, Jerry M and Merkel, Seth T and Da Silva, Marcus P and Keefe, George A and Rothwell, Mary B and Ohki, Thomas A and others},
  journal={Physical review letters},
  volume={109},
  number={8},
  pages={080505},
  year={2012},
  publisher={APS}
}

@article{erhard2019characterizing,
  title={Characterizing large-scale quantum computers via cycle benchmarking},
  author={Erhard, Alexander and Wallman, Joel J and Postler, Lukas and Meth, Michael and Stricker, Roman and Martinez, Esteban A and Schindler, Philipp and Monz, Thomas and Emerson, Joseph and Blatt, Rainer},
  journal={Nature communications},
  volume={10},
  number={1},
  pages={5347},
  year={2019},
  publisher={Nature Publishing Group UK London}
}

@article{harper2020efficient,
  title={Efficient learning of quantum noise},
  author={Harper, Robin and Flammia, Steven T and Wallman, Joel J},
  journal={Nature Physics},
  volume={16},
  number={12},
  pages={1184--1188},
  year={2020},
  publisher={Nature Publishing Group UK London}
}

@article{harper2021fast,
  title={Fast estimation of sparse quantum noise},
  author={Harper, Robin and Yu, Wenjun and Flammia, Steven T},
  journal={PRX Quantum},
  volume={2},
  number={1},
  pages={010322},
  year={2021},
  publisher={APS}
}

@article{fazio2026characterizing,
  title={Characterizing Physical and Logical Errors in a Transversal CNOT Gate via Cycle Error Reconstruction},
  author={Fazio, Nicholas and Freund, Robert and Sannamoth, Debankan and Steiner, Alex and Marciniak, Christian D and Rispler, Manuel and Harper, Robin and Monz, Thomas and Emerson, Joseph and Bartlett, Stephen D},
  journal={Physical Review X},
  volume={16},
  number={1},
  pages={011030},
  year={2026},
  publisher={APS}
}

@article{hockings2025scalable,
  title={Scalable noise characterization of syndrome-extraction circuits with averaged circuit eigenvalue sampling},
  author={Hockings, Evan T and Doherty, Andrew C and Harper, Robin},
  journal={PRX quantum},
  volume={6},
  number={1},
  pages={010334},
  year={2025},
  publisher={APS}
}

@article{hockings2025improving,
  title={Improving error suppression with noise-aware decoding},
  author={Hockings, Evan T and Doherty, Andrew C and Harper, Robin},
  journal={arXiv preprint arXiv:2502.21044},
  year={2025}
}

@article{wallman2016noise,
  title={Noise tailoring for scalable quantum computation via randomized compiling},
  author={Wallman, Joel J and Emerson, Joseph},
  journal={Physical Review A},
  volume={94},
  number={5},
  pages={052325},
  year={2016},
  publisher={APS}
}

@article{hashim2020randomized,
  title={Randomized compiling for scalable quantum computing on a noisy superconducting quantum processor},
  author={Hashim, Akel and Naik, Ravi K and Morvan, Alexis and Ville, Jean-Loup and Mitchell, Bradley and Kreikebaum, John Mark and Davis, Marc and Smith, Ethan and Iancu, Costin and O'Brien, Kevin P and others},
  journal={arXiv preprint arXiv:2010.00215},
  year={2020}
}

@article{vepsalainen2020impact,
  title={Impact of ionizing radiation on superconducting qubit coherence},
  author={Veps{\"a}l{\"a}inen, Antti P and Karamlou, Amir H and Orrell, John L and Dogra, Akshunna S and Loer, Ben and Vasconcelos, Francisca and Kim, David K and Melville, Alexander J and Niedzielski, Bethany M and Yoder, Jonilyn L and others},
  journal={Nature},
  volume={584},
  number={7822},
  pages={551--556},
  year={2020},
  publisher={Nature Publishing Group UK London}
}

@article{wilen2021correlated,
  title={Correlated charge noise and relaxation errors in superconducting qubits},
  author={Wilen, Christopher D and Abdullah, S and Kurinsky, NA and Stanford, C and Cardani, L and d’Imperio, G and Tomei, C and Faoro, L and Ioffe, LB and Liu, CH and others},
  journal={Nature},
  volume={594},
  number={7863},
  pages={369--373},
  year={2021},
  publisher={Nature Publishing Group UK London}
}

@article{mcewen2022resolving,
  title={Resolving catastrophic error bursts from cosmic rays in large arrays of superconducting qubits},
  author={McEwen, Matt and Faoro, Lara and Arya, Kunal and Dunsworth, Andrew and Huang, Trent and Kim, Seon and Burkett, Brian and Fowler, Austin and Arute, Frank and Bardin, Joseph C and others},
  journal={Nature Physics},
  volume={18},
  number={1},
  pages={107--111},
  year={2022},
  publisher={Nature Publishing Group UK London}
}

@article{li2025cosmic,
  title={Cosmic-ray-induced correlated errors in superconducting qubit array},
  author={Li, Xuegang and Wang, Junhua and Jiang, Yao-Yao and Xue, Guang-Ming and Cai, Xiaoxia and Zhou, Jun and Gong, Ming and Liu, Zhao-Feng and Zheng, Shuang-Yu and Ma, Deng-Ke and others},
  journal={Nature Communications},
  volume={16},
  number={1},
  pages={4677},
  year={2025},
  publisher={Nature Publishing Group UK London}
}

@article{larson2025quasiparticle,
  title={Quasiparticle poisoning of superconducting qubits with active gamma irradiation},
  author={Larson, CP and Yelton, E and Dodge, K and Okubo, K and Batarekh, J and Iaia, V and Kurinsky, NA and Plourde, BLT},
  journal={PRX Quantum},
  volume={6},
  number={3},
  pages={030339},
  year={2025},
  publisher={APS}
}

@article{binney2026distinguishing,
  title={Distinguishing types of correlated errors in superconducting qubits},
  author={Binney, Hannah P and Pinckney, H Douglas and Azar, Kate and Harrington, Patrick M and Jha, Shantanu and Li, Mingyu and Yang, Jiatong and Contipelli, Felipe and Pi{\~n}ero, Ren{\'e}e DePencier and Gingras, Michael and others},
  journal={arXiv preprint arXiv:2603.16494},
  year={2026}
}

@article{cardani2021reducing,
  title={Reducing the impact of radioactivity on quantum circuits in a deep-underground facility},
  author={Cardani, Laura and Valenti, Francesco and Casali, Nicola and Catelani, Gianluigi and Charpentier, Thibault and Clemenza, Massimiliano and Colantoni, Ivan and Cruciani, Angelo and D’Imperio, G and Gironi, Luca and others},
  journal={Nature communications},
  volume={12},
  number={1},
  pages={2733},
  year={2021},
  publisher={Nature Publishing Group UK London}
}

@article{iaia2022phonon,
  title={Phonon downconversion to suppress correlated errors in superconducting qubits},
  author={Iaia, V and Ku, J and Ballard, A and Larson, CP and Yelton, E and Liu, CH and Patel, S and McDermott, R and Plourde, BLT},
  journal={Nature Communications},
  volume={13},
  number={1},
  pages={6425},
  year={2022},
  publisher={Nature Publishing Group UK London}
}

@article{bratrud2025measurement,
  title={Measurement of correlated charge noise in superconducting qubits at an underground facility},
  author={Bratrud, G and Lewis, S and Anyang, K and Col{\'o}n Cesan{\'\i}, A and Dyson, T and Magoon, H and Sabhari, D and Spahn, G and Wagner, G and Gualtieri, R and others},
  journal={Nature Communications},
  volume={16},
  number={1},
  pages={9906},
  year={2025},
  publisher={Nature Publishing Group UK London}
}

@article{mcewen2024resisting,
  title={Resisting high-energy impact events through gap engineering in superconducting qubit arrays},
  author={McEwen, Matt and Miao, Kevin C and Atalaya, Juan and Bilmes, Alexander and Crook, Alex and Bovaird, Jenna and Kreikebaum, John Mark and Zobrist, Nicholas and Jeffrey, Evan and Ying, Bicheng and others},
  journal={Physical Review Letters},
  volume={133},
  number={24},
  pages={240601},
  year={2024},
  publisher={APS}
}

@article{mcjunkin2026demand,
  title={On-Demand Correlated Errors in Superconducting Qubits from a Particle Accelerator},
  author={McJunkin, Thomas and Hunt, AW and Jones-Alberty, Yenuel and Haard, TM and Spear, MK and Shackford, James and Gilliss, Tom and Amezcua, Mayra and Watson, CA and Sweeney, TM and others},
  journal={arXiv preprint arXiv:2603.13124},
  year={2026}
}

@article{pinckney2026characterization,
  title={Characterization of Radiation-Induced Errors in Superconducting Qubits Protected with Various Gap-Engineering Strategies},
  author={Pinckney, H Douglas and McJunkin, Thomas and Hunt, Alan W and Harrington, Patrick M and Binney, Hannah P and Hays, Max and Jones-Alberty, Yenuel and Azar, Kate and Contipelli, Felipe and Pi{\~n}ero, Ren{\'e}e DePencier and others},
  journal={arXiv preprint arXiv:2603.13460},
  year={2026}
}

@article{kurilovich2025correlated,
  title={Correlated Phase Error Bursts in a Gap-Engineered Superconducting Qubit Array},
  author={Kurilovich, Vladislav D and Roberts, Gabrielle and Martin, Leigh S and McEwen, Matt and Eickbusch, Alec and Faoro, Lara and Ioffe, Lev B and Atalaya, Juan and Bilmes, Alexander and Kreikebaum, John Mark and others},
  journal={Physical Review X},
  year={2026},
  volume={16},
  pages={021025},
  doi={10.1103/1bl4-b2f7}
}

@article{kono2024mechanically,
  title={Mechanically induced correlated errors on superconducting qubits with relaxation times exceeding 0.4 ms},
  author={Kono, Shingo and Pan, Jiahe and Chegnizadeh, Mahdi and Wang, Xuxin and Youssefi, Amir and Scigliuzzo, Marco and Kippenberg, Tobias J},
  journal={Nature Communications},
  volume={15},
  number={1},
  pages={3950},
  year={2024},
  publisher={Nature Publishing Group UK London}
}

@article{martinis2021saving,
  title={Saving superconducting quantum processors from decay and correlated errors generated by gamma and cosmic rays},
  author={Martinis, John M},
  journal={npj Quantum Information},
  volume={7},
  number={1},
  pages={90},
  year={2021},
  publisher={Nature Publishing Group UK London}
}

@article{nayak2026iterative,
  title={Iterative Decoding of Stabilizer Codes under Radiation-Induced Correlated Noise},
  author={Nayak, Anuj K and Baity, Paul G and Love, Peter J and Jeon, Nicholas and Yoon, Byung-Jun and Hoisie, Adolfy and Varshney, Lav R},
  journal={arXiv preprint arXiv:2603.18231},
  year={2026}
}

@article{vallero2026radiation,
  title={Radiation-Induced Fault Detection in Superconducting Quantum Devices},
  author={Vallero, Marzio and Casagranda, Gioele and Vella, Flavio and Rech, Paolo},
  journal={Advanced Quantum Technologies},
  volume={9},
  number={4},
  pages={e00469},
  year={2026},
  publisher={Wiley Online Library}
}

@article{poster2026reloqate,
  title={ReloQate: Transient Drift Detection and In-Situ Recalibration in Surface Code Quantum Error Correction},
  author={Poster, Maxwell and Chadwick, Jason and Baker, Jonathan Mark},
  journal={arXiv preprint arXiv:2603.00837},
  year={2026}
}

@article{baity2026characterizing,
  title={Characterizing Quantum Error Correction Performance of Radiation-induced Errors},
  author={Baity, Paul G and Nayak, Anuj K and Varshney, Lav R and Jeon, Nicholas and Yoon, Byung-Jun and Love, Peter J and Hoisie, Adolfy},
  journal={arXiv preprint arXiv:2602.06202},
  year={2026}
}

@article{pesah2025fault,
  title={Fault-tolerant transformations of spacetime codes},
  author={Pesah, Arthur and Daniel, Austin K and Tzitrin, Ilan and Vasmer, Michael},
  journal={arXiv preprint arXiv:2509.09603},
  year={2025}
}

@misc{xiao2026code,
  author={Xiao, Xiao},
  title={Code that generates the data in {In Situ Benchmarking of Fault-Tolerant Quantum Circuits: Clifford Circuits}},
  year={2026},
  url={https://github.com/xiaoxiao-01/in_situ_clifford}
}

\onecolumngrid

\newpage
\begin{appendix}
\crefalias{section}{appendix}

\section{Overview}

\begin{figure*}[h]
\centering
    \includegraphics[width=1\textwidth]{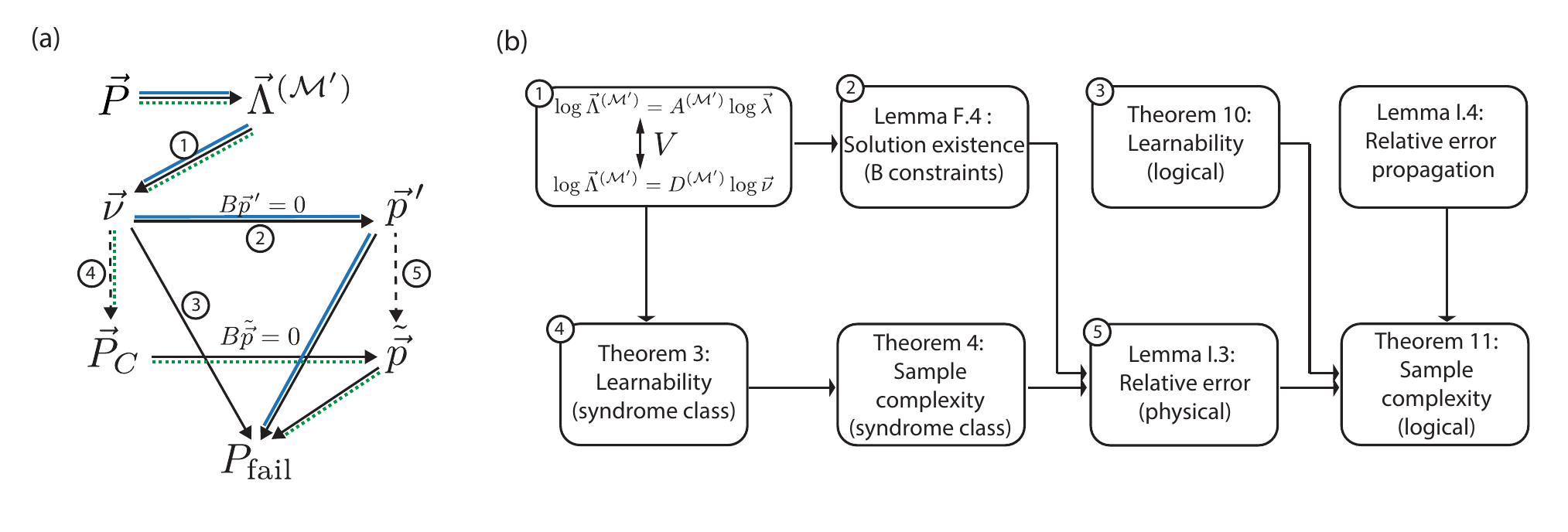}
    \caption{An overview of the structure of the lemma and theorems. (a) The graph relates different quantities throughout the physical and logical-level noise learning process. $\vec{p}$ are the true physical error rates. $\vec{\Lambda}$ are the syndrome expectation values (Pauli eigenvalues of the channel). $\vec{\nu}$ are the learnable transformed eigenvalues that approximate the syndrome class error rates $\vec{P}_C\approx-\frac{\log \vec{\nu}}{2}$. $\vec{p}\,'$ are the error rates that are consistent with $\vec{\nu}$ and the extra constraints $B\vec{p}\,'=0$ within each syndrome class. $\tilde{\vec{p}}$ are the estimated physical error rates based on the estimated $\vec{P}_C$ and the same extra constraints. A solid black arrow from $A$ to $B$ indicates that $B$ can be derived from $A$, whereas a dashed black arrow indicates that $A$ and $B$ are approximately related in the low-error regime. Some edges are labeled with the lemmas and theorems used in (b). The solid blue line shows the path of learning the logical error rate without low-error rate approximation by solving the non-linear equation (circled 1 and 2) using optimization methods, but we don't have a proof for the sample complexity in this case. The dotted green line shows the path along which one can linearize $\log \nu$ and ignore second-order corrections to obtain error rates from linear equations; we can then prove the sample complexity in this case for sufficiently low physical error rates. (b) Lemmas and theorems, starting from the syndrome equations, eventually lead to the sample-complexity scaling for learning the logical error rate from the syndrome.}
    \label{fig: lemma_illustration}
\end{figure*}

The appendices present the proofs of the theorems in the main text, along with the lemmas used in the proofs and the algorithms used for noise learning. \cref{fig: lemma_illustration} illustrates the structure of the main results in this paper. In particular, we show two routes that start from the syndrome expectation values generated from the true error rates to learning physical error rates, and eventually learning the logical error rates. The first route, in solid blue lines, represents an exact solution that uses nonlinear equation solving via optimization. The second route in dotted green lines represents the first-order linear approximation method, where we are able to prove, in a certain error regime, an exponential advantage of sampling complexity in learning the logical error rate from syndrome data, compared with learning from logical measurement results.

We now explain the structure of the appendices. The essential results are represented by \cref{theorem: syndrome class learning}, \cref{theorem: logical learning}, \cref{theorem: syndrome class learning sample complexity}, \cref{theorem: logical sample complexity}. These are results on the learnability at the physical and logical levels, along with their sample complexities. 

\cref{appendix: log u basis} introduces the basis transform from Pauli eigenvalues $\Lambda$ to the transformed eigenvalues $U$, which allows us to have a minimal set of parameters describing the total error channel and the local channels. The matrix $D$ of the linear equations matrix in this new basis has a nice structure that allows us to prove learnability conditions. \cref{appendix: A rank general} describes the structure of the linear equations from syndrome data, and proves the rank of the corresponding matrices, i.e., the minimum number of irredundant stabilizer elements needed to extract noise information from a given Pauli noise model. \cref{appendix: minimal equations} describes an efficient algorithm to find the minimal subset of stabilizer elements.

\cref{appendix: physical proof} presents the proof of \cref{theorem: syndrome class learning}, where we show that syndrome-class error rates can be approximately determined from the syndrome data under the assumption of low error rates. This observation also provides a way to linearize the equation between syndrome class error rates and syndrome data directly. We show that in the case of the detector error model/inclusive error model \cite{remm2025experimentally, takou2025estimating} where all errors are assumed to be uncorrelated, the detector error rate, i.e., probability that an odd number of errors occur for each syndrome class, can be learned exactly for any stabilizer code. 

\cref{appendix: existence of solution with B constraints} proves \cref{lemma: existence of solution}, which shows that there always exist error rates that satisfy both the syndrome constraints and extra constraints under low physical error rates. \cref{appendix: spacetime code} describes the spacetime code formulation that maps Clifford circuits to stabilizer/subsystem codes in more detail, and we give an algorithm that generates the stabilizer and logical operators of the spacetime code. \cref{appendix: circuit logical learning} proves \cref{theorem: logical learning}, which shows that if the Clifford circuit is fault-tolerant to the Pauli error model we consider, then any error rates in this model that satisfy the syndrome constraints must give rise to the same true logical error rate with respect to the logical measurements in the circuit.

\cref{appendix: sample complexity} proves the sample complexity scaling for learning physical syndrome class error rate (\cref{theorem: syndrome class learning sample complexity}) and logical error rates (\cref{theorem: logical sample complexity}) in the low physical error rate regime.  We show that the original solution \cite{remm2025experimentally} for solving detector error models assuming independent error mechanisms can be used to learn syndrome class error in our more general Pauli model with constant sample complexity (with increasing system size) if the spacetime code of the circuit is qLDPC and the Pauli channel is local and sparse. On top of that, we show that the logical error rate can be learned with a polynomial sample size, where the polynomial overhead arises from the relative error propagation from the physical to the logical levels. This then has the exponential advantage in sample complexity for learning logical error rates from syndrome data over direct fidelity/tomography estimation at the logical level. This is because logical error rates decrease exponentially with the distance of the code and therefore require an exponentially large sample size to learn the logical error rate within a finite relative error. The proof requires a decreasing upper bound on the physical error rates as the system size increases; in the practically relevant case when physical error rates are fixed, this sample-complexity advantage is only proven up to some finite size. However, we believe this advantage holds in general when using the non-linear optimization method without a first-order approximation. Finally, \cref{appendix: Intraclass uniformity} describes a method used in the experimental analysis of adding intraclass uniform constraints that reduce overfitting in the higher shot noise regime.

\section{The transformed eigenvalue basis\label{appendix: log u basis}}
One of the central questions we want to understand is given the Pauli eigenvalues $\Lambda(M)$ of the Pauli channel $\mathcal{N}_{\Gamma}$ for elements in the measured stabilizer group $\mathcal{M}$, what can we learn about the total channel $\mathcal{N}_{\Gamma}$ and the individual noise channels $\mathcal{N}_{\gamma}, \forall \gamma\in\Gamma$? \textcite{wagner2022pauli,wagner2023learning} used the canonical factors $F$ of the Pauli eigenvalues, which is a linear transformation $\log \Lambda\to \log F$, to parametrize the total error channel $\mathcal{N}_{\Gamma}$ and showed that there is effectively $|\mathcal{E}_{\Gamma}|$ degrees of freedom. They further proved the sufficient condition of learnability of $N_{\Gamma}$. Our work is based on another linear transformation of $\log \Lambda$, denoted as $\log U$, and we refer to $U$ as the transformed eigenvalues. $\log U$ serve as another parametrization of the channel $\mathcal{N}_{\Gamma}$ and shares similar properties to the canonical factors $\log F$. Based on $\log U$, we have a simple proof of the learnability results \cref{theorem: wagner} with the additional necessary condition for learnability, and more importantly, answer the questions (\cref{theorem: syndrome class learning}) of what we can generally learn about the parameters of the local channels $\mathcal{N}_{\gamma}$, when the learnability condition in \cref{theorem: wagner} is not satisfied. See \cref{fig: basis transform} for an illustration of the basis transforms. In this appendix section, we first introduce the canonical factors $\log F$ \cite{wagner2022pauli,wagner2023learning}, and then introduce the transformed eigenvalues $\log U$ used in this work.

\begin{figure*}[h]
\centering
    \includegraphics[width=0.42\textwidth]{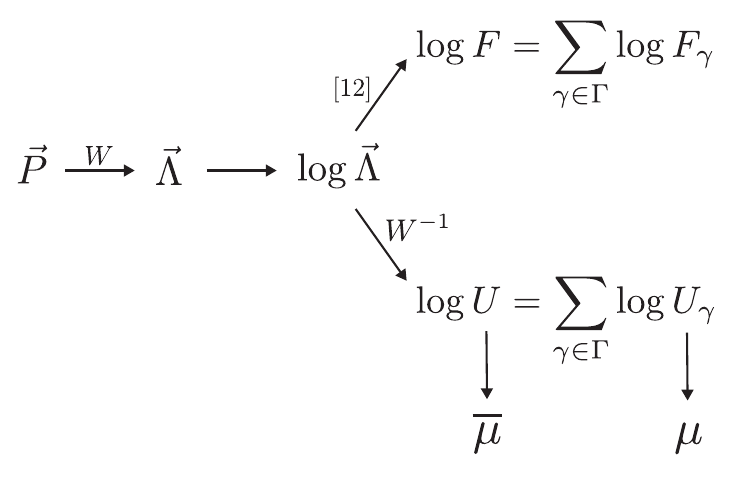}
    \caption{An overview of the basis transform to the transformed eigenvalues $U$. The Pauli eigenvalues $\vec{\Lambda}$ of the error channel are obtained from the Walsh-Hadamard transform of the true physical error rates. \textcite{wagner2023learning} performs a transformation to get the canonical moments $F$ that parametrize the total channel $\mathcal{N}_{\Gamma}$. Our work performs an inverse Hadamard transform on $\log \vec{\Lambda}$ instead and obtains the transformed eigenvalues $\overline{\mu}$ and $\mu$ as parametrization of the total channel $\mathcal{N}_{\Gamma}$ and the local channels $\mathcal{N}_{\gamma}$ respectively.}
    \label{fig: basis transform}
\end{figure*}

\subsection{$\log F$: Basis of the canonical factor  \cite{wagner2023learning}}
Given the convolutional factor graph structure \cref{fig: conv_factor_graph_total} of the noise model, we have the eigenvalue $\Lambda(O)$ of any Pauli $O\in \mathcal{P}_n$ for the total error channel $\mathcal{N}_{\Gamma}$ decomposed as a factor graph (see \cref{fig: conv_factor_graph_total}):
\begin{align}
    \Lambda(O)=\prod_{\gamma\in\Gamma}\Lambda_{\gamma}(O)=\prod_{\gamma\in\Gamma}\Lambda_{\gamma}(O_{\gamma})
\end{align}
where the local eigenvalues $\Lambda_{\gamma}(O)$ of $O$ only depends on its restriction on the support $\gamma$, i.e., $\Lambda_{\gamma}(O)=\Lambda_{\gamma}(O_{\gamma})$.

To understand the degrees of freedom in both the total Pauli error rates $P$ and in each local error rate $P_{\gamma}$, \textcite{wagner2023learning} considers factors in another basis where we refer to the canonical basis from the canonical factors as defined typically in a factor graph \cite{wagner2023learning,flammia2020efficient,abbeel2006learning}.

\begin{appdefinition}[Canonical factor \cite{wagner2023learning}]\label{definition: canonical moments}
    The canonical factors $F(O), \forall O\in\mathcal{P}_n$ of the total Pauli channel with eigenvalues $\Lambda(O), \forall O\in\mathcal{P}_n$ is defined as:
    \begin{equation}
        F(O)=\prod_{O'\leq O}\Lambda^{(-1)^{|O|-|O'|}}(O').
    \end{equation}
    where $a\leq b$ if and only if $a$ is a Pauli substring of $b$, i.e., $b_{\mathrm{supp}(a)}=a$.
\end{appdefinition}

Equivalently, this can be written as a linear transformation of $\log \Lambda(O), \forall O\in \mathcal{P}_n$:
\begin{equation}
    \log F(O) = \sum_{O'\leq O}(-1)^{|O|-|O'|}\log \Lambda(O')\label{equation: log canonical global}
\end{equation}
Similarly, the local canonical factors $F_{\gamma}(O), \forall O\in\mathcal{P}_n$ of a local Pauli channels in a system of $n$ qubits with support $\gamma$ with eigenvalues $\Lambda_{\gamma}(O)=\Lambda_{\gamma}(O_{\gamma}), \forall O\in\mathcal{P}_n$ is defined as:
    \begin{equation}
        \log F_{\gamma}(O)=\sum_{O'\leq O} (-1)^{|O|-|O'|} \log\Lambda_{\gamma}(O').\label{equation: log canonical local}
    \end{equation}

\begin{appproperty}\label{property: canonical factors}
The (local) canonical factors $F$ ($F_{\gamma}$) have the following properties \cite{wagner2023learning}:
\begin{enumerate}
    \item The canonical factors are the product of local canonical factors, or equivalently:
        \begin{equation}
            \log F(O)=\sum_{\gamma\in\Gamma}\log F_{\gamma}(O).
            \label{equation: global canonical factors as local canonical factors}
        \end{equation}
    \item Local canonical factors $\log F_{\gamma}(O)=0$, if $\text{supp}(O)\not\subseteq \gamma$. For the global canonical factors we have $\log F(O)=0$, if $\text{supp}(O)\not\subseteq \gamma, \forall \gamma\in\Gamma$.
    \item The linear transformation in \cref{equation: log canonical global} and \cref{equation: log canonical global} is full rank. The eigenvalue of local Pauli channels for operator $O$ can be written as products of all the canonical factors of substrings of $O$: $\log\Lambda_{\gamma}(O)=\sum_{O'\leq O}\log F_{\gamma}(O')$. Similarly, the eigenvalues of the total channel are given by \cite{wagner2023learning}: $\log\Lambda(O)=\sum_{O'\leq O}\log F(O')$.
    
\end{enumerate}   
\end{appproperty}

As a result, $\{F(e)|\forall e\in\mathcal{E}_{\Gamma}\}$ is an irredundant set that parametrizes the total Pauli channel $\mathcal{N}_{\Gamma}$. Similarly, for each local channel $\gamma\in\Gamma$, $\{F_{\gamma}(e)|e\in\mathcal{E}_{\gamma}\}$ parametrizes the local channel $\mathcal{N}_{\gamma}$. Note that there is gauge freedom in general for each local channel if one only fixes the total channel $\mathcal{N}_{\Gamma}$. One benefit of looking at the canonical factors $F$ as opposed to the Pauli eigenvalue $\Lambda$ is that the gauge freedom is straightforward, as shown in \cref{equation: global canonical factors as local canonical factors}. More explicitly, whenever we have two non-overlapping distinct channel supports $\gamma_1,\gamma_2\in \Gamma$ and Pauli error $e$ such that $\mathrm{supp}(e)\subseteq \gamma_1\cap \gamma_2$, we have associated gauge freedom $\alpha$:
\begin{equation}\label{equation: gauge freedom}
\log F_{\gamma_1}(e) \longrightarrow \log F_{\gamma_1}(e)+\alpha \quad \text{and} \quad \log F_{\gamma_2}(e) \longrightarrow \log F_{\gamma_2}(e)-\alpha,
\end{equation}
such that the total channel stays invariant.
And given the total Pauli channel $\mathcal{N}_{\Gamma}$, one could determine the local Pauli channel $\mathcal{N}_{\gamma}$ if and only if $\gamma\cap\gamma'=\emptyset$ for all other local supports $\gamma'\in\Gamma$. This is the case mostly considered in the main text when mapping standard circuit-level Pauli noise to the spacetime code.

\subsection{$\log U$: Basis of the transformed eigenvalues}
Here, we introduced the alternative basis that allows us to design the linearized approximate algorithm for noise learning and prove the theorems in this work. In the main text, we defined the set of transformed local eigenvalues $\mu$, used to approximate the physical error rates. In this section, we introduce the formalism more systematically by first introducing the associated global quantities, the transformed eigenvalues of the total channel $U(O), \forall O\in \mathcal{P}_n$, whose logarithm is the inverse Walsh-Hadamard transform of the Pauli eigenvalues $\log \Lambda$ of the total channel $\mathcal{N}_{\Gamma}$:
\begin{equation}
    \log U(O)= -\frac{2}{|\mathcal{P}_n|}\sum_{O'\in\mathcal{P}_n}[[O',O]]\log \Lambda(O'_{\gamma}).\label{equation: U from Lambda}
\end{equation}
The extra $-2$ factor here is a convention that gives a simple binary form of the matrix $D$ in \cref{equation: D entries}. For each individual channel $\mathcal{N}_{\gamma}$, we similarly define the transformed eigenvalues $U_{\gamma}(O)$ of the individual channel, or simply the transformed eigenvalues as
\begin{equation}
    \log U_{\gamma}(O)= -\frac{2}{|\mathcal{P}_n|}\sum_{O'\in\mathcal{P}_n}[[O',O]]\log \Lambda_{\gamma}(O'_{\gamma}).\label{equation: U from Lambda local}
\end{equation}
Here are some properties of the transformed eigenvalues $U$ and $U_{\gamma}$ analogous to \cref{property: canonical factors}:

\begin{applemma}[Properties of $U$ and $U_{\gamma}$]\label{lemma: properties of U}
For any $O\in\mathcal{P}_n$:
\begin{enumerate} 
    \item $\log U(O)=\sum_{\gamma\in\Gamma}\log U_{\gamma}(O)$.
    \item For any local support $\gamma\in\Gamma$, $\log U_{\gamma}(O)=0$, if $\text{supp}(O)\not\subseteq \gamma$. And $\log U(O)=0$, if $\forall \gamma\in\Gamma$, $\text{supp}(O)\not\subseteq \gamma$.
\end{enumerate}   
\end{applemma}
\begin{proof}
To prove the first property, note that for the Pauli eigenvalues $\Lambda(O)=\prod_{\gamma\in\Gamma}\Lambda_{\gamma}(O)$, therefore we have $\log \Lambda(O)=\sum_{\gamma\in\Gamma}\log\Lambda_{\gamma}(O)$. Since the logarithm of the transformed eigenvalues $\log U$ and $\log U_{\gamma}$ are the same linear transformation of $\log \Lambda$ and $\log \Lambda_{\gamma}$ respectively, property $1$ follows directly.

To prove property 2, for any support $\gamma\in\Gamma$:
\begin{align}
    \log U_{\gamma}(O)&= -\frac{2}{|\mathcal{P}_n|}\sum_{O'\in\mathcal{P}_n}[[O',O]]\log \Lambda_{\gamma}(O'_{\gamma})\\
    &=-\frac{2}{|\mathcal{P}_n|}\sum_{O'_{\gamma^c}\in\mathcal{P}_{\gamma^c}}[[O'_{\gamma^c},O_{\gamma^c}]]\sum_{O'_{\gamma}\in\mathcal{P}_{\gamma}}[[O'_{\gamma},O_{\gamma}]] \log \Lambda_{\gamma}(O'_{\gamma})\\
    &=\frac{-2}{4^{|\gamma|}}\sum_{O'_{\gamma}\in\mathcal{P}_{\gamma}}[[O'_{\gamma},O_{\gamma}]] \log \Lambda_{\gamma}(O'_{\gamma}).
\end{align}
where we denote the complement of the support $\gamma$ as $\gamma^c=[n]\backslash \gamma$, and $\mathcal{P}_{\gamma}$ is the subgroup of $\mathcal{P}_n$ containing all Pauli operators supported on $\gamma$.
Since $\sum_{O'_{\gamma^c}\in\mathcal{P}_{\gamma^c}}[[O'_{\gamma^c},O_{\gamma^c}]]=0$ if $O\notin \mathcal{E}_{\gamma}$, $\log U_{\gamma}(O)=0$ if $\mathrm{supp}(O)\not\subseteq \gamma$. And if $\forall \gamma\in\Gamma$, $\text{supp}(O)\not\subseteq \gamma$, then $\log U(0)=0$ due to property 1.
\end{proof}

This motivate us to define the set of transformed local eigenvalues for the total channel as $\overline{\mu}=\{U(e)|e\in\overline{\mathcal{E}}_{\Gamma}\}$ and the corresponding vector $\vec{\overline{\mu}}=(U(e)|e\in\overline{\mu})$ with entries in the order of \cref{equation: order}. This set of transformed local eigenvalues is the set of values of $\log U(e)$ of the local errors $e\in\overline{\mathcal{E}}_{\Gamma}$, and it is the set of nontrivial factors that determines the total channel $\mathcal{N}_{\Gamma}$. This is because for all $e\notin \overline{\mathcal{E}}_{\Gamma}\cup\{I\}$, we have $\log U(e)=0$ according to \cref{lemma: properties of U}, while for $e=I$, $\log U(I)$ depends on the values in $\mu$:
\begin{equation}\label{equation: dependence on log U(I)}
    \log U(I)=-\sum_{e\in \overline{\mathcal{E}}_{\Gamma}} \log U(e),
\end{equation}
since $\log \Lambda(I)\propto \sum_{e\in\mathcal{P}_n}\log U(e)=\sum_{e\in\overline{\mathcal{E}}_{\Gamma}+\{I\}} \log U(e)=0$. In particular, for any Pauli eigenvalue $\Lambda(e), \forall e\in\mathcal{P}_{n}$ we have:
\begin{align}
    \log \Lambda(e)&= -\frac{1}{2}\sum_{e'\in\mathcal{P}_n}[[e',e]]\log U(e')\\
                &=-\frac{1}{2}\left(\log U(I) + \sum_{e'\in\overline{\mathcal{E}}_{\Gamma}}[[e',e]]\log U(e')\right)\\
                &=\sum_{e'\in\overline{\mathcal{E}}_{\Gamma}}\frac{-[[e',e]]+1}{2}\log U(e')\label{equation: source of J minus W}\\
                &=\sum_{e'\in\overline{\mathcal{E}}_{\Gamma}}\langle e',e\rangle\log U(e')\\
                &=\sum_{e'\in\overline{\mathcal{E}}_{\Gamma}}\langle e',e\rangle\log \overline{\mu}_{e'}.
\end{align}
Here, the third equality follows from \cref{equation: dependence on log U(I)}. We define the binary commutator for Pauli operators $\langle A,B\rangle$ as
\begin{equation}
    \langle A,B\rangle=\begin{cases}
        0&\text{if $A$, $B$ commute}\\
        1&\text{else.}\\
    \end{cases}
\end{equation}
Therefore, for any set $\mathcal{H}$ of Pauli operators, we can write
\begin{equation}\label{equation: eigenvalues and mu bar}
    \log \vec{\Lambda}^{(\mathcal{H})}=\overline{D}^{(\mathcal{H})} \log \vec{\overline{\mu}},
\end{equation}
where $\overline{D}^{(\mathcal{H})}$ is a $|\mathcal{H}|$ by $|\overline{\mathcal{E}}_{\Gamma}|$ matrix with entries: $\overline{D}^{(\mathcal{H})}[h,e]=\langle h,e\rangle, h\in\mathcal{H}, e\in\overline{\mathcal{E}}_{\Gamma}$.

Similarly, to parametrize local channels, we define the set $\mu=\{U_{\gamma_e}(e)|e\in \mathcal{E}_{\Gamma}\}$ and the corresponding vector $\vec{\mu}=(U(e)|e\in \mu)$ and we have
\begin{equation}
    \log \vec{\Lambda}^{(\mathcal{H})}=D^{(\mathcal{H})} \log \vec{\mu},\label{equation: syndrome matrix}
\end{equation}
where $D^{(\mathcal{H})}$ is a $|\mathcal{H}|$ by $|\mathcal{E}_{\Gamma}|$ matrix with entries: $D^{(\mathcal{H})}[h,e]=\langle h,e\rangle, h\in\mathcal{H}, e\in\mathcal{E}_{\Gamma}$. We refer to the matrix $D$ as the \emph{syndrome matrix} for the collection of local channels.

Note that there is intrinsic gauge freedom in $\log \mu$ of the form similar to \cref{equation: gauge freedom}:
\begin{equation}
    \log U_{\gamma_1}(e) \longrightarrow \log U_{\gamma_1}(e)+\alpha \quad \text{and} \quad \log U_{\gamma_2}(e) \longrightarrow \log U_{\gamma_2}(e)-\alpha,
\end{equation}
whenever we have two non-overlapping distinct channel supports $\gamma_1,\gamma_2\in \Gamma$ and Pauli error $e$ such that $\mathrm{supp}(e)\subseteq \gamma_1\cap \gamma_2$. However, working with the local channels using the transformed local eigenvalues $\log \mu$ is more natural for physical level benchmarking, since $\log \mu$ is an approximation of the Pauli error rates associated with physical circuit operation (see \cref{appendix: physical proof}).

In \cref{appendix: A rank general}, we prove the rank of matrices $A$ and $D$, which gives the minimal number of measured stabilizer elements needed to solve for the system of linear equations.

\section{\label{appendix: A rank general}Rank of matrix $A$ and $D$}

Given a Pauli error model with a set of supports $\Gamma$ and the group of measured stabilizer $\mathcal{M}$, we constructed \cref{equation: log equation} that connects local Pauli eigenvalues to the syndrome expectation values. In this section, we present more details of the structure of this system of linear equations represented by the matrix $A^{(\mathcal{M})}$ by performing a basis transform to the syndrome matrix $D^{(\mathcal{M})}$ (\cref{equation: syndrome matrix} in \cref{appendix: log u basis}) and show that the rank of $A^{(\mathcal{M})}$ and $D^{(\mathcal{M})}$ is the number of nontrivial syndrome classes $|\mathcal{C}^*|$. Therefore one only need to select a subset $\mathcal{M}'\subseteq\mathcal{M}$ of $|\mathcal{C}^*|$ many stabilizer elements such that the rows among $A^{(\mathcal{M}')}$ and $D^{(\mathcal{M}')}$ are all independent. We give a formal definition of syndrome class with respect to a set of Pauli operators:
\begin{appdefinition}[Syndrome error class]\label{definition: syndrome class}
    Given a set of Pauli operators $\mathcal{H}$ and the set of Pauli errors $\mathcal{E}_{\Gamma}$, we define the syndrome vector for any $e\in\mathcal{E}_{\Gamma}$ as $\vec{s}(e)=\left(\langle e, h\rangle \right)_{h\in\mathcal{H}}$. We define the set of syndrome classes $\mathcal{C}^{(\mathcal{H})}$ of $\mathcal{E}_{\Gamma}$ with respect to $\mathcal{H}$ as a set of equivalent classes of $\mathcal{E}_{\Gamma}$ induced by the relation $e\sim e'$ if and only if $\vec{s}(e)=\vec{s}(e')$. We denote by $C_0$ the trivial syndrome class in $\mathcal{C}^{(\mathcal{H})}$, i.e., the class $\{e\in\mathcal{E}_{\Gamma}|\vec{s}(e)=\vec{0}\}$. And we denote the set of nontrivial classes $\mathcal{C}^{*(\mathcal{H})}=\mathcal{C}^{(\mathcal{H})}\backslash \{C_0\}$. When the set $\mathcal{H}$ is the measured stabilizer group $\mathcal{M}$, we omit the superscript.
\end{appdefinition}

From \cref{equation: source of J minus W}, we note that matrix $A$ and the syndrome matrix $D$ differ by a full rank
transformation:
\begin{applemma}
$\rank D^{(\mathcal{H})}=\rank A^{(\mathcal{H})}$ for any set of Pauli operators $\mathcal{H}$.
    \label{lemma: rank equal}
\end{applemma}
\begin{proof}
The full rank matrix $V$ of the linear transformation $D^{(\mathcal{H})}=A^{(\mathcal{H})}V$ is given by:
    \begin{equation}
        V=\bigoplus_{\gamma\in\Gamma}(J_{4^{|\gamma|}-1}-W'_{4^{|\gamma|}})/2
        \label{equation: Matrix V expression}
    \end{equation}
    where we denote $W'_{4^{|\gamma|}}$ as the order $4^{|\gamma|}$ Walsh-Hadamard matrix deleting the first row and first column. $J_n$ is a $n$ by $n$ matrix with all entries being $1$. For example, consider $V$ on a single qubit, then we have
    \begin{equation}
        W_4=\begin{pmatrix}
1 & 1 & 1 & 1 \\
1 & 1 & -1 & -1 \\
1 & -1 & 1 & -1 \\
1 & -1 & -1 & 1 \\
\end{pmatrix}
\mapsto
W'_4=\begin{pmatrix}
 1 & -1 & -1 \\
 -1 & 1 & -1 \\
 -1 & -1 & 1 \\
\end{pmatrix}\label{equation: reduced WH}
\end{equation}
and
\begin{equation}
    V=\frac{1}{2}
\begin{pmatrix}
 1 & 1 & 1 \\
 1 & 1 & 1 \\
 1 & 1 & 1 \\
\end{pmatrix}
-\frac{1}{2}
\begin{pmatrix}
 1 & -1 & -1 \\
 -1 & 1 & -1 \\
 -1 & -1 & 1 \\
\end{pmatrix}
=
\begin{pmatrix}
 0 & 1 & 1 \\
 1 & 0 & 1 \\
 1 & 1 & 0 \\
\end{pmatrix}.
\end{equation}
For example, the syndrome matrix $D^{(\mathcal{H})}$ of single-qubit errors can be constructed in such a way where we look at the three columns of any qubit, its $X$ column is the sum of the $Y$ and $Z$ column of $A^{(\mathcal{H})}$, its $Y$ column is the sum of the $X$ and $Z$ column of $A^{(\mathcal{H})}$, and its $Z$ column is the sum of the $X$ and $Y$ column of $A^{(\mathcal{H})}$. \cref{equation: Matrix V expression} generalizes this procedure to an arbitrary set of local channel supports $\Gamma$.
\end{proof}

Since $D^{(\mathcal{M})}$ encodes the commutation of each error with the measured stabilizer group, the columns of $D$ for errors in the same syndrome error class are identical. We define the reduced syndrome matrix $D'$ by removing columns of these redundant errors in $D$:
\begin{appdefinition}
    The reduced syndrome matrix $D'^{(\mathcal{H})}$ of a given set of Pauli error support $\Gamma$ and a set of Pauli operators $\mathcal{H}$ is a $|\mathcal{H}|$ by $|\mathcal{C}^{*(\mathcal{H})}|$ matrix where each column correspond to one representative in every nontrivial syndrome error class $C\in\mathcal{C}^{*(\mathcal{H})}$, and the syndrome classes are defined for the set $\mathcal{E}_{\Gamma}$. The reduced syndrome matrix $\overline{D}'^{(\mathcal{H})}$ for the total channel is defined similarly but the syndrome classes $\mathcal{C}^{*(\mathcal{M})}$ is defined for the error set $\overline{\mathcal{E}}_{\Gamma}$ instead.
    \label{definition: reduced D}
\end{appdefinition}

\begin{applemma}
    $\rank \overline{D}^{(\mathcal{M})} = \rank \overline{D}'^{(\mathcal{M})} = \rank D^{(\mathcal{M})} = \rank D'^{(\mathcal{M})} = |\mathcal{C}^{*(\mathcal{M})}|$ for any subgroup of the Pauli group $\mathcal{M}\leq\mathcal{P}_n$, where $\mathcal{C}^{*(\mathcal{M})}$ is the set of nontrivial syndrome classes (\cref{definition: syndrome class}) for the set of errors $\mathcal{E}_{\Gamma}$.
    \label{lemma: rank D}
\end{applemma}
\begin{proof}
    The first three equalities are trivial from \cref{definition: reduced D} and \cref{equation: eigenvalues and mu bar,equation: syndrome matrix}. For the last equality, one could show that $(D'^{(\mathcal{M})})^T(D'^{(\mathcal{M})})$ is full rank. The observation is that
    \begin{align}
    [&(D'^{(\mathcal{M})})^T(D'^{(\mathcal{M})})]_{i,j}\\
    =&\left|\{M\in\mathcal{M}|[[M,e_i]]=[[M,e_j]]=-1\}\right|\\
    =&
    \begin{cases}
        |\mathcal{M}|/2& \text{if }i=j,\\
        |\mathcal{M}|/4& \text{if }i\neq j.\\
    \end{cases}\label{equation: anticommute counting}
    \end{align}
    Here, $e_i$ is the error correspond to the $i$th column of $D'^{(\mathcal{M})}$. To justify the last equality \cref{equation: anticommute counting}, we select a generating set $\mathcal{Q}$ of $\mathcal{M}$ such that $\mathcal{M}=\langle\mathcal{Q}\rangle$ and so that $|\mathcal{M}|=2^{|\mathcal{Q}|}$. Then every $M\in\mathcal{M}$ corresponds to an element $\mathcal{Q}^{(M)}$ of the power set $2^{\mathcal{Q}}$:
    \begin{equation}
        M=\prod_{Q\in \mathcal{Q}^{(M)}}Q.
    \end{equation}
    We define $\mathcal{Q}_{e_k}=\{Q\in\mathcal{Q}|\,[[Q,e_k]]=-1\}$, i.e., the set of generators that anticommute with $e_k$.
    To find the number of $M$ that anticommute with $e_i$, $\mathcal{Q}^{(M)}$ has to satisfy:
    \begin{equation}
    |\mathcal{Q}^{(M)}\cap \mathcal{Q}_{e_i}|\equiv 1 \mod 2,    
    \end{equation}
    so the number of choices of $\mathcal{Q}^{(M)}$ is:
    \begin{equation}
        2^{|\mathcal{Q}|-|\mathcal{Q}_{e_i}|}\sum_{\text{odd }j}\binom{|\mathcal{Q}_{e_i}|}{j}=2^{|\mathcal{Q}|-1}=|\mathcal{M}|/2.
    \end{equation}
    We denote by $\mathcal{Q}_1\backslash \mathcal{Q}_2$ the relative complement of $\mathcal{Q}_1$ and $\mathcal{Q}_2$, i.e., the set of elements in $\mathcal{Q}_1$ that are not in $\mathcal{Q}_2$. To find the number of $M$ that anticommute with $e_i$ and $e_j$ ($i\neq j$), $\mathcal{Q}^{(M)}$ has to satisfy:
    \begin{align}
        |\mathcal{Q}^{(M)}\cap \mathcal{Q}_{e_i}\cap\mathcal{Q}_{e_j}|+|\mathcal{Q}^{(M)}\cap(\mathcal{Q}_{e_i}\backslash\mathcal{Q}_{e_j})|
        \equiv|\mathcal{Q}^{(M)}\cap \mathcal{Q}_{e_i}\cap\mathcal{Q}_{e_j}|+|\mathcal{Q}^{(M)}\cap(\mathcal{Q}_{e_j}\backslash\mathcal{Q}_{e_i})|
        \equiv 1\mod 2,
        \label{equation: condition of anticommute both}
    \end{align}
    so one could first choose any subset of generators in $(\mathcal{Q}_{e_i}\cup\mathcal{Q}_{e_j})^c$ and $\mathcal{Q}_{e_i}\cap\mathcal{Q}_{e_j}$, and then select generators in $\mathcal{Q}_{e_i}\backslash\mathcal{Q}_{e_j}$ and $\mathcal{Q}_{e_j}\backslash\mathcal{Q}_{e_i}$ accordingly to satisfy \cref{equation: condition of anticommute both}. The number of choices of $\mathcal{Q}^{(M)}$ is therefore:
    \begin{align}
        &2^{|Q|-|\mathcal{Q}_{e_i}\cup\,\mathcal{Q}_{e_j}|+|\mathcal{Q}_{e_i}\cap\mathcal{Q}_{e_j}|}\left(\sum_{\text{odd (even) }r}\binom{|\mathcal{Q}_{e_i}\backslash\mathcal{Q}_{e_j}|}{r}\right)\left(\sum_{\text{odd (even) }s}\binom{|\mathcal{Q}_{e_j}\backslash\mathcal{Q}_{e_i}|}{s}\right)\\
        =&|\mathcal{M}|/4.
    \end{align}
    Therefore from \cref{equation: anticommute counting} we have:
    \begin{equation}
        (D'^{(\mathcal{M})})^T(D'^{(\mathcal{M})})=\frac{|\mathcal{M}|}{4}(J_{|\mathcal{C}^*|}+I_{|\mathcal{C}^*|}),
    \end{equation}
    where $J_n$ is the $n$ by $n$ matrix of all ones and $I_n$ is the $n$ by $n$ identity matrix. As a result, $(D'^{(\mathcal{M})})^T(D'^{(\mathcal{M})})$ is full rank thus $D'^{(\mathcal{M})}$ has full column rank which is $|\mathcal{C}^{*(\mathcal{M})}|$.
\end{proof}

\begin{appcorollary}
    $\rank A^{(\mathcal{M})}=\rank D^{(\mathcal{M})}=|\mathcal{C}^{*(\mathcal{M})}|$.
    \label{lemma: rank of A}
\end{appcorollary}
\begin{proof}
    This follows directly from \cref{lemma: rank equal} and \cref{lemma: rank D}.
\end{proof}

\section{\label{appendix: minimal equations}Minimal number of equations for the syndrome data}
In this section, we explain the algorithm used to find the minimum subset $\mathcal{M}'$ of stabilizer elements needed for the noise learning. This is motivated because to solve \cref{equation: log equation} or \cref{equation: log equation mu}:
\begin{align}
    \log\vec{\Lambda}^{(\mathcal{M})}&=A^{(\mathcal{M})}\log \vec{\lambda},\\
    \log\vec{\Lambda}^{(\mathcal{M})}&=D^{(\mathcal{M})}\log \vec{\mu},
\end{align}
we encounter the scaling issue as there are exponentially many equations ($|\mathcal{M}|=2^{n-k}$) if assuming full syndrome is measured, i.e., $\mathcal{M}=\mathcal{S}$. (See the relation of the two sets of equations in \cref{appendix: physical proof}). However, we focus on the case when $\mathcal{N}_{\Gamma}$ is a ($r_{\Gamma}, c_{\Gamma}$)-Pauli channel, i.e., local Pauli noise models with $|\gamma|_{\max}=r_{\Gamma}=O(1)$ and $|\Gamma|=\mathcal{O}(n)$. In this local noise model setting, the total number of noise model parameters scales linearly with the system size $n$. From  \cref{definition: reduced D}, \cref{lemma: rank equal}, \cref{lemma: rank D}, and \cref{lemma: rank of A} we have:
\begin{align}
    \text{rank }A^{(\mathcal{M})}=\text{rank }D^{(\mathcal{M})}=\text{rank }A^{(\mathcal{M'})}=\text{rank }D^{(\mathcal{M'})}=\text{rank }D'^{(\mathcal{M'})}=|\mathcal{C}^*|
\end{align}
where $|\mathcal{C}^*|$ is the number of nontrivial syndrome classes that also scales linearly with $n$ in this setting.

Here, we show an algorithm to select a subset $\mathcal{M}'\subseteq\mathcal{M}$ of size $|\mathcal{C}^*|$ satisfying $\rank A^{(\mathcal{M}')} =|\mathcal{C}^*|$ so that one could solve the following linear size system of equations (\cref{equation: log equation reduced} in main text) instead:
\begin{align}
    \log\vec{\Lambda}^{(\mathcal{M}')}&=A^{(\mathcal{M}')}\log \vec{\lambda}\\
    \log\vec{\Lambda}^{(\mathcal{M})}&=D^{(\mathcal{M}')}\log \vec{\mu}.
\end{align}
 We find $\mathcal{M}'$ in the reduced search space since one only needs to consider products of "correlated" stabilizer generators, as shown in proposition \ref{proposition: correlated stabilizer suffice}. 
\begin{applemma}
    Given a set of generators of $\mathcal{M}$ and the set of error supports $\Gamma$, we call a pair of elements in $\mathcal{M}$ ``correlated'' if they both overlap with some region $\gamma\in\Gamma$. To construct matrix $A^{(\mathcal{M}')}$ with maximal rank, it suffices to pick stabilizer elements from a product of generators such that for any $M_i$ in the product, there exists $M_j\neq M_i$ in the factors of the product that is correlated with $M_i$.
    \label{proposition: correlated stabilizer suffice}
\end{applemma}
\begin{proof}
We denote the row corresponding to element $H\in\mathcal{H}$ of the matrix $A^{(\mathcal{H})}$ as $A^{(\mathcal{H})}[H]$.
    For any $M\in\mathcal{M}$ that corresponds to an independent row of $A$ and can be written as $M=M_iM_j$ for some uncorrelated pair $M_i, M_j \in \mathcal{M}$, i.e., $\text{supp}(M_1)\cap \text{supp}(M_2)\cap \gamma=\emptyset,\forall\gamma\in\Gamma$, we have $\Lambda(M)=\Lambda(M_i)\Lambda(M_j)$ and $A[M]=A[M_i]+A[M_j]$. Therefore, replacing $M$ by $M_i$ is a row operation on $A$ unless $M_i$ is linearly dependent on other rows, in which case one could replace $M$ by $M_j$. We could repeat the process until all rows of $A$ are the product of overlapping stabilizer generators without decreasing the rank of $A$.
\end{proof}

This motivates us to define the stabilizer correlation graph which the algorithm \ref{algorithm: reduced stabilizer group} utilizes, a graph that encodes the correlation between a set of measured stabilizer generators given the individual error events determined by $\Gamma$:  
\begin{appdefinition}[Stabilizer correlation graph]
    Given a set of measured stabilizer generators, the stabilizer correlation graph $G$ is a graph with vertices corresponding to the measured stabilizer generators. Two vertices share an edge if and only if the two generators have overlapping support $\gamma\in\Gamma$.
    \label{correlation_graph}
\end{appdefinition}

Given the graph $G$, all elements of $\mathcal{M}$ satisfying the condition of  Lemma \ref{proposition: correlated stabilizer suffice} correspond to the connected subgraphs of $G$. One can implement the algorithm below and guarantee to output a subset $\mathcal{M}'\subseteq\mathcal{M}$ with $A^{(\mathcal{M}')}$ reaches the maximal rank $\rank A^{(\mathcal{M}')}=\rank A^{(\mathcal{M})}=|\mathcal{C}^*|$. One can use similar algorithms based on the reduced syndrome matrix $D'$.

\begin{algorithm}[H]
\caption{\label{algorithm: reduced stabilizer group}Constructing the matrix $A^{(\mathcal{M}')}$}
\KwIn{
\begin{enumerate}
    \item Local generators of $S$
    \item Syndrome error classes $\mathcal{C}(\mathcal{S},\Gamma)$ and its nontrivial classes $\mathcal{C}^*$.
\end{enumerate}
}
Construct the stabilizer correlation graph $G$ from Def. \ref{correlation_graph}.\\
Let $v=0$, $\mathcal{M}'=\{\}$.\\
Let $A^{(\mathcal{M}')}$ be the empty matrix.\\
\While{$\rank A^{(\mathcal{M}')} < |\mathcal{C}^*|$}
{
$v = v + 1$\\
\For{connected subgraph $g$ of order $v$ in $G$}
{
Add stabilizer elements $M$ that correspond to $g$ to $\mathcal{M}'$.\\
Construct matrix $A^{(\mathcal{M}')}$.\\
}
}
\Return{$\mathcal{M}'$ and $A^{(\mathcal{M}')}$}
\end{algorithm}

\begin{applemma}
    The runtime of \cref{algorithm: reduced stabilizer group} is $\text{Poly}(n)$ for qLDPC codes under local sparse Pauli channel $\mathcal{N}_{\Gamma}$.
\end{applemma}
\begin{proof}
    The upper bound of the order $v_{\mathrm{max}}$ of the connected subgraph needed is given by the maximum number of stabilizer generators that a single local error can flip, as shown in \cref{equation: recursive solving} and proven in \cref{lemma: recursive solving solves the linear equation}.
\end{proof}

 For each order $v$ of the connected subgraphs being searched, the runtime of the algorithm scales $\text{Poly}(n)$ and $\exp(v)$. However, the upper bound of the order $v_{\mathrm{max}}$ of the connected subgraph needed is given by the maximum number of stabilizer generators that a single local error can flip, as shown in \cref{equation: recursive solving} in \cref{lemma: recursive solving solves the linear equation}. Therefore, the maximum order $v_{\mathrm{max}}$ of the search is a constant for qLDPC codes. For example, the upper bound of the order $v$ is $4$ for single-qubit errors on the rotated surface code of any size using the typical local stabilizer generators.

The learning algorithm can be summarized as solving two sets of equations, one is the system of equations characterized by $A^{(\mathcal{M}')}$ or $D'^{(\mathcal{M}')}$ in \cref{equation: log equation reduced,equation: D matrix reduced} with respect to the log of the local Pauli eigenvalues $\log \vec{\lambda}$. The other one is the extra constraints characterized by matrix $B$ in \cref{equation: extra constraints}, which are linear constraints in the probability distribution/local Pauli eigenvalues $\vec{\lambda}$.
The number of rows of $B$, i.e., the total number of these extra constraints, is $|C_0|+\sum\limits_{C\in\mathcal{C}^{*}}(|C|-1)$. The number of independent equations from syndrome expectation values is $|\mathcal{C}^*|$. As a result, the total number of constraints equals:
\begin{align}  
&\#\,\text{constraints}\\
=&|C_0|+\sum\limits_{C\in\mathcal{C}^{*}}(|C|-1)+|\mathcal{C}^*|\\
=&|\mathcal{{E}}_{\Gamma}|-1
\end{align}
which is exactly the degrees of freedom in the Pauli noise model. In \cref{appendix: existence of solution with B constraints}, we prove that a solution under these two sets of equations must exist for low error rates.

\section{\label{appendix: physical proof}Proof of syndrome class error rate learnability condition}

To understand the learnability of the total channel, we first look at the parametrization of the total channel $\mathcal{E}_{\Gamma}$ with the set of transformed local eigenvalues of the total channel $\overline{\mu}$ defined from the $\log U$ basis. In particular, we have the logarithm of Pauli eigenvalues of $\mathcal{N}_{\Gamma}$ written as a linear function of $\log \overline{\mu}$ as in \cref{equation: eigenvalues and mu bar}:
\begin{equation}
    \log \vec{\Lambda}^{(\mathcal{H})}=\overline{D}^{(\mathcal{H})} \log \vec{\overline{\mu}}.
\end{equation}

\begin{applemma}\label{lemma: valid mu is an openset}
    For a Pauli channel $\mathcal{N}_{\Gamma}$ with a support set $\Gamma$ satisfying for any $\gamma\in\Gamma, e\in\mathcal{E}_{\gamma}$ $P_{\gamma}(e)>0$ and $P_{\gamma}(I)>1/2$ (\cref{assumption: open set probability}), the set of transformed eigenvalues $\vec \mu$ for the collection of individual channels $\{\mathcal{N}_{\gamma}\}_{\gamma\in\Gamma}$ is an open set in $\mathbb{R}^{|\mathcal{E}_{\Gamma}|}$. Similarly, the set of transformed eigenvalues $\vec{ \overline{\mu}}$ for the total channel $\mathcal{N}_{\Gamma}$ is an open set in $\mathbb{R}^{|\overline{\mathcal{E}}_{\Gamma}|}$.
\end{applemma}
\begin{proof}
    Given \cref{assumption: open set probability}, we have all the possible local error rates $\vec{p}$ as an open set in $\mathbb{R}^{|\mathcal{E}_{\Gamma}|}$. And we have
    \begin{equation}
        \vec{\mu}=V^{-1}\log\vec{\lambda}=V^{-1}\log (\vec{1}-2V\vec{p})
    \end{equation}
    where $V$ is the invertible matrix in $\cref{equation: Matrix V expression}$. Therefore, the map between $\vec{p}$ and $\vec{\mu}$ is a continuous bijection, and therefore the set of possible transformed eigenvalues $\vec \mu$ is also an open set in $\mathbb{R}^{|\mathcal{E}_{\Gamma}|}$.

    In addition, the logarithm of the transformed eigenvalues for the total $\log\vec{\overline{\mu}}$ of the total channel is related to $\log\vec{\mu}$ by a surjective maps (see the first properties of \cref{lemma: properties of U}), so the set of possible $\vec{\overline{\mu}}$ is also an open set, in $\mathbb{R}^{|\overline{\mathcal{E}}_{\Gamma}|}$.
\end{proof}

\noindent\textbf{\cref{theorem: syndrome class learning} (Restated)}.  
Let a total Pauli error channel $\mathcal{N}_{\Gamma}$ be a composition of Pauli channels $\mathcal{N}_{\gamma}$ with support $\gamma\in\Gamma$. Under \cref{assumption: open set probability}, we have:
    \begin{enumerate}
        \item The total channel $\mathcal{N}_{\Gamma}$ can be learned if and only if every error $e\in\overline{\mathcal{E}}_{\Gamma}$ has a unique and nontrivial syndrome.
        \item Each individual channel $\mathcal{N}_{\gamma}$ is learnable if and only if every $e\in\mathcal{E}_{\gamma}$ produces a syndrome that is nontrivial and unique among the syndromes of all errors in $\mathcal{E}_{\Gamma}$.
        \item For all nontrivial syndrome error classes $C\in\mathcal{C}^*$, the sum of syndrome error class error rates $P_C=\sum_{e\in C} p_{e}$ can be learned up to an $\mathcal{O}(\|\vec{p}^{(\Gamma_C)}\|^2)$ deviation.
    \end{enumerate}

\begin{proof}
To prove the first statement regarding the learnability condition of the total channel $\mathcal{N}_{\Gamma}$, note that there is an invertible linear map (Walsh-Hadamard transform) between the logarithm of Pauli eigenvalues $\log \Lambda(e)$ of $\mathcal{N}_{\Gamma}$ and the transformed eigenvalues $\log U(e)$ for all $e\in\mathcal{P}_n$. On the other hand, given the restricted structure of the total channel $\mathcal{N}_{\Gamma}$ with the set of support $\Gamma$, we proved in \cref{appendix: log u basis} that the transformed eigenvalues $\log U(e)=0$ if $e\neq I$ and $e\not\in\overline{\mathcal{E}}_{\Gamma}$, and $\log U(I)$ depends on $\{\log U(e)|e\in\overline{\mathcal{E}}_{\Gamma}\}$. Therefore, we define a local subset of the transformed eigenvalues $\overline{\mu}=\{\log U(e)|e\in\overline{\mathcal{E}}_{\Gamma}\}$ and two total channels $\mathcal{N}_{\Gamma}$ and $\mathcal{N}'_{\Gamma}$ are equal if and only if their corresponding "$\overline{\mu}$ set" $\overline{\mu}$ and $\overline{\mu}'$ are equal.

On the other hand, we have the relation between the syndrome expectation and the transformed local eigenvalues $\vec{\overline{\mu}}$ for the total channel
\begin{equation}
\log\vec{\Lambda}^{(\mathcal{M})}=\overline{D}^{(\mathcal{M})}\vec{\overline{\mu}}.
\end{equation}
where $\overline{D}^{(\mathcal{M})}$ is the $|\mathcal{M}|$ by $|\overline{\mathcal{E}}_{\Gamma}|$ syndrome matrix of the error set $\overline{\mathcal{E}}_{\Gamma}$. First, note that the set of expectation values of all the measured stabilizer elements $\vec{\Lambda}^{(\mathcal{M})}$ is a sufficient statistic of the syndrome data. This is because it is a Walsh-Hadamard transformation of the probability distribution over the $m$-bit syndrome, where $m$ is the number of generators of the measured stabilizer group $\mathcal{M}$.

According to \cref{lemma: rank D}, the rank of $\overline{D}^{(\mathcal{M})}$ is the number of nontrivial syndrome classes of $\mathcal{E}_{\Gamma}$ (or equivalently $\overline{\mathcal{E}}_{\Gamma}$). In addition, if the values of the transformed local eigenvalues $\vec{\overline{\mu}}$ is valid given the assumption about the channel $\mathcal{N}_{\Gamma}$, then the neighborhood of $\vec{\overline{\mu}}$ is also valid according to \cref{lemma: valid mu is an openset}. As a result, given the syndrome expectation values $\vec{\Lambda}^{(\mathcal{M})}$, we have a continuous family of valid parametrization $\log \vec{\mu'}=\log \vec{\mu} + \vec{\epsilon}$ around a valid point $\log \vec{\mu}$, where $\vec{\epsilon}\in \mathrm{ker}\left(\overline{D}^{(\mathcal{M})}\right)$. Therefore, we can conclude that $\overline{\mu}$ and thus the total channel $\mathcal{N}_{\Gamma}$ is learnable if and only if the number of nontrivial syndrome classes is equal to $|\overline{\mathcal{E}}_{\Gamma}|$, i.e., all syndromes in $\overline{\mathcal{E}}_{\Gamma}$ are nontrivial and distinct.
The second statement regarding the learnability of local channels can be proven similarly by arguing that all transformed local eigenvalues $\mu$ from a local channel $\mathcal{N}_{\gamma}$ can be learned if and only if the errors in $\mathcal{E}_{\gamma}$ has distinct and nontrivial syndrome among all syndromes of $\mathcal{E}_{\Gamma}$ by looking at columns of the syndrome matrix $D^{(\mathcal{M})}$ in \cref{equation: log equation mu}.

To prove what we can learn about the local channels when the necessary and sufficient condition of learnability potentially fails, we look at the transformation of the basis in \cref{equation: log equation} regarding all the local Pauli eigenvalues $\vec{\lambda}$:
\begin{equation}
    \log\vec{\Lambda}^{(\mathcal{M})}=A^{(\mathcal{M})}\log \vec{\lambda}.
\end{equation}
to this basis of the syndrome matrix $D^{(\mathcal{H})}$ with the transformed eigenvalues in \cref{equation: syndrome matrix}:
\begin{align}
    \log\vec{\Lambda}^{(\mathcal{M})}=D^{(\mathcal{M})}\log \vec{\mu},\label{equation: D matrix appendix}
\end{align}
where the transformation $V$ is given in \cref{equation: Matrix V expression} so that we have:
\begin{align}
    D^{(\mathcal{M})}&=A^{(\mathcal{M})}V\\
    \log\vec{\mu}&=V^{-1}\log \vec{\lambda}\label{equation: mu from lambda}.
\end{align}
Explicitly, $V^{-1}$ is the direct sum of a constant times the Walsh-Hadamard matrix (\cref{equation: reduced WH}):
\begin{equation}
    V^{-1}=-2\bigoplus_{\gamma\in\Gamma}\frac{1}{4^{|\gamma|}}W'_{4^{|\gamma|}}. \label{equation: V inverse}\end{equation}

In general, for an error $e$ supported on $\gamma\in\Gamma$, the transformed eigenvalues $\mu_{e}$ can be written as a fraction where the numerator/denominator are products of Pauli eigenvalues of errors supported on $\gamma$ that anticommute/commute with $e$ respectively, then raised the power of $\frac{1}{2^{(2|\gamma|-1)}}$. Moreover, if $e$ has a unique syndrome, then the column of $D^{(\mathcal{M})}$ corresponding to $e$ is independent of other columns (\cref{lemma: rank D}). Since $D^{(\mathcal{M})}$ encodes the syndrome of all local errors $e\in\mathcal{E}_{\Gamma}$ in its columns, the learnable quantities are therefore the sum of $\log \mu$ of each syndrome error class. Denoting $\mathcal{C}^{*}$ as the set of nontrivial syndrome classes, we can learn
\begin{equation}
    \{\sum_{e\in C}\log\mu_{e}|C\in\mathcal{C}^{*}\}.
\end{equation}

At low physical error rate for each local channel $\gamma\in\Gamma$, $\log \lambda=\lambda - 1 + \mathcal{O}(\|\vec{p}\|^2)$ and  $\log \mu$ is just the local Walsh-Hadamard transform of $\log \lambda$. Therefore according to \cref{equation: V inverse} we have:
\begin{equation}
    -\frac{1}{2}\log \mu_{e}=p_{e}+\mathcal{O}(\|\vec{p}^{(\gamma)}\|^2)
\end{equation}
for any $e\in\mathcal{E}_{\Gamma}$, where $\vec{p}^{(\gamma)}=(p_{e})_{e\in\mathcal{E}_{\gamma}}$. For each syndrome error class $C\in\mathcal{C}^*$, we have the learnable quantity referred to as the learnable eigenvalues
\begin{equation}
    \log \nu_C = \sum_{e\in C}\log \mu_{e}=-2\sum_{e\in C}p_{e}+\mathcal{O}(\|\vec{p}^{(\Gamma_C)}\|^2).\label{equation:syndrome error true}
\end{equation}
A straightforward approximate estimator for the sum of syndrome class error rates is given by
\begin{equation}
    \sum_{e\in C}\tilde{p}_{e}=-\frac{1}{2}\log\nu_C
\end{equation}
and the magnitude of the bias is of second order in the local error rates, i.e., $|\sum_{e\in C}\tilde{p}_{e}-\sum_{e\in C}p_{e}|=\mathcal{O}(\|\vec{p}^{(\Gamma_C)}\|^2)$.
\end{proof}

\begin{applemma}[Independent detector error model learnability]
    If the local Pauli channel can be written as a composition of independent single-Pauli channels:
\begin{equation}
    \mathcal{N}(\rho)=\circ_{e\in\mathcal{E}_{\Gamma}}\mathcal{N}^{(e)}(\rho)
\end{equation}
where $\mathcal{N}^{(e)}(\rho)=(1-p_{e})\rho+p_{e} e\rho e^{\dagger}$. Then the probability of an odd number of errors in each syndrome error class, or the detector error rate \cite{takou2025estimating}, can be exactly learned.
\end{applemma}
\begin{proof}
    Note that if each local channel can be further decomposed into a composition of channels $\mathcal{N}^{(e)}$ involving a single Pauli error $e$, then each transformed eigenvalue $\mu^{(e)}$ is of the simple form:
\begin{equation}
    \mu^{(e)}_{e'} = \begin{cases}
        1-2p_e & \text{if } e=e',\\
        1 & \text{else.}
    \end{cases}
\end{equation}
Then the detector error rate associated with the syndrome error class $C$ is:
\begin{equation}
    P^{\mathrm{detect}}_C=\frac{1-\prod_{e\in C}(1-2p_{e})}{2}=\frac{1}{2}(1-\nu_C)\label{equation: detector error rates learnability}
\end{equation}
which is learnable as shown in \cref{equation: learnables}.
\end{proof}

\begin{appexample}
    For single-qubit Pauli channels $\mathcal{N}_{\{1\}}(\rho)=\sum\limits_{e\in\{I,X_1,Y_1,Z_1\}}p_{e} e\rho e^{\dagger}$, the transformed eigenvalue $\mu_{X_1}$ is
\begin{equation}
    \mu_{X_1}=\left(\frac{\Lambda(Y_1)\Lambda(Z_1)}{\Lambda(X_1)}\right)^{\frac{1}{2}}.
\end{equation}
But when the Pauli channel contains only one error, e.g., $\mathcal{N}^{(X_1)}(\rho)=(1-p_{X_1})\rho + p_{X_1}X_1\rho X_1^{\dagger}$, we have
\begin{equation}
   \Lambda^{(X_1)}(e)=(-1)^{\langle X_1,e\rangle}(1-2p_{X_1})=
\begin{cases}
1 & \text{if } X_1 \text{ and }e \text{ commute},\\
1-2p_{X_1} & \text{if } X_1 \text{ and }e \text{ anticommute}.
\end{cases}
\end{equation}
So we have $\mu^{(X_1)}_{X_1}=1-2p_{X_1}$ and $\mu^{(X_1)}_P=1$ for other single-qubit Pauli P.
\end{appexample}
As a result, given a set of local errors $\mathcal{E}_{\Gamma}$, \cref{equation: log equation mu} can be reduced to:
\begin{equation}
    \log\begin{bmatrix}
\Lambda(M_1) \\
\Lambda(M_2) \\
\vdots \\
\end{bmatrix}=
\begin{bmatrix}
\langle M_1, e_1 \rangle & \langle M_1, e_2 \rangle & \cdots \\
\langle M_2, e_1 \rangle & \langle M_2, e_2 \rangle & \cdots \\
\vdots & \vdots & \ddots
\end{bmatrix}
\log
\begin{bmatrix}
\mu_{e_1} \\
\mu_{e_2} \\
\vdots
\end{bmatrix}
=\begin{bmatrix}
\langle M_1, e_1 \rangle & \langle M_1, e_2 \rangle & \cdots \\
\langle M_2, e_1 \rangle & \langle M_2, e_2 \rangle & \cdots \\
\vdots & \vdots & \ddots
\end{bmatrix}
\log
\begin{bmatrix}
1 - 2p_{e_1} \\
1 - 2p_{e_2} \\
\vdots
\end{bmatrix}
\label{equation: D inclusive}.
\end{equation}
where each $e_i$ are in the set $\mathcal{E}_{\Gamma}$.
As a result, what we could learn is
\begin{equation}
    \nu_{C}=\prod_{e\in C}\mu_{e}=\prod_{e\in C}(1-2p_{e})=\exp(D'^{+}\log\vec{\Lambda}^{(\mathcal{M}')}).
\end{equation}
where $D'^{+}$ is the pseudoinverse of $D'$ in \cref{equation: D' matrix main}. The probability of a syndrome of $C$ being flipped (the detector error rates associated with the hyperedges in a decoder graph) is
\begin{equation}
    P^{\mathrm{detect}}_C=\frac{1}{2}(1-\nu_C).
\end{equation}

\section{Existence of a family of logically equivalent physical error rates\label{appendix: existence of solution with B constraints}}
In this section, we prove that under certain low-error-rate conditions for the true error rates, there exists a family of error rates parametrized by the extra constraints $B$ matrix, i.e., \cref{equation: extra constraints}, and satisfies the system of equations given by the ideal syndrome expectations, i.e., \cref{equation: log equation reduced} or \cref{equation: D matrix reduced}.

By grouping error rates into equivalence classes $C\in\mathcal{C}$ according to their syndrome, we see in \cref{lemma: rank of A} that the syndrome data gives us $|\mathcal{C^{*}}|$ independent linear equations. Here $\mathcal{C^{*}}=\mathcal{C}\backslash C_0$ is the set of nontrivial syndrome error classes. And in \cref{theorem: syndrome class learning}, we show that we can learn the sum of the error rates $P_C$ in each class $C\in\mathcal{C^*}$ up to a second-order deviation in the physical error rates. However, we would like to estimate each error rate associated with each local error $\mathcal{N}_{\gamma}, \forall \gamma\in\Gamma$ for several reasons. For one, it gives us a direct sense of gate fidelity in the circuit, rather than syndrome-class error rates (or detector error rates), which are only natural in decoding because they are related to the weights of edges in the decoder graph. Secondly, fixing a set of error rates from a family of possible error rates determined by the syndrome data is part of the logical error rate learning algorithm, since we estimate the logical error rate via Monte Carlo sampling. Finally, as shown in \cref{subsection: logical sample complexity}, our proof of the bounds on the sampling complexity for learning the logical error rate is based on a single instance (determined by extra constraints) of all possible error rates.

We view the system of equations \cref{equation: D matrix reduced,equation: extra constraints} as a vector-valued function on a real vector space and prove the existence of such a family of error rates using a simplified version of a general quantitative inverse function theorem:

\begin{applemma}[Theorem 4.1 of \textcite{xinghua1999convergence}]\label{lemma: banach quantitative inverse function theorem}
Let $X,Y$ be Banach spaces and let $f:D\subset X\to Y$ be differentiable.
Fix $x_0\in D$ and assume $f'(x_0)^{-1}$ exists. $\mathcal{B}_r(x_0)$
denotes the ball of radius $r$, centered at $x_0$. Let $L>0$ be such that
$\mathcal{B}_{\frac{1}{L}}(x_0)\subset D$ and
\begin{equation}\label{eq:wang48}
\bigl\|\,f'(x_0)^{-1} f'(x) - I\,\bigr\| \le L\,\|x-x_0\|,
\qquad \forall x\in \mathcal{B}_{\frac{1}{L}}(x_0).
\end{equation}
Then the local inverse $f^{-1}_{x_0}$ exists and is differentiable on the open ball of radius $r=\frac{1}{2L\,\|f'(x_0)^{-1}\|}$:
\begin{equation}\label{eq:wang49}
\mathcal{B}_{r}\left(f(x_0)\right) \subset f\!\left(\mathcal{B}_{\frac{1}{L}}(x_0)\right).
\end{equation}
\end{applemma}

\begin{applemma}[Quantitative inverse function theorem for $C^1$ functions over $R^n$]\label{lemma: quantitative inverse function theorem}
Let $f : \mathbb{R}^n \to \mathbb{R}^n$ be continuously differentiable ($C^1$) and $\|\cdot\|$ be a norm on $\mathbb{R}^n$ and also its induced operator norm on matrices. 
Write $J_f(\vec{x})$ as the Jacobian matrix at $\vec{x}$, and $\vec{x}_0 \in \mathbb{R}^n$ is a point such that $J_f(\vec{x}_0)$ is invertible.
Define $\delta= \frac{1}{\|J_f(\vec{x}_0)^{-1}\|} > 0$. Let $R > 0$ and a constant $K>0$ be such that
\begin{equation}
    \|J_f(\vec{x})-J_f(\vec{x}_0)\|\leq K\|\vec{x}-\vec{x}_0\|, \qquad\forall \vec{x} \in \mathcal{B}_R(\vec{x}_0).
\end{equation}
If $R\geq \frac{\delta}{K}$, then there exists a function $g : \mathcal{B}_{\frac{\delta^2}{2K}}(f(\vec{x}_0)) \to \mathcal{B}_{\frac{\delta}{K}}(\vec{x}_0)$ such that
\[
f(g(\vec{y})) = \vec{y}
\quad\text{for all } \vec{y} \in\mathcal{B}_{\frac{\delta^2}{2K}}(f(\vec{x}_0)).
\]
\end{applemma}
\begin{proof}
    Let $f'(x)=J_f(x)$ so that the notation matches the one used in the general case in \cref{lemma: banach quantitative inverse function theorem}. If $\|J_f(x)-J_f(x_0)\|\leq K \|x-x_0\|,\, \forall x\in \mathcal{B}_R(x_0)$, then we have
    \begin{align}
    \|f'(x_0)^{-1}f'(x)-I\|&=\|f'(x_0)^{-1}(f'(x)-f'(x_0))\|\\
    &\leq \|f'(x_0)^{-1}\|\|f'(x)-f'(x_0)\|\\
    &\leq K\|f'(x_0)^{-1}\|\|x-x_0\|\\
    &=L\|x-x_0\|
    \end{align}
    for any $x\in\mathcal{B}_R(x_0)$, where we define $L=K\|f'(x_0)^{-1}\|=\frac{K}{\delta}$. Moreover, if $R\geq \frac{\delta}{K}$, then $R\geq \frac{1}{L}$. Therefore according to \cref{lemma: banach quantitative inverse function theorem}, there exists inverse $g=f^{-1}$ in open ball of radius $\frac{1}{2L\|f'(x_0)^{-1}\|}=\frac{\delta^2}{2K}$ with image inside the ball of radius $\frac{\delta}{K}$.
\end{proof}

Define the detectable local error rates as $\vec{p}\,^*=(P_{\gamma_e}(e)|e\in\mathcal{E}_{\Gamma}^*)$, where $\mathcal{E}_{\Gamma}^*=\mathcal{E}_{\Gamma}\backslash C_0$ is the set of detectable local errors. Define a function $H: \mathbb{R}^{|\mathcal{E}^*_{\Gamma}|}\to \mathbb{R}^{|\mathcal{C}^{*}|}$ being the map to the values of the learnable eigenvalues in \cref{theorem: syndrome class learning}, as a function of the local error rates of the detectable errors
\begin{equation}
H(\vec{p}\,^*)=\log \vec{\nu}(\vec{p}\,^*).
\end{equation}
We denote $H_i$ as the component of the function $H$ for the syndrome class $C_i\in\mathcal{C}^{*}$.
%For $0<R<\frac{1}{4^{|\gamma|}}$, We define $L_H=\sup_{\vec{p}\in\mathcal{B}_{R}(\vec{0})}\|J_H(\vec{p})\|_{\infty}$ and $K_H=\sup_{\vec{p}\in\mathcal{B}_{R}(\vec{0})}(\max_{C\in\mathcal{C}^{*}}\sum_{i,j}|\frac{\partial H}{\partial p_i \partial p_j}|)$ as a upperbound of the first and second derivative of $H$ respectively.
To use \cref{lemma: quantitative inverse function theorem}, we define another function $G_B(x): \mathbb{R}^{|\mathcal{C}^{*}|} \to \mathbb{R}^{|\mathcal{C}^{*}|}$ and 
\begin{equation}
    G_B(\vec{x})=H(K_B\vec{x})=\log \vec{\nu}.
\end{equation}
Here $\vec{x}$ refers to the nontrivial syndrome error class error rates and $x_i=P_{C_i}$ denotes the syndrome class error rate for the $i$ th syndrome class $C_i$.
$K_B$ is a matrix such that $\vec{p}\,^*=K_B\vec{x}$. To be more specific, note that the function $G_B$ depends on the extra constraints $\vec{B}\vec{p}=\vec{0}$ and it gives rise to a linear map from syndrome error rate to the detectable error rates $\vec{p}\,^*=K_B\vec{x}$. Here, each row $i$ of the matrix $K_B$ has a single element at column $j$, corresponding to the fraction of the error rate $p^*_i$ in its syndrome class $P_{C_j}$:
\begin{equation}
    [K_B]_{i,j}=\frac{r_{e_i}}{\sum_{e\in C_j}r_e}\mathbb{I}_{C_j}(e_i),\label{equation: element of K_B}
\end{equation}
where $r_e$ are the relative strengths of the errors $e$ that construct matrix $B$ in \cref{equation: extra_nontrivial,equation: extra_trivial}, and $\mathbb{I}_C$ is the indicator function for the syndrome class, i.e., $\mathbb{I}_C(e)=1$ if error $e$ is in the syndrome class $C$, and $0$ otherwise.

We define the second-order derivative, i.e., Hessian, of a function $f(\vec{x}): \mathbb{R}^{n}\to \mathbb{R}^{n}$ as $J^2_{f}(\vec{x})$, which is a rank-3 tensor with entries:
\begin{equation}
    [J^2_{f}(\vec{x})]_{i,j,k}=\frac{\partial f_{i}}{\partial x_j\partial x_k}.
\end{equation}
We define the induced norm $\|J^2_{f}\|_{\infty}$ as
\begin{equation}
    \|J^2_{f}\|_{\infty}=\max_{\|\vec{v}\|_{\infty}=1}\|J^2_{f} \vec{v}\|_{\infty},
\end{equation}
where the norm on the RHS is the usual vector-induced norm for matrices, and $J^2_{f} \vec{v}$ is a matrix obtained by the tensor contraction $(J^2_{f} \vec{v})_{i,j}=\sum_k [J^2_{f}]_{i,j,k} v_k$.
The norm can be written explicitly as
\begin{equation}
    \|J^2_{f}\|_{\infty}=\max_{i}\sum_{j,k}\left|\frac{\partial f_i}{\partial x_j \partial x_k}\right|.
\end{equation}
\begin{applemma}
    Consider a local sparse $(r_{\Gamma},c_{\Gamma})-$Pauli channel on an $n$-qubit qLDPC code under \cref{assumption: open set probability}. Then $\|J_H\|_{\infty}=\mathcal{O}(1)$,$\|J_H^2\|_{\infty}=\mathcal{O}(1)$, and  $\|J^2_{G_B}\|_{\infty}=\mathcal{O}(1)$ for any extra constraints $B\vec{p}=0$.\label{lemma: bounded first and second derivative}
\end{applemma}
\begin{proof}
    For a local sparse $(r_{\Gamma},c_{\Gamma})-$Pauli channel on an LDPC code, the size of the syndrome class $C\in\mathcal{C}^*$ is bounded by a constant $|C|_{\mathrm{max}}$. For each syndrome class $C$, the collection of all errors in the involved channels $\Gamma_C=\{\gamma|\mathcal{E}_{\gamma}\cap C\neq\emptyset\}$, denoted by $\mathcal{E}_{\Gamma_C}$, is also of constant size. Due to \cref{assumption: open set probability}, we assume the total error rate of each channel $\sum_{e\in\mathcal{E}_{\gamma}}p_{e}< \frac{1}{2}$, we have each local Pauli eigenvalue lower bounded by a positive constant $\lambda_{\mathrm{min}}$: $\lambda_{e}=1-2\sum_{e'\in\mathcal{E}_{\gamma}|[[e,e']]=-1}p_{e'}\geq\lambda_{\mathrm{min}}>0$ for all local channel $\gamma\in\Gamma$ and $e\in\mathcal{E}_{\gamma}$.
    
    For a syndrome class $C_i\in \mathcal{C}$, we first look at each term $\log \mu_{e}$ of the component of the function $H_i=\log \nu_{C_i}=\sum_{e\in C_i}\log\mu_{e}$. Note that for any error $e\in C_i$ with its individual channel $\mathcal{N}_{\gamma_e}$,
    \begin{equation}
        \log \mu_{e}=\frac{2}{4^{|\gamma_e|}} \sum_{e'\in\mathcal{E}_{\gamma_e}}[[e',e]]\log\lambda_{e'}
    \end{equation}
    and
    \begin{equation}
        \left|\frac{\partial \log \mu_{e}}{\partial p^*_i}\right|\leq\sum_{e'|e'\text{ anticommute }e_i}\frac{4}{4^{|\gamma_e|}}\lambda^{-1}_{e'}\leq \frac{4}{\lambda_{\mathrm{min}}}.
    \end{equation}
    Therefore for each component $H_i$, we have
    \begin{equation}
        \left|\frac{\partial H_i}{\partial p^*_j}\right|=\sum_{e\in C_i}\left|\frac{\partial \log\mu_{e}}{\partial p^*_j}\right|\leq\frac{4|C|_{\mathrm{max}}}{\lambda_{\mathrm{min}}}
    \end{equation}
    So we have an upper bound for the Jacobian
    \begin{equation}
        \|J_H\|_{\infty}\leq\sum_{e\in\mathcal{E}_{\Gamma_C}}\frac{4|C|_{\mathrm{max}}}{\lambda_{\mathrm{min}}}=\mathcal{O}(1).
    \end{equation}
    For the second-order derivative
    \begin{align}
        \left|\frac{\partial \log \mu_{e}}{\partial p^*_i\partial p^*_j}\right|
        \leq\sum_{e'|e'\text{ anticommute }e_i,e_j}\frac{8}{4^{|\gamma_e|}}\lambda^{-2}_{e'}\leq \frac{8}{\lambda^2_{\mathrm{min}}}.
    \end{align}
    Similarly,
    \begin{equation}
    \|J^2_H\|_{\infty}\leq\sum_{e,e'\in\mathcal{E}_{\Gamma_C}}\frac{8|C|_{\mathrm{max}}}{\lambda^2_{\mathrm{min}}}=\mathcal{O}(1).
    \end{equation}
    To upper bound $\|J^2_{G_B}\|_{\infty}$, we use the chain rule:
    \begin{align}
        \|J_{G_B}(\vec{x})-J_{G_B}(\vec{y})\|_{\infty}&=\|(J_H(K_B\vec{x})-J_H(K_B\vec{y}))K_B\|_{\infty}\\
        &\leq \|J^2_H\|_{\infty}\|K_B\|_{\infty}\|\vec{x}-\vec{y}\|_{\infty}\|K_B\|_{\infty}\\
        &=\|J^2_H\|_{\infty}\|K_B\|^2_{\infty}\|\vec{x}-\vec{y}\|_{\infty}\\
        &\leq \|J^2_H\|_{\infty}\|\vec{x}-\vec{y}\|_{\infty}
    \end{align}
where we used the fact that $\|K_B\|_{\infty}\leq 1$, since each row of $K_B$ has a single non-zero element less than 1 (\cref{equation: element of K_B}).
Therefore, we have $\|J^2_{G_B}\|_{\infty}\leq \|J^2_H\|_{\infty}=\mathcal{O}(1)$.

\end{proof}

\begin{applemma}[Existence of physical error under $B$ constraints]\label{lemma: existence of solution}
     Consider a local sparse Pauli channel $\mathcal{N}_{\Gamma}=\circ_{\gamma\in\Gamma}\mathcal{N}_{\gamma}$, which acts on a qLDPC stabilizer/subsystem code with measured stabilizer subgroup $\mathcal{M}$. There exists a constant $0<p_{\mathrm{max}}\leq\frac{1}{4^{r_{\Gamma}}}$ such that if the true error rates $\vec{p}$ in $\mathcal{N}_{\Gamma}$ satisfy $\|\vec{p}^*\|_{\infty}< p_{\mathrm{max}}$, where $\vec{p}\,^*$ denotes the detectable error rates, then there exists a real solution for $\vec{p}'$ of the system of equations that consists of \cref{equation: D' matrix main}:
    \begin{equation}
        \log\vec{\Lambda}^{(\mathcal{M})}=D'^{(\mathcal{M})}\log \vec{\nu},
    \end{equation}
    and any set of extra constraints \cref{equation: extra constraints}:
    \begin{equation}
    B\vec{p}\,'=\vec{0}.
    \end{equation}
    Furthermore, if the positivity condition
    \begin{equation}
        P_C\geq \frac{1}{4}(K_H+ L^2_H K_{G_B})\|\vec p\,^*\|_{\infty}^2,\quad \forall C\in\mathcal{C}^*,
    \end{equation}
    for $O(1)$ constants $L_H, K_H, K_{G_B}$ that bound the norm of the first and second order derivatives of $H$ and $G_B$, is satisfied for the true error rates, then the solution $\vec{p}\,'^* \geq 0$. 
    
\end{applemma}
\begin{proof}
    According to \cref{theorem: syndrome class learning}, we have the Taylor expansion $G_B(\vec{x})=-2\vec{x}+\mathcal{O}(\|\vec{x}\|^2)$. Therefore at $\vec{x}_0=\vec{0}$, the Jacobian $J_{G_B}(\vec{x}_0)=-2I$ is invertible and
    \begin{equation}
        \delta_{G_B}=\frac{1}{\|J_{G_B}(\vec{x}_0)^{-1}\|_{\infty}}=2.
    \end{equation}
    We consider a radius $\|\vec{x}\|_{\infty}\leq P^C_{\mathrm{max}}$ for the function $G_B(\vec{x})$ and define an upper bound $K_{G_B}$ of the norm of the second-order derivative of $G_B$ as
    \begin{equation}\label{equation: define K_{G_B}}
        K_{G_B}=\max\left(\max_{\vec{x}\in\mathcal{B}_{P^C_{\mathrm{max}}}(\vec{0})} \|J^2_{G_B}(\vec{x})\|_{\infty},\quad\frac{\delta_{G_B}}{P^C_{\mathrm{max}}}\right).
    \end{equation}
    Let $R_p=\frac{1}{4^{r_{\Gamma}}}$, we define an upper bound $L_H$ of the norm of the Jacobian of $H$ as:
    \begin{equation}
        L_H=\max\left(\max_{\vec{x}\in\mathcal{B}_{R_p}(\vec{0})} \|J_H(\vec{x})\|_{\infty},\quad\frac{\delta^2_{G_B}}{2K_{G_B}R_p}\right).\label{equation: define L_H}
    \end{equation}
    Both $K_{G_B}$ and $L_H$ are $\mathcal{O}(1)$ according to \cref{lemma: bounded first and second derivative}. Therefore, we have,
    \begin{equation}
        \|J_{G_B}(\vec{x})-J_{G_B}(\vec{0})\|_{\infty}\leq K_{G_B}\|\vec{x}\|_{\infty}, \qquad\forall \vec{x} \in \mathcal{B}_{P^C_{\mathrm{max}}}(\vec{0})
    \end{equation}
    and $P^C_{\mathrm{max}}\geq \frac{\delta_{G_B}}{K_{G_B}}$ according to \cref{equation: define K_{G_B}}. As a result, $G_B(\vec{x})=\vec{y}$ has solution when $\vec{y}\in \mathcal{B}_{r_y}(\vec{0})$ where the radius $r_y=\frac{\delta^2_{G_B}}{2K_{G_B}}=\frac{2}{K_{G_B}}$, and the solution $\vec x$ lies in $\mathcal{B}_{\frac{\delta_{G_B}}{K_{G_B}}}(\vec 0)$, according to \cref{lemma: quantitative inverse function theorem}.
    To find the radius for $\vec{p}\,^*$ around $\vec{0}$ such that $\vec{y}=H(\vec{p}\,^*)\in \mathcal{B}_{r_y}(\vec{0})$, note that we have
    \begin{equation}
        \|J_H(\vec{p}\,^*)\|_{\infty}\leq L_H, \qquad \forall \vec{p}\,^*\in\mathcal{B}_{R_p}(\vec{0}).
    \end{equation}
    Let $p_{\mathrm{max}}=\frac{r_y}{L_H}=\frac{2}{K_{G_B}L_H}$ then $p_{\mathrm{max}}\leq R_p=\frac{1}{4^{r_{\Gamma}}}$ according to \cref{equation: define L_H}.
    For $\|\vec{y}\|_{\infty}$ we have
    \begin{equation}\label{equation: bound y}
        \|\vec{y}\|_{\infty}=\|H(\vec{p}\,^*)\|_{\infty}\leq L_H\|\vec{p}\,^*\|_{\infty}
    \end{equation}
    and it will be within radius $r_y$ around $\vec{0}$ if $\|\vec{p}\,^*\|_{\infty}<p_{\mathrm{max}}$, thus a solution exists.

To prove the positivity condition, note that
\begin{align}
    \vec y&=H(\vec p\,^*)=-2\vec P_C + R_H=-2\vec P_C + \mathcal{O}(\|\vec p^*\|_{\infty}^2) \\
    \vec y&=G_B(\vec x)= -2 \vec x + R_{G_B} = -2\vec x + \mathcal{O}(\|\vec x\|_{\infty}^2),\label{equation: y expansion}
\end{align}
and $\|J_H^2\|_{\infty}=\mathcal{O}(1)$, and  $\|J^2_{G_B}\|_{\infty}=\mathcal{O}(1)$ according to \cref{lemma: bounded first and second derivative}. Let $K_{G_B}$ and $K_H$ be their upper bound respectively, therefore
\begin{align}
    \|\vec y + 2\vec P_C\|_{\infty} &= \|R_H\| \leq \frac{1}{2} K_{H}\|\vec p\,^*\|_{\infty}^2\label{equation: define K_H}\\
    \|\vec y + 2\vec x\|_{\infty} &=\|R_{G_B}\| \leq \frac{1}{2} K_{G_B}\|\vec x\|_{\infty}^2.
\end{align}
From the inverse function theorem, conditioned on real solution exists, $\|\vec x\|_{\infty}\leq \delta_{G_B}/K_{G_B}=\frac{2}{K_{G_B}}$, therefore $\|\vec x\|_{\infty}\geq \frac{1}{2} K_{G_B}\|\vec x\|_{\infty}^2$, and from \cref{equation: bound y,equation: y expansion}, and the reverse triangle inequality, we have
\begin{equation}
    L_H\|\vec{p}\,^*\|_{\infty}\geq \|y\|_{\infty}=\| -2 \vec x + R_{G_B}\|_{\infty}\geq  2 \|\vec x\|_{\infty} -
    \|R_{G_B}\|\geq 2 \|\vec x\|_{\infty} -
    \|\vec x\|_{\infty}=\|x\|_{\infty}
\end{equation}
and as a result,
\begin{align}
    \|\vec P_C - \vec x\|_{\infty}\leq \frac{1}{4}(K_H+ L^2_H K_{G_B})\|\vec p\,^*\|_{\infty}^2.
\end{align}
Therefore, if $P_C \geq \frac{1}{4}(K_H+ L^2_H K_{G_B})\|\vec p\,^*\|_{\infty}^2$ for all $C\in\mathcal{C}^*$, then we have $\vec{x}\geq 0$, thus $K_B\vec x=\vec p\,'^*\geq 0$.

\end{proof}

\section{\label{appendix: spacetime code}Spacetime code}
We can map a non-adaptive Clifford circuit into a spacetime code \cite{bacon2017sparse, gottesman2022opportunities, delfosse2023spacetime}, which we introduce briefly in this section. We also note that the spacetime code is highly related to the detector error model that is being widely used in the field of fault tolerance \cite{fowler2014scalable,derks2024designing}. In particular, the measurement result of the generators of the measured stabilizer group in the spacetime code corresponds to the detectors in the detector error model. Here, we used the spacetime code formalism since the learning protocol in the static code case can be carried out here directly. For simplicity, we omit the superscript $^{(\text{ST})}$ here for operators in the spacetime code.

Given a circuit with $n$ qubits and $T$ time steps, the spacetime code constructed from the Clifford circuit is a subsystem stabilizer code of $n(T+1)$ qubits, where we put layers of $n$ qubits in between every pair of neighboring circuit layers, as well as a layer before and a layer after the circuit. We assign a temporal label $t=0.5,1.5,..., T+0.5$ for each layer of spacetime qubits. We borrow the notation (change from $\eta$ to $\kappa$) in \cite{bacon2017sparse,delfosse2023spacetime} and denote $\kappa_{t-0.5}(O)\in\mathcal{P}_{n(T+1)}$ as the Pauli operator in the spacetime code by placing Pauli $P\in\mathcal{P}_n$ at $t$th layer of the spacetime code tensor product with $I_n\in\mathcal{P}_n$ on all other layers. For any spacetime operator $O\in\mathcal{P}_{n(T+1)}$, $O_{t-0.5}\in\mathcal{P}_n$ denotes the component of $O$ at the $t$th layer of the spacetime code. Moreover, let $t\in [T], t\leq t'\leq T+1$,, and we denote the unitary gate at layer $t$ as $U_{t}$, the product of gates from layer $t$ to $t'-1$ as $U_{t,t'}$, and let $U_{t,t'}=I$ if $t'\leq t$.

The goal of the construction is that circuit-level Pauli errors can be considered as errors acting on these spacetime qubits such that both the syndrome and logical effect of the error can be captured in the static spacetime code. As a subsystem code, the structure of the conventional spacetime code is typically defined via the gauge group $\mathcal{G}$ of the spacetime code \cite{bacon2017sparse,gottesman2022opportunities}. It is defined as the set of Pauli errors in the spacetime code for which the corresponding errors in the circuit do not trigger any detection events and are logically trivial. In the conventional case, when one wants to preserve all logical operators of the input state stabilizer code, or simply the input code, $\mathcal{G}$ can be constructed by adding the following Pauli operators as generators:
\begin{enumerate}
    \item At the input layer, include the input code stabilizer generators as $\kappa_{0.5}(S^{(\text{ini})}_i)$ in $\mathcal{G}^{\text{(ST)}}$.
    \item For each measurement $M^{(\text{circ})}$ in the circuit at layer $t=t(M^{(\text{circ})})$, include $\kappa_{t+0.5}(M^{(\text{circ})})$.
    \item For each layer $t=1,...,T$, add $\kappa_{t-0.5}(P_i)\kappa_{t+0.5}( U_{t,t+1}P_iU^{\dagger}_{t,t+1})$, where $P_i$ runs over a generating set of $N(\mathcal{M}^{(\text{circ})}_t)$, i.e., the normalizer of the group of the measured operator $\mathcal{M}_t$ at layer $t$. For example, if there is no measurement at layer $t$, then we let $P_i$ be $X_1,Z_1,...,X_n,Z_n$, and if there is a $Y_2Z_3$ measurement at layer $l$, $P_i$ runs over $\{X_1,Z_1, Y_2,Z_3, X_2X_3, X_4,Z_4,...,Z_n\}$.
    %\Dom{Does that mean that you just ever need to add two-layer operators? I thought you needed to back-propagate all the way?}
\end{enumerate}

As in \cref{section: logical error rate and fidelity} of the main text, we are only interested in the logical errors affecting the circuit's logical measurement. This motivates us to extend the gauge group introduced above. We will first construct the measured stabilizer group $\mathcal{M}$ and the measured logical group $\mathcal{L}'$ (\cref{equation: measured spacetime logical subgroup}) using the output of \cref{algorithm: spacetime stabilizer}. Then concretely, the extended gauge group can be obtained by replacing operators in step 1 with operators of the form $\kappa_{0.5}(O)$ where $O$ goes over a generating set of $N(\mathcal{L}'^{(\text{ini})})$. Here we define $\mathcal{L}'^{(\text{ini})}=\{L'_{0.5}\mid L'\in\mathcal{L}'\}$ as the subgroup of the logical group of the input code associated with the logical measurement of the circuit.

We now introduce \cref{algorithm: spacetime stabilizer} below. The algorithm keeps track of the stabilizer group, destabilizer group, and the logical group at each layer and records the relevant measurement indices whenever a stabilizer or logical operator is measured. These records, in the form of two matrices $\mathcal{O}^{\perp}$ and $\mathcal{K}^L$ are then used to construct the measured stabilizer and logical group of the spacetime code. The stabilizer group $\mathcal{S}$ of the spacetime code is the center of the conventional gauge group $\mathcal{G}$, but in general, not all elements in $\mathcal{S}$ are measured, and the measured stabilizer group $\mathcal{M}$ is a subgroup of $\mathcal{S}$. Similarly, for a Clifford circuit, the logical measurements only correspond to a subset of the initial logical operators, and therefore, the measured logical operators in the spacetime code are only a subset of the normalizer of the conventional gauge group. We use the following notation in the algorithm below: For a binary vector $\vec a\in \mathbb{Z}_2^m$, we use the notation $S^{\vec{a}}=[S^{a_1}, S^{a_2}, ... ,S^{a_m}]$, and $\mathcal{S}\gets\mathcal{S}\odot S'^{\vec{a}}$ represents the process of replacing the list of operators $\mathcal{S}=[S_1,...,S_m]$ by $\mathcal{S}\odot S'^{\vec{a}}=[S_1S'^{a_1},S_2S'^{a_2},...,S_mS'^{a_m}]$.

\begin{algorithm}[H]
    \caption{Measured stabilizer generators and logical generators for the spacetime code\label{algorithm: spacetime stabilizer}}
    \KwIn{
    \begin{enumerate}
        \item The Clifford circuit with $m$ measurements $\{M^{(\text{circ})}_1,...,M^{(\text{circ})}_m\}$ with input state encoded in a [[n,k]] stabilizer code.
        \item Initial stabilizer: $\mathcal{S}=[S^{(\text{ini})}_1,...,S^{(\text{ini})}_s]$ and initial destabilizer: $\mathcal{D}=[D^{(\text{ini})}_1,...,D^{(\text{ini})}_s]$.
        \item Initial $X$ logical operators: $\overline{\mathcal{X}}=[\overline{X}^{(\text{ini})}_1,...,\overline{X}^{(\text{ini})}_k]$ and initial $Z$ logical operators: $\overline{\mathcal{Z}}=[\overline{Z}^{(\text{ini})}_1,...,\overline{Z}^{(\text{ini})}_k]$.
    \end{enumerate}
    }
    \KwOut{
    \begin{enumerate}
        \item $\mathcal{O}^{\perp}$, parity checks of the output bitstrings.
        \item $\mathcal{K}^L$, sets of operator indices of measured logical operators.
    \end{enumerate}
    }
    Initialize $\vec{c}_s,\vec{c}_d\in\mathbb{Z}_2^s$,  $\vec{c}_{\overline{x}},\vec{c}_{\overline{z}}\in\mathbb{Z}_2^k$. Let $\{\vec e_i\}^{s+m}_i$ be the set of unit vectors in $\mathbb Z_2^{s+m}$.\\
    Initialize two sets $\mathcal{O}^{\perp}=\{\}$ and $\mathcal{K}^L=\{\}$.\\
    \For{$i= 1,..,s$}{Initialize $\vec v_s^{(i)}=\vec e_i \in \mathbb{Z}_2^{s+m}$.}
    \For{$i= 1,..,k$}{Initialize $\vec v_{\overline{x}}^{(i)}=\vec v_{\overline{z}}^{(i)}=\vec 0 \in \mathbb{Z}_2^{s+m}$.}

    \For{$r=1,...,m$}
    {
        Propagate $M^{(\text{circ})}_r$ to the beginning of the circuit as $M_*$.\\
        Let $\vec{c}_s=(\langle M_*, \mathcal{D}[i]\rangle)^s_{i=1}$, $\vec{c}_d=(\langle M_*, \mathcal{S}[i]\rangle)^s_{i=1}$, $\vec{c}_{\overline{x}}=(\langle M_*, \overline{\mathcal{Z}}[i]\rangle)^k_{i=1}$, $\vec{c}_{\overline{z}}=(\langle M_*, \overline{\mathcal{X}}[i]\rangle)^k_{i=1}$.\\
        
        \If{$\vec{c}_d\neq \vec{0}$ (destabilizer measurement)}
        {
            Let $i^*:=\min\{i\in [s]:(c_d)_i = 1\}$ and $D_*=S_{i^*}$.\\
            Update tableau $\mathcal{S}\gets\mathcal{S}\odot D_*^{\vec{c}_d}$, $\mathcal{D}\gets\mathcal{D}\odot D_*^{\vec{c}_s}$, $\overline{\mathcal{X}}\gets\overline{\mathcal{X}}\odot D_*^{\vec{c}_{\overline{z}}}$, $\overline{\mathcal{Z}}\gets\overline{\mathcal{Z}}\odot D_*^{\vec{c}_{\overline{x}}}$.\\
            For all $i=1,...,s$, $j=1,...,k$, update operation vectors $\vec v^{(i)}_s\mathrel{\texttt{+=}}(c_d)_i\cdot \vec v^{(i^*)}_s$, $\vec v^{(j)}_{\overline{x}}\mathrel{\texttt{+=}}(c_{\overline{z}})_j\cdot \vec v^{(i^*)}_s$, $\vec v^{(j)}_{\overline{z}}\mathrel{\texttt{+=}}(c_{\overline{x}})_j\cdot \vec v^{(i^*)}_s$\\
            Delete the $i^*$th element of $\mathcal{S}$ and $\mathcal{D}$, append $M_*$, $D_*$ to $\mathcal{S}$ and $\mathcal{D}$ respectively.\\
            Update $\vec v_s^{(i)}\gets\vec v_s^{(i+1)}$ for $i=i^*,...s-1$. Let $\vec v_s^{(s)}=\vec e_{s+r}$.\\
        }
        \If{$\vec{c}_d=\vec{0}$ and $\vec{c}_{\overline{x}}\neq\vec{0}$ (or $\overline{X}\leftrightarrow \overline{Z}$ if $\vec{c}_{\overline{x}}=\vec{0}$ and $\vec{c}_{\overline{z}}\neq\vec{0}$) (logical measurement)} 
        {
            Let $i^*:=\min\{i\in [k]:(c_{\overline{x}})_i = 1\}$ and $D_*=\overline{Z}_{i^*}$.\\
            Update tableau $\mathcal{D}\gets\mathcal{D}\odot D_*^{\vec{c}_s}$, $\overline{\mathcal{X}}\gets\overline{\mathcal{X}}\odot D_*^{\vec{c}_{\overline{z}}}$, $\overline{\mathcal{Z}}\gets\overline{\mathcal{Z}}\odot D_*^{\vec{c}_{\overline{x}}}$.\\
            For all $j=1,...,k$, update operation vectors $\vec v^{(j)}_{\overline{x}}\mathrel{\texttt{+=}} (c_{\overline{z}})_j\cdot \vec v^{(i^*)}_{\overline{x}}$, $\vec v^{(j)}_{\overline{z}}\mathrel{\texttt{+=}}(c_{\overline{x}})_j\cdot \vec v^{(i^*)}_{\overline{x}}$\\
            Delete the $i^*$th element of $\overline{\mathcal{X}}$ and $\overline{\mathcal{Z}}$, append $M_*$, $D_*$ to $\mathcal{S}$ and $\mathcal{D}$ respectively.\\
            Append $\vec u^L=\vec e_{s+r} + \sum^s_{i} (c_s)_i \vec v^{(i)}_s + \sum^k_{i} [(c_{\overline{x}})_i \vec v^{(i)}_{\overline{x}} + (c_{\overline{z}})_i \vec v^{(i)}_{\overline{z}}]$ to $\mathcal{K}^{L}$.\\
            Update $\vec v_{\overline{x}}^{(i)}\gets\vec v_{\overline{x}}^{(i+1)},\vec v_{\overline{z}}^{(i)}\gets\vec v_{\overline{z}}^{(i+1)}$ for $i=i^*,...k-1$. Let $\vec v_s^{(s)}=\vec e_{s+r}$.\\
            $s\gets s+1$, $k\gets k-1$.
        }
        \If{$\vec{c}_d=\vec{c}_{\overline{x}}=\vec{c}_{\overline{z}}=\vec{0}$ (stabilizer measurement)}
        {
            Let $i^*:=\min\{i\in [s]:(c_s)_i = 1\}$ and $D_*=D_{i^*}$.\\
            Update tableau $\mathcal{D}\gets\mathcal{D}\odot D_*^{\vec{c}_s}$.\\
            Delete the $i^*$th element of $\mathcal{S}$ and $\mathcal{D}$, append $M'$, $D_*$ to $\mathcal{S}$ and $\mathcal{D}$ respectively.\\
            Append $\vec u=\vec e_{s+r} + \sum^s_{i} (c_s)_i \vec v^{(i)}_s$ to $\mathcal{O}^{\perp}$.\\
            Update $\vec v_s^{(i)}\gets\vec v_s^{(i+1)}$ for $i=i^*,...s-1$. Let $\vec v_s^{(s)}=\vec e_{s+r}$.
        }
        
    }
\Return{$\mathcal{O}^{\perp}$, $\mathcal{K}^L$}
\end{algorithm}
\cref{algorithm: spacetime stabilizer} is an extended version of Algorithm 1 in Ref.~\cite{delfosse2023spacetime}, as we also included the necessary information to construct the spacetime logical operators. Moreover, this also improves the classical computation cost from $O(n^3)$ to $O(n^2)$ by keeping track of destabilizers and calculating commutation instead of doing Gaussian elimination.

A key operation on spacetime Pauli operators is called (back-)cumulant \cite{delfosse2023spacetime} (spackle \cite{bacon2017sparse}), which is used to construct spacetime stabilizer and logical generators. The cumulant and back-cumulant of a spacetime operator $F$ denoted by $\overrightarrow{F}$ and $\overleftarrow{F}$ are defined by
\begin{align}
    (\overrightarrow{F})_{l-0.5} &= \prod_{l'=1}^{l'<l}U_{l',l}F_{l'-0.5}U^{\dagger}_{l',l},\\
    (\overleftarrow{F})_{l-0.5} &= \prod_{l'=T+1}^{l'>l}U^{\dagger}_{l,l'}F_{l'-0.5}U_{l,l'}.
\end{align}

Consider a Clifford circuit with initial stabilzer generators $\mathcal{S}=\{S^{(\text{ini})}_1,...,S^{(\text{ini})}_s\}$, destabilizer: $\mathcal{D}=\{D^{(\text{ini})}_1,...,D^{(\text{ini})}_s\}$, $X$ logical operators: $\overline{\mathcal{X}}=\{\overline{X}^{(\text{ini})}_1,...,\overline{X}^{(\text{ini})}_k\}$ and $Z$ logical operators: $\overline{\mathcal{Z}}=\{\overline{Z}^{(\text{ini})}_1,...,\overline{Z}^{(\text{ini})}_k\}$. We define a tuple of operators of the circuit with length $s+m$:
\begin{equation}\label{equation: B operators}
\vec{E}=
(
S^{(\text{ini})}_1,...,S^{(\text{ini})}_s, M^{(\text{circ})}_1,...,M^{(\text{circ})}_m)
\end{equation}
and define the $i$ th operator in $\vec E$ as $E_i$. We denote $t_j$ as the temporal index of operator $E_j$ in the circuit, and for initial stabilizer and logical operators, $t_j=0$.
The measured spacetime stabilizer generators $M\in\mathcal{M}$ are given by
\begin{equation}
        M(\vec u)=\overleftarrow{\prod_{j=s+1}^{s+m} \kappa_{t_j-0.5}(E_j^{u_j})}, \forall \vec u\in\mathcal{O}^{\perp}.\label{equation: spacetime stabilizer backpropagation}
\end{equation}
The measured spacetime logical generators $L'\in\mathcal{L}'$ are
\begin{equation}\label{equation: spacetime logical backpropagation}
L'(\vec u^L)=\overleftarrow{\prod_{j=s+1}^{s+m}\kappa_{t_j-0.5}(E^{u_j}_j)}, \forall \vec u^L\in \mathcal{K}^L.
\end{equation}
There can also be unmeasured logical operators in the spacetime code. Given any initial logical operator $L^{(\text{ini})}$ that has not been measured, we have the corresponding spacetime logical operator $L=\overrightarrow{\kappa_{0.5}(L^{(\text{ini})})}$.

The measured stabilizer group of the spacetime code is given by
\begin{equation}
    \mathcal{M}=\langle M(\vec u) \rangle_{\vec u\in\mathcal{O}^{\perp}}.
\end{equation}
The measured logical operators $\mathcal{L}'^{\text{(ST)}}$ is the group generated by all the measured stabilizer generators and the relevant spacetime logical generators $\{L'_1,...,L'_{m_L}\}$ in \cref{equation: spacetime logical backpropagation}:
\begin{equation}
    \mathcal{L}'=\langle M_1,...,L'_1,...,L'_{m_L}\rangle.\label{equation: measured spacetime logical subgroup}
\end{equation}

 Below is a summary of the structure of the spacetime code in our work. Given a set of logical operators measured in the circuit, we map them to a generating set of measured spacetime logical operators according to \cref{algorithm: spacetime stabilizer} and \cref{equation: spacetime logical backpropagation}. These operators generate the abelian subgroup of measured spacetime logical operators $\mathcal{L}'$ according to \cref{equation: measured spacetime logical subgroup}. The linear constraints ('detectors') $\vec u \in\mathcal{O}^{\perp}$ of the measurements in the ideal circuits are mapped to the measured spacetime stabilizer group $\mathcal{M}$. The group of dressed logical operators $\mathcal{L}'_{\text{dr}}$ contains all the undetectable operators. It is defined as the normalizer of the measured spacetime stabilizer $\mathcal{L}'_{\text{dr}}=N(\mathcal{M})$. The gauge group $\mathcal{G}'$ for $\mathcal{L}'$ corresponds to the set of errors in the circuit that are undetectable, but also logically trivial with respect to the logical measurements. It is defined as $\mathcal{G}'=N(\mathcal{L}')$, i.e., it consists of operators in the spacetime code that commute with the measured logical group $\mathcal{L}'$, which encompasses the measured stabilizer group $\mathcal{M}$. See diagram in \cref{fig: spacetime_code_commutation}.

We note the two distinct roles/perspectives of Pauli operators here. Pauli operators, on the one hand, are Hermitian operators that can encode the information about the state, e.g., stabilizers of a state or Pauli measurement of a state. On the other hand, it is a unitary operator that acts on the state as an operation/error. The (bare) logical group $\mathcal{L}'$ is the first case where it encodes the logical information being measured, whereas the dressed logical group $\mathcal{L}'_{\text{dr}}$, or more precisely the quotient group $\mathcal{L}'_{\text{dr}}/\mathcal{G}'$, can be interpreted as the logical operation/error on the logical states. Although the construction of measured spacetime logical operators $\mathcal{L}'$ is not necessary in the physical error rate learning algorithm, $\mathcal{L}'$ is needed in the proof of logical error learnability in \cref{theorem: logical error rate learnability} in \cref{appendix: circuit logical learning}.

\begin{figure*}[h]
\centering
    \includegraphics[width=0.95\textwidth]{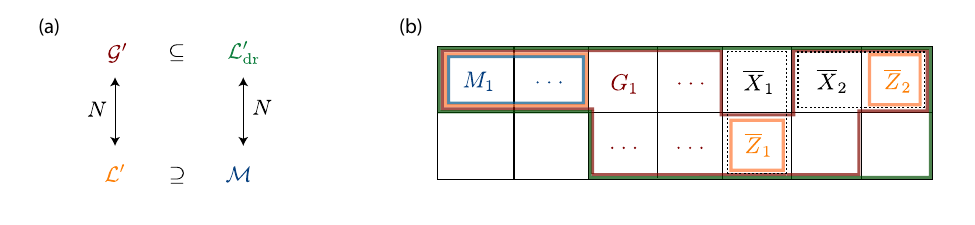}
    \caption{(a) Four subgroups of the spacetime subsystem codes and their relation with the diagram in \cite{wagner2023learning}. $\mathcal{M}$ is the measured spacetime stabilizer group, which is contained in the measured spacetime logical subgroup $\mathcal{L}'$. $\mathcal{L}'_{\text{dr}}=N(\mathcal{M})$ is the dressed logical subgroup, which contains the gauge group $\mathcal{G}'=N(\mathcal{L}')$. $N$ here denotes the normalizer in the Pauli group, and we have $A=N(B) \iff B=N(A)$ for any two Pauli subgroups $A, B$. (b) Example of the commutation between generators of different subgroups in (a) and the inclusion relation of these subgroups (colored boundaries). Each box represents a Pauli operator, and each column represents a pair of anticommuting operators, while operators in the same row commute. Here we consider an $n$-qubit spacetime code of a circuit with two logical qubits, and we measure $\overline{Z}_1$ at an odd time step while measuring $\overline{Z}_2$ at an even time step. The resulting spacetime logical Pauli $\overline{X}_2$ and $\overline{Z}_2$ for the second logical qubit commute. However, we can still define distance as $d^{(\mathcal{L}')}=\min_{e\in\mathcal{L}'_{\text{dr}}\backslash \mathcal{G}'} \text{weight}(e)$, and the inclusion relation in (a) still holds.}
    \label{fig: spacetime_code_commutation}
\end{figure*}

We prove below the equivalent condition of error detectability and logical errors on logical measurements in the circuit, in terms of the commutation relation between errors and measured stabilizer/logical operators in the spacetime code.

\begin{applemma}[\textcite{delfosse2023spacetime} Adjoint of the accumulator]\label{lemma: prop_adjoint}
For any pair of spacetime operators $A,B\in\mathcal{P}_{n(T+1)}$, we have
\begin{equation}
    \langle \overrightarrow{A},B\rangle=\langle A,\overleftarrow{B}\rangle.
\end{equation}
\end{applemma}

\begin{applemma}[\textcite{delfosse2023spacetime} Effect of faults]\label{lemma: effect of faults}
Given an Clifford circuit with $m$ measurements $\{M_1^{(\text{circ})},...,M_m^{(\text{circ})}\}$ and circuit-level Pauli error $e\in\mathcal{P}_{n(T+1)}$ in the spacetime code, define the output measurement bitstrings in the circuit as $\vec o^{(\text{circ})}\in\mathbb{Z}_2^m$ and its probability distribution with and without error $e$ as $\mathbb{P}^{(e)}(\vec o^{(\text{circ})})$ and $\mathbb{P}(\vec o^{(\text{circ})})$ respectively, we have
\begin{equation}
\mathbb{P}^{(e)}(\vec o^{(\text{circ})})=\mathbb{P}(\vec o^{(\text{circ})}+\vec f^{(\text{circ})}(e)),
\end{equation}
where $\vec f^{(\text{circ})}(e)=(f^{(\text{circ})}(e)_1,\ldots,f^{(\text{circ})}(e)_m)\in\mathbb{Z}_2^{m}$, $f^{(\text{circ})}(e)_j=\left\langle\overrightarrow{e}_{t(M^{(\text{circ})}_j)-0.5},\,M^{(\text{circ})}_j\right\rangle$, and we denote $t(M_j^{(\text{circ})})$ as the temporal label when $M_j^{(\text{circ})}$ occurs in the circuit.
\end{applemma}

Let the initial stabilizer generator values be $\vec{o}^{(\text{ini})}=(o^{(\text{ini})}_1,...,o^{(\text{ini})}_s)$ and let the measurement outcome of the $m$ measurements $\{M^{(\text{circ})}_1,...,M^{(\text{circ})}_m\}$ be $\vec{o}^{(\text{circ})}=(o^{(\text{circ})}_1,...,o^{(\text{circ})}_m)$. For a parity check $\vec u\in\mathcal{O}^{\perp}$, let $\vec u^{(\text{ini})}=(u_1,...,u_s)$, and $\vec u^{(\text{circ})}=(u_{s+1},...,u_{s+m})$. Similarly, for $\vec u^L\in\mathcal{K}^L$, let $\vec u^{L(\text{circ})}=(u^L_{s+1},...,u^L_{s+m})$

\begin{applemma}[Equivalence of syndrome in circuit and spacetime code]\label{lemma: Equivalence of syndrome in circuit and spacetime code}
    For any circuit-level Pauli error corresponding to $e\in\mathcal{P}_{n(T+1)}$ in the spacetime code in a Clifford circuit, its syndrome bitstring $\vec s(e)$ corresponding to the parity check $\mathcal{O}^{\perp}$ of the outcome measurements in the circuit is equal to the scalar commutator with the measured stabilizer generators in the spacetime code $(\langle e,M(\vec u) \rangle)_{\vec u \in \mathcal{O}^{\perp}}$.
\end{applemma}
\begin{proof}
    The syndrome bitstring, by construction, is deterministically $\vec{0}$ for an noiseless circuit, i.e., $\mathbb{P}(\vec o)\neq 0$ if and only if $\vec o\in\text{ker}\,\mathcal{O}^{\perp}$. With error $e$, according to \cref{lemma: effect of faults}, $\mathbb{P}^{(e)}(\vec o^{(\text{circ})})=\mathbb{P}(\vec o^{(\text{circ})}+\vec f^{(\text{circ})}(e))$. Equivalently, we have
    \begin{equation}
        \mathbb{P}^{(e)}(\vec o)=\mathbb{P}(\vec o+\vec f(e)),\label{equation: full outcome effect}
    \end{equation}
    where $\vec f(e)=$ is obtained by prepending $\vec f^{(\text{circ})}(e)$ with $\vec 0\in\mathbb{Z}^s_2$. Therefore, $\vec o$ must satisfy $\vec o+\vec f(e)\in\text{ker}\,\mathcal{O}^{\perp}$, which means $\vec s(e)=(\vec{u}\cdot \vec o)_{\vec u \in \mathcal{O}^{\perp}}=(\vec{u}\cdot \vec f(e))_{\vec u \in \mathcal{O}^{\perp}}$. For every $\vec u \in\mathcal{O}^{\perp}$, we also have
    \begin{align}
        \vec{u}\cdot \vec f(e)=&\vec u^{(\text{circ})}\cdot\vec f^{(\text{circ})}(e)\\
        =&\sum^m_{i=1}u^{(\text{circ})}_i\left\langle \overrightarrow{e}_{t(M^{(\text{circ})}_i)-0.5},M^{(\text{circ})}_i\right\rangle\\
        =&\sum^{s+m}_{i=s}\left\langle \overrightarrow{e}_{t_i-0.5},E^{u_i}_i\right\rangle\\
        =&\left\langle \overrightarrow{e},\prod^{s+m}_{i=s}\kappa_{t_i-0.5}(E^{u_i}_i)\right\rangle\\
        =&\left\langle e,\overleftarrow{\prod^{s+m}_{i=s}\kappa_{t_i-0.5}(E^{u_i}_i)}\right\rangle\\
        =&\left\langle e,M(\vec u)\right\rangle.\label{equation: equivalence to commutator}
    \end{align}
    Here, the second equality follows from \cref{lemma: effect of faults}, and the second-last equality follows from \cref{lemma: prop_adjoint}.
\end{proof}

\begin{applemma}[Equivalence of logical error in the spacetime code]\label{lemma: Equivalence of logical error in the spacetime code}
    Denote as $\vec o_L\in\mathbb{Z}_2^{m_L}$ the logical measurement results after syndrome measurement and correction in a Clifford circuit with some input codeword.
    Consider any undetectable circuit-level Pauli error corresponding to $e\in\mathcal{P}_{n(T+1)}$ in the spacetime code due to the circuit error and subsequent correction from some decoder. We have the logical outcome distribution
    \begin{equation}
        \mathbb{P}_L^{(e)}(\vec o_L)=\mathbb{P}_L(\vec o_L + \vec f_L(e)),
    \end{equation}
    where $\mathbb{P}_L$ is the ideal logical outcome distribution and we refer to $\vec f_L(e)\in\mathbb{Z}_2^{m_L}$ as the logical error on the logical measurements. Moreover, $\vec f_L(e)$ is equal to the binary commutation of $e$ with the measured logical operators $\vec f_L(e) = (\langle e,L'(\vec u^L)\rangle)_{\vec u^L\in\mathcal{K}^L}$.
\end{applemma}
\begin{proof}
     Given the corrected circuit outcome $\vec o$, the logical measurement result is $\vec o_L=(\vec u^L \cdot \vec o)_{\vec u^L\in\mathcal{K}^L}$. If the circuit has $m_L$ independent logical measurements, then $\mathcal{K}^L$ has full row rank with $\rank \mathcal{K}^L=m_L$. Therefore, the logical output distribution is
    \begin{align}
        \mathbb{P}_L^{(e)}(\vec o_L)&=\sum_{\vec o'\in\mathrm{ker}{\mathcal{K}^{L}}}\mathbb{P}^{(e)}\left((\mathcal{K}^{L})^{+}\vec o_L+\vec o\,'\right)\\
        &=\sum_{\vec o'\in\mathrm{ker}{\mathcal{K}^{L}}}\mathbb{P}\left((\mathcal{K}^{L})^{+}\vec o_L+\vec o\,' + \vec f(e)\right)\quad(\text{from }\cref{equation: full outcome effect})\\
        &=\sum_{\vec o'\in\mathrm{ker}{\mathcal{K}^{L}}}\mathbb{P}\left((\mathcal{K}^{L})^{+}(\vec o_L+ \mathcal{K}^L\vec f(e))+\vec o\,'\right)\\
        &=\mathbb{P}_L(\vec o_L + \vec f_L(e))
    \end{align}
where $(\mathcal{K}^{L})^{+}$ is any right inverse of $\mathcal{K}^{L}$, i.e., $\mathcal{K}^{L}(\mathcal{K}^{L})^{+}=I_{m_L}$ and the logical error $\vec f_L(e)=\mathcal{K}^L \vec f(e)$. Similar to the derivation of \cref{equation: equivalence to commutator}, we have $\vec u^L \cdot \vec f(e)=\langle e, L'(\vec u^L)\rangle,\, \forall \vec u^L\in\mathcal{K}^L$.
\end{proof}

\begin{appcorollary}[Correctness condition for spacetime code]\label{corollary: Correctness condition for spacetime code}
    For any circuit-level Pauli error corresponding to $e\in\mathcal{P}_{n(T+1)}$ in the spacetime code, $e$ is undetectable from the circuit measurements if and only if $e$ commutes with all elements $M\in\mathcal{M}$ in the measured spacetime stabilizer group. An undetectable error $e$ (from decoding and correction) does not affect the distribution of logical measurement results for any input codeword, if and only if $e$ commutes with all elements in the measured spacetime logical group $\mathcal{L}'$.
\end{appcorollary}

\section{\label{appendix: circuit logical learning}Logical error rate learnability}

In this section, we prove \cref{theorem: logical error rate learnability} in the main text regarding the learnability of logical error rate for a group of logical measurements in the non-adaptive Clifford circuit from syndrome data. Using the transformed eigenvalue basis introduced in the \cref{appendix: log u basis} together with the spacetime code formalism introduced in \cref{appendix: spacetime code}, we give a simpler proof with both sufficient and necessary conditions of the logical error class probability learnability results in the context of fault-tolerant circuits with logical measurements, compared to the proof of Ref.~\textcite{wagner2023learning} for sufficient condition in the code capacity setting. We note that with a circuit decoder, the learnability of the logical class error rate implies the learnability of the logical error rate. We omit the superscript $^{(\text{ST})}$ for operators in the spacetime codes.

\noindent\textbf{\cref{theorem: logical error rate learnability} (Restated)}.     Given a non-adaptive Clifford circuit with a set of logical measurements $\mathcal{M}_L$ under a circuit-level Pauli noise model (satisfies \cref{assumption: open set probability} when mapped to noise model on spacetime code), the logical error class probabilities $\overline{P}$ for $\mathcal{M}_L$ can be learned from the syndrome expectation value of the circuit if and only if the circuit is fault-tolerant (\cref{definition: fault tolerant}) to the noise model for $\mathcal{M}_L$. As a consequence, given a decoder of the circuit, the associated logical error rates $P_{\mathrm{fail}}$ for the set of logical measurements $\mathcal{M}_L$ can be learned from the syndrome data if the circuit is fault-tolerant to the noise model for $\mathcal{M}_L$.
\begin{proof}
We map linear constraints of the measurement in the circuits to the measured spacetime stabilizer group $\mathcal{M}$ and map measured logical operators in the circuit to the group of spacetime logical operators $\mathcal{L}'$. The gauge group $\mathcal{G}'$ is given by $\mathcal{G}'=N(\mathcal{L}')$. We denote as $P_L(e)$ the logical error class probability for the error $e\in\mathcal{P}_{n(T+1)}$ of the spacetime code for $\mathcal{L}'$ as
\begin{align}
    P_L(e)&=\overline{P}(\vec s(e),\vec \ell(e))\\
    &=\sum_{g\in\mathcal{G}'}P(eg)\label{equation: logical error class error appendix}\\
    &=\sum_{e'\in\mathcal{P}_n}P(e')\mathbb{I}_{\mathcal{G}'}(ee')\\
    &=P\ast\mathbb{I}_{\mathcal{G}'}(e)
\end{align}
where $\mathbb{I}_{\mathcal{H}}$ is the indicator function of the group $\mathcal{H}$, taking $\mathbb{I}_{\mathcal{H}}(e)=1$ if $e\in \mathcal{H}$ and $0$ otherwise. The Fourier transform of $P_L$ is therefore:
\begin{align}\label{equation: fourier transform of logical error class}
    \widehat{P_L}(O)=\widehat{P}(O)\,\widehat{\mathbb{I}}_{\mathcal{G}'}(O)=|\mathcal{G}'|\Lambda(O)\,\mathbb{I}_{\mathcal{L}'}(O)
\end{align}
where $\mathcal{L}'$ is the group of "bare" logical operators defined as the group of operators $\mathcal{L}=N(\mathcal{G})$ that commute with the gauge group. As a result, the Fourier transform of the logical error class probability $\widehat{P_L}(O)$ is only non-zero when $O\in\mathcal{L}'$, i.e., it only depends on the Pauli eigenvalue $\Lambda(L'), L\in\mathcal{L}'$ of the measured bare logical group.

Similar to \cref{equation: log equation mu}, we could write the logarithm of the Pauli eigenvalue of elements in the measured logical group $L\in\mathcal{L}'$ in terms of the transformed eigenvalues $\vec \mu$:
\begin{equation}
    \log\vec{\Lambda}^{(\mathcal{\mathcal{L}'})}=D^{(\mathcal{\mathcal{L}'})}\log \vec{\mu}.
    \label{equation: log equation logical}
\end{equation}
where $D^{(\mathcal{\mathcal{L}'})}[L,e]=\langle L,e\rangle$. Here, the Pauli eigenvalue of $\Lambda(M)$ for $M$ in the measured stabilizer code $\mathcal{M}$ of the spacetime code can be calculated from the syndrome expectation value of the circuit using the parity check $\mathcal{O}^{\perp}$ according to \cref{lemma: Equivalence of syndrome in circuit and spacetime code}.

Given the fault-tolerant condition (\cref{definition: fault tolerant}) is satisfied, to prove that the logical error rate is determined by the syndrome measurement, i.e., Pauli eigenvalues of operators in $\mathcal{M}$, one just needs to show all rows of $D^{(\mathcal{\mathcal{L}})}$ is linear dependent on rows of the matrix for the syndrome data $D^{(\mathcal{\mathcal{M}})}$ in \cref{equation: log equation}. Consider a circuit with circuit-level Pauli noise defined by $\mathcal{N}_{\Gamma}$ that is fault-tolerant to the set of logical operators in the circuit in the sense of \cref{definition: fault tolerant}. This fault-tolerant condition, according to \cref{corollary: Correctness condition for spacetime code}, is equivalent to following statement when mapped to the spacetime code: For any undetectable error $e\in\mathcal{E}_{\Gamma}$ in the spacetime code, $e$ commute with all logical operators in $\mathcal{L}'$, and for any two errors $e_1,e_2\in\mathcal{E}_{\Gamma}$ with the same syndrome, then $e_1e_2$ commute with all logical operators in $\mathcal{L}'$, i.e., $e_1$ and $e_2$ commute/anticommute with the same set of operators in $\mathcal{L}'$. Therefore, the syndrome class (\cref{definition: syndrome class}) $\mathcal{C}^{*(\mathcal{M})}$ and $\mathcal{C}^{*(\mathcal{L}')}$ with respect to the measured stabilizer group $\mathcal{M}$ and the measured logical group $\mathcal{L}'$ respectively, are the same. As a result, according to $\cref{lemma: rank D}$ we have
\begin{equation}
    \rank D^{(\mathcal{L})}=\rank D^{(\mathcal{M})}.
\end{equation}
Since $\mathcal{M}\subseteq\mathcal{L}'$, all rows of $D^{(\mathcal{L})}$ are linearly dependent on rows of $D^{(\mathcal{M})}$ and the Pauli eigenvalue of all logical operators $\Lambda(L), \forall L\in\mathcal{L}$ only depends on the measured syndrome expectations $\{\Lambda(M)| \forall M \in \mathcal{M}\}$. Therefore, using the syndrome data, we can recover the logical error classes probability \cref{equation: logical error class error appendix} and also the logical error rates according to \cref{equation: logical error rate} and the equivalence of logical error in the circuit and the spacetime code (\cref{lemma: Equivalence of logical error in the spacetime code}).

To prove the other direction, if (case 1) there exists an error $e$ $\in\mathcal{E}_{\Gamma}$, or (case 2) $e=e_1e_2$ for some $e_1,e_2\in\mathcal{E}_{\Gamma}$, such that $e$ is undetectable from the syndrome but anticommutes with some operator in $\mathcal{L}'$, then the rank of $\mathcal{D}^{(\mathcal{L}')}$ is strictly larger than the rank of $\mathcal{D}^{(\mathcal{M})}$. This is because the new set of nontrivial syndrome error classes (\cref{definition: syndrome class}) with respect to $\mathcal{L}'$ can be obtained by either adding some undetectable error $e$ as a new nontrivial class for case 1, or further partitioning the existing nontrivial syndrome classes (case 2). Therefore, according to \cref{lemma: rank D} we have
\begin{equation}
    \rank \mathcal{D}^{(\mathcal{L}')} > \rank \mathcal{D}^{(\mathcal{M})}.
\end{equation}
As a result, for an existing transformed eigenvalues $\vec{\mu}$ satisfying the syndrome constraints and resulting in certain Pauli eigenvalues $\vec{\Lambda}^{(\mathcal{L}')}$, there exists a vector $\vec{\epsilon}\in \mathrm{ker}(\mathcal{D}^{(\mathcal{M})})$ such that $\vec{\mu}'=\vec{\mu}+\vec{\epsilon}$ is still valid given the \cref{assumption: open set probability} (see \cref{lemma: valid mu is an openset}) and $\vec{\mu}'$ produces a different set of Pauli eigenvalues $\vec{\Lambda}^{(\mathcal{L}')}$. Since the logical error class probabilities are related to $\vec{\Lambda}^{(\mathcal{L}')}$ by an invertible Walsh-Hadamard transformation as in \cref{equation: fourier transform of logical error class}, $\vec{\mu}'$ and $\vec{\mu}'$ lead to different sets of logical error class probabilities. Therefore, the logical error class probabilities are not learnable.

%we map $\mathcal{L_{\text{circ}}}'$ to the group of spacetime logical operators $\mathcal{L}'^{\text{(ST)}}$ (see algorithm \ref{algorithm: spacetime stabilizer}) which we abbreviate as $\mathcal{L}'$ here and construct the corresponding spacetime gauge group $\mathcal{G}'=N(\mathcal{L}')$, then the corresponding spacetime code of the circuit satisfy the condition that for all pairs of supports $\forall\gamma_i,\gamma_j\in\Gamma$ of the local Pauli channels on the spacetime code, $\gamma_i\cup\gamma_j$ do not contain any "dressed" spacetime logical operator $L \in N(\mathcal{M})\backslash\mathcal{M}$. Otherwise, we would have order $p$ probability of having a logical error. See also Theorem 6.1. in \cite{gottesman2022opportunities}.
\end{proof}

Note that in the learnable case, the log of syndrome expectation values and the learnable transformed eigenvalues $\log \vec{\nu}$ are related by full column rank matrix $D'^{(\mathcal{M})}$ (\cref{lemma: rank D}), therefore the logical error rates are also determined by $\log \vec{\nu}$.

\section{\label{appendix: sample complexity} Sample complexity}

\subsection{Sample complexity for learning physical error rates\label{subsection: physical sample complexity}}

Here, we study the sample complexity of learning syndrome class error rate, where we first consider solving for the learnable transformed eigenvalues $\log \vec{\nu}$ in \cref{equation: D matrix reduced} and then estimate the syndrome class error rate using the approximation $P_C \approx-\frac{\log \nu_C}{2}$, as explained in \cref{appendix: physical proof}. Note that $P_C \approx-\frac{\log \nu_C}{2}\approx \frac{1-\nu_C}{2}$ when error rates are low, and $\frac{1-\nu_C}{2}$ in practice is a slightly more accurate estimator. Moreover, $\frac{1-\nu_C}{2}$ equals the probability of an odd number of errors occurring in $C$ for the detector error model, as shown in \cref{equation: detector error rates learnability}. However, for simplicity in error propagation, we will argue based on the approximation $-\frac{\log \nu_C}{2}$.

We first introduce some notation. Given a set of measured stabilizer generators $\{M_i\}^m_i$ for the measured stabilizer group $\mathcal{M}=\langle M_1,\dots, M_m\rangle$, we denote the indices of the stabilizer generators violated for each syndrome class $C\in\mathcal{C}$ as $J(C)\subseteq [m]$. We define the variables to solve $X_{C}=\log\nu_{C}$ as the learnable transformed eigenvalues, and for any $R\subseteq [m]$, let $Y_R=\log\Lambda(\prod_{i\in R}M_i)$ be the logarithm of the Pauli eigenvalues of the stabilizer elements, i.e., syndrome expectation values.

\begin{applemma}\label{lemma: recursive solving solves the linear equation}
    \cref{equation: recursive solving} is the solution of the system of linear equations \cref{equation: D matrix reduced}, where $\mathcal{M}'$ is the subset of measured stabilizer group $\mathcal{M}'=\bigcup_{C\in\mathcal{C}}\{\prod_{i\in J'}M_i|J'\subseteq J(C)\}$.
\end{applemma}
\begin{proof}
We denote $D_{R,C}$ as the matrix element of the syndrome matrix $D$ in the row labeled by the stabilizer element $\prod_{i\in R}M_i$ and column labeled by the syndrome class $C$ and it is given by
\begin{equation}
    D_{R,C}=
    \begin{cases}
        1 & \text{, if }R\cap J(C)\text{ is odd,}\\
        0 & \text{, if }R\cap J(C)\text{ is even.}
    \end{cases}
\end{equation}
The system of linear equations \cref{equation: D' matrix main} can be rewritten as
\begin{equation}
    Y_R=\sum_{C\in\mathcal{C}}D_{R,C}X_C.\label{equation: D reduced relabel}
\end{equation}
where $R\subseteq [m]$.
Let us define a linear transformation of $Y_R, \forall R\subseteq [m]$:
\begin{equation}
    Z_R=\sum_{S\subseteq R}(-1)^{|S|-|R|}Y_S.
\end{equation}
Since the power set of the stabilizer generator indices $[m]$ is a partially ordered set by the inclusion relation, we could order the rows $R$ in the $D$ matrix according to this partial order, and this is essentially row operations on \cref{equation: D reduced relabel}. Plugging this into \cref{equation: D reduced relabel}:
\begin{align}
    Z_R
    &=\sum_{S\subseteq R}(-1)^{|S|}\sum_{C\in\mathcal{C}}D_{S,C}X_C\\
    &=\sum_{C\in\mathcal{C}}(\sum_{S\subseteq R}(-1)^{|S|}D_{S,C})X_C\\
    &=\sum_{C\in\mathcal{C}}\chi_{R,C} X_C\label{equation: Z in X}
\end{align}
The matrix $\chi$ is "upper triangular" by looking at its elements
\begin{align}
    \chi_{R,C}= \sum_{S\subseteq R} (-1)^{|S|} D_{S,C}
    = \sum_{\substack{S\subseteq R \\ \text{$|S\cap J(C)|$ is odd}}} (-1)^{|S|}
    = \begin{cases}
        -2^{|R|-1}, & \text{if } R\subseteq J(C), \\
        0, & \text{else.}
    \end{cases}
\end{align}
In particular, $Z_R=0$ if $R$ satisfy that for any $C\in\mathcal{C}$, $R\not \subseteq J(C)$. Therefore, only stabilizer elements in $\mathcal{M}'=\bigcup_{C\in\mathcal{C}}\{\prod_{i\in J'}M_i|J'\subseteq J(C)\}$ are used.

Combining \cref{equation: D reduced relabel} and \cref{equation: Z in X} we obtain:
\begin{equation}
    -Z_R=2^{|R|-1}\sum_{\substack{C\in\mathcal{C}\\J(C)\supseteq R}}X_C=\sum_{S\subseteq R}(-1)^{|S|-1}Y_S
\end{equation}
From the second equality, we obtain the solution
\begin{equation}
    X_C=-\sum_{\substack{C'\in\mathcal{C}\\ J(C')\supset J(C)}}X_{C'}+\frac{1}{2^{|J(C)|-1}}\sum_{S\subseteq J(C)}(-1)^{(|S|-1)}Y_S. \label{equation: solve recursive appendix}
\end{equation}
\end{proof}

We now show that the sample complexity is constant given a desired precision on the logarithm of the learnable transformed eigenvalues $X_{C}=\log \nu_C$ for a local sparse qLDPC code and a $(r_{\Gamma}, c_{\Gamma})$-Pauli channel $\mathcal{N}_{\Gamma}$.

\begin{applemma}\label{lemma: sample size log nu}
    Consider a local sparse $(r_s,c_s)$-Pauli channel on an $n$-qubit qLDPC code. Choose $\epsilon, \delta>0$ and any syndrome class $C\in\mathcal{C}^*$, then the sample size $N$ required to estimate the logarithm of the learnable transformed eigenvalues $\log \nu_C$ to additive error within $\epsilon$ with confidence $1-\delta$ using \cref{equation: recursive solving} is independent of $n$ and scales as
    \begin{equation}
        N=\mathcal{O}\left(\frac{1}{\epsilon^2}\log(\frac{1}{\delta})\right).
    \end{equation}
\end{applemma}
\begin{proof}
We use the notation in the proof of \cref{lemma: recursive solving solves the linear equation}. In the local error model setting, we have $|J(C)|\leq r_{\Gamma}c_s$, and $|\{C'\in\mathcal{C}|J(C')\supset J(C)\}|\leq r_{s}c_{\Gamma}$ for any syndrome class $C$ since each generator acts on at most $r_s$ qubits and each qubit is affected by at most $c_{\Gamma}$ channels. Let $\mathcal{M}'=\bigcup_{C\in\mathcal{C}}\{\prod_{i\in J'}M_i|J'\subseteq J(C)\}$ be the subset of stabilizer elements involved in the recursive solution \cref{equation: recursive solving}. For any measured stabilizer element $M$ in $\mathcal{M}'$. 
Given \cref{assumption: open set probability} of the local channels, the Pauli eigenvalues of any local channels are lower bounded by a constant $\lambda_{\mathrm{min}}$, i.e., $\Lambda_{\gamma}(e)\geq\lambda_{\mathrm{min}}>0, \forall \gamma\in\Gamma$.
\begin{equation}
    \Lambda(M)=\prod_{\substack{\gamma\in\Gamma\\\text{supp}(M)\subseteq\gamma}}\Lambda_{\gamma}(M_{\gamma})\geq \lambda_{\mathrm{min}}^{r_{\Gamma} s_{\Gamma} c_s c_{\Gamma}}=\Theta(1)
\end{equation}
Given the additive error $\epsilon_{\Lambda(M)}$ of syndrome expectation values $\Lambda(M)$ is less than $\Lambda(M)$, the additive error $\epsilon_{Y_S}$ of $Y_S=\log \Lambda(M)$ due to shot noise is related to the error $\epsilon_{\Lambda(M)}$ as:
\begin{equation}
    \epsilon_{Y_S}\leq \frac{\epsilon_{\Lambda(M)}}{\Lambda(M)-\epsilon_{\Lambda(M)}}\leq \frac{\epsilon_{\Lambda(M)}}{C_{\Lambda}}
\end{equation}
for some constant $C_{\Lambda}$.

For any syndrome class $C$, $X_{C}=\log \nu_{C}$ can be written as linear function of $\{Y_{S}\}$ using a bounded number of recursions and resulting in a bounded number $C_r$ of terms of $Y_S$ according to \cref{equation: recursive solving}. Therefore, it is sufficient that the additive error for each $\epsilon_{\Lambda(M)}\leq \min\left(\frac{C_{\Lambda}}{C_r}\epsilon, \lambda_{\mathrm{min}}^{r_{\Gamma} s_{\Gamma} c_s c_{\Gamma}}\right)=\epsilon^{\mathrm{max}}_{\Lambda}$ with confidence $\leq\frac{\delta}{C_r}$ (union bound).
Then we can use Hoeffding's inequality:
\begin{equation}
    P\left(\left|\frac{1}{N}\sum_i m_i-\Lambda(M)\right|\geq \epsilon^{\mathrm{max}}_{\Lambda}\right)\leq 2\exp(-\frac{N(\epsilon^{\mathrm{max}}_{\Lambda})^2}{2})\leq \frac{\delta}{C_r}.
\end{equation}
to obtain the required sample size
\begin{equation}
    N\geq \frac{2}{(\epsilon^{\mathrm{max}}_{\Lambda})^2}\log(\frac{2 C_r}{\delta})= \mathcal{O}\left(\frac{1}{\epsilon^2}\log(\frac{1}{\delta})\right).
\end{equation}
\end{proof}

With the above lemma, we can prove the sample complexity of estimating syndrome class error rates to additive precision for general qLDPC stabilizer codes (qLDPC spacetime subsystem code) under local sparse Pauli channels, when the second-order correction is smaller than the shot noise due to finite sample size.

%\begin{theorem}\label{theorem: syndrome class learning sample complexity}
\noindent\textbf{\cref{theorem: syndrome class learning sample complexity} (Restated)}[Sample complexity for syndrome class learning]
    Consider a local sparse Pauli channel $\mathcal{N}_{\Gamma}$ on an $n$-qubit qLDPC code. Given $\epsilon>0,0<\delta<1$, there exists a constant $C_c>0$ such that if $\|\vec{p}\|^2_{\infty}<\frac{\epsilon}{2C_c}$, the sample size $N$ required to estimate the syndrome error class $P_C, \forall C\in\mathcal{C}^*$ to additive error within $\epsilon$ with confidence $1-\delta$ using \cref{equation: recursive solving} is
    \begin{equation}\label{equation: sample scaling syndrome class}
        N=\mathcal{O}(\frac{1}{\epsilon^2}\log(\frac{1}{\delta})),
    \end{equation}
    independent of the system size $n$.

\begin{proof}
    For any syndrome class $C\in \mathcal{C}^{*}$, we can approximate the syndrome error class using
    \begin{align}
        P_C&=-\frac{1}{2}\log \nu_C + \mathcal{O}(\|\vec{p}\|^2).
    \end{align}
    And we can bound the correction term
    \begin{equation}\label{equation: define C_c}
        |P_C+\frac{1}{2}\log\nu_C|\leq C_c p_{\mathrm{max}}^2
    \end{equation}
    for some $C_c>0$ due to Taylor's theorem. Moreover, $C_c$ is constant in the context of local sparse Pauli channel on qLDPC codes (see \cref{lemma: bounded first and second derivative}).
    If error rates are small enough, i.e., $p_{\mathrm{max}}<\left(\frac{\epsilon}{2C_c}\right)^{\frac{1}{2}}$, then shot noise on $\log \nu_C$ dominates the additive error and the sample size scaling is given by \cref{lemma: sample size log nu}.
    
\end{proof}

\subsection{Sample complexity for learning logical error rates\label{subsection: logical sample complexity}}
In this section, we show that for a non-adaptive Clifford circuit that is fault-tolerant to a circuit-level Pauli noise, which corresponds to a local sparse Pauli Channel on the spacetime code of the circuit, then the logical error rate can be learned to the same relative precision as the physical error rate with only an additional overhead in the sample complexity that in the worst case scales quadratically in the spacetime volume. In the low physical error rate regime where the minimum weight fault paths dominate, the overhead is $\Theta(d^2)$.

We prove the sample complexity of learning the logical error to a given relative precision by first proving the sample size required to learn each physical error rate to a certain relative precision and then proving the bounded error propagation to the logical error rate.

\begin{applemma}[Sample complexity for physical error rate estimation]\label{lemma: Physical error rate relative precision}
Consider a non-adaptive logical Clifford circuit with $n$ physical qubits, and of physical depth $T$ that corresponds to a qLDPC spacetime code. Then consider a local sparse Pauli channel $\mathcal{N}_{\Gamma}$ with local error rates $\vec{p}$, mapped from the circuit-level Pauli error model (\cref{definition: Circuit-level Pauli error model}), that acts on the spacetime code. Let $\vec{p}\,'$ be the error rates that generate the same logical error rate as $\vec{p}$ and satisfy any chosen extra constraints $B\vec{p}\,'=0$. Let $L_H, K_{G_B}$ be the two constants in \cref{lemma: existence of solution}. For any chosen $\eta>0, 0<\delta<1$, if the true error rates satisfy $\|\vec{p}\|_{\infty}\leq \min\left(\frac{2}{L_H, K_{G_B}},\quad\left(\frac{\eta P_{C\mathrm{min}}}{2C_c}\right)^{\frac{1}{2}}\right)$, the sample size $N$ required to learn all physical error rate with relative error within $\eta$ with confidence $1-\delta$, when compared to $\vec{p}\,'$, is given by
    \begin{equation}\label{equation: sample scaling physical relative precision}
        N=\mathcal{O}\left(\frac{1}{\eta^2P^2_{C\mathrm{min}}}\log(\frac{nT}{\delta})\right).
    \end{equation}
\end{applemma}
\begin{proof}
    We use the circuit-to-code mapping in \cref{section: spacetime code mapping} to map the $n$-qubit non-adaptive Clifford circuit of depth $T$ to an $n(T+1)$-qubit spacetime code, which is assumed to be qLDPC.
    According to \cref{lemma: existence of solution} and \cref{theorem: syndrome class learning sample complexity}, given any extra constraints $B\vec{p}=\vec{0}$ and by setting $\epsilon=P_{C\mathrm{min}}\eta$, when
    \begin{equation}
    \|\vec{p}\|_{\infty}\leq \min\left(\frac{2}{L_H K_{G_B}},\quad\left(\frac{\eta P_{C\mathrm{min}}}{2C_c}\right)^{\frac{1}{2}}\right),
    \end{equation}
    there exist error rates $\vec{p}\,'$ satisfying the extra constraints $B\vec{p}=0$ and having identical learnable transformed eigenvalues $\log \vec{\nu}(\vec{p}\,')=\log \vec{\nu}(\vec{p})$ with the ones of the true noise. And we need $N=\frac{\alpha}{\epsilon^2}\log(\frac{1}{\delta})$ to reach $\epsilon$ additive error in estimating each syndrome class error rate, for some constant $\alpha$. We further estimate each component of $\vec{p}\,'$ by applying $B\vec{p}=0$ to each estimated syndrome class error rate. To reach relative error $\eta$ for all syndrome class with confidence $1-\delta$, for each syndrome class $C$, it is sufficient that $\epsilon\leq \eta P_{C\mathrm{min}}$ with confidence $1-\delta_C$ where $\delta_C=\frac{\delta}{|\mathcal{C}^*|}$ (union bound). Moreover, the number of nontrivial syndrome classes $|\mathcal{C}^*|=\Theta(nT)$ since each syndrome class has a constant size. Therefore, plugging in \cref{equation: sample scaling syndrome class}, we get the scaling as shown in \cref{equation: sample scaling physical relative precision}. Each physical error rate is estimated with relative precision within $\eta$ compared with $\vec{p}\,'$ since $B\vec{p}\,'=0$ is satisfied.

\end{proof}

\begin{applemma}[Relative error propagation to logical level]\label{lemma: Relative error propagation to logical level}
    Consider a logical Clifford circuit with $n$ physical qubits and of physical depth $T$, and a circuit-level Pauli noise channel that corresponds to a local sparse Pauli channel $\mathcal{N}_{\Gamma}$ on the spacetime code. Conditioned on the positivity condition in \cref{lemma: existence of solution}, we can achieve learning the probability of any subset of logical Pauli errors that affect logical measurements in the circuit with relative error within $\eta_L$ if the maximal relative error $\eta^{\mathrm{max}}_p$ of the learned physical error rates is $\eta^{\mathrm{max}}_p\leq \frac{1}{2|\Gamma|}\eta_L=\Omega(\frac{\eta_L}{nT})$, where $|\Gamma|$ is the number of local channels.
\end{applemma}
\begin{proof}
    We define a fault path $f$ in a spacetime code as a subset of the possible errors $\mathcal{E}_{\Gamma}$ (exclusive in each local channel). We call $f$ a failing fault path if it leads to any logical error considered after decoding, and we denote the set of all such failing paths as $\mathcal{F}_{\mathrm{fail}}$. For any local channel $\mathcal{N}_{\gamma},\, \gamma \in \Gamma$ and a local error $e$ in the channel, we denote the error rate as $p_e$, and denote $q_{\gamma}=1-\sum_{e\in\mathcal{E}_{\gamma}}p_e$. We denote the set of local supports not involved in a path $f$ as $\Gamma_f^c$. The logical error rate is then given by:
\begin{equation}
P_{\mathrm{fail}}=\sum_{f\in\mathcal{F}_{\mathrm{fail}}}\prod_{e_i\in f}p'_{e_i}\prod_{\gamma\in \Gamma^{c}_f}q'_{\gamma}.
\end{equation}
We bound the relative error for each term $P^{(f)}=\prod_{e_i\in f}p'_{e_i}\prod_{\gamma\in \Gamma^{c}_f}q'_{\gamma}$ as
\begin{align}
    \eta_{P^{(f)}}&\leq \frac{|\prod_{e_i\in f}p'_{e_i}(1+\eta^{\mathrm{max}}_{p})\prod_{\gamma\in \Gamma^{c}_f}q'_{\gamma}(1+\eta^{\mathrm{max}}_{q})-P^{(f)}|}{P^{(f)}}\\
    &\leq (1+\eta^{\mathrm{max}}_{p})^{|\Gamma|}-1
\end{align}
where $\eta^{\mathrm{max}}_q$ refers to the maximal relative error of $q_{\gamma}$ over $\gamma\in\Gamma$, and we used the fact that for small error rates, $\eta^{\mathrm{max}}_q < \eta^{\mathrm{max}}_p$. Since $P_{\mathrm{fail}}$ is just a sum of these positive terms, the relative error $\eta_{P_{\mathrm{fail}}}$ of $P_{\mathrm{fail}}$ is just bounded by the largest $\eta_{P^{(f)}}$ among all failing paths $f$:
\begin{equation}
    \eta_{P_{\mathrm{fail}}}\leq (1+\eta^{\mathrm{max}}_{p})^{|\Gamma|}-1.
\end{equation}
Therefore we have
\begin{equation}
    \eta^{\mathrm{max}}_{p}\geq (1+\eta_{P_{\mathrm{fail}}})^{\frac{1}{|\Gamma|}}-1 \geq e^{\frac{1}{|\Gamma|}\ln (1+\eta_{P_{\mathrm{fail}}})}-1\geq \frac{1}{2|\Gamma|}\eta_{P_{\mathrm{fail}}}, \qquad \forall\: 0\leq \eta_{P_{\mathrm{fail}}}\leq 1,\: |\Gamma|\geq 1,
\end{equation}
where $|\Gamma|=\mathcal{O}(nT)$ is the number of local Pauli channels since we assume $\mathcal{N}_{\Gamma}$ is a local sparse Pauli channel, and we used the fact that $e^x - 1\geq x$. As a result, if $\eta^{\mathrm{max}}_{p}\leq \frac{1}{2|\Gamma|}\eta_L$, then the relative error $\eta_{P_{\mathrm{fail}}}\leq\eta_L$.
\end{proof}

%\begin{theorem}[Sample complexity for logical error rate learning]\label{theorem: logical sample complexity}
\noindent\textbf{\cref{theorem: logical sample complexity} (Restated)}
    Consider a non-adaptive logical Clifford circuit with $n$ physical qubits and of physical depth $T$ that corresponds to a qLDPC spacetime code. We assume the circuit is fault-tolerant to a circuit-level noise that corresponds to a local sparse Pauli channel $\mathcal{N}_{\Gamma}$ on the spacetime code, for a set of logical measurements $\mathcal{M}_L$. Denote by $P_{C\mathrm{min}}$ the minimal syndrome class error rate. We can learn the logical error rates $P_{\mathrm{fail}}$ for $\mathcal{M}_L$ with relative error within $\eta_L$ with confidence $1-\delta$ using a sample size
    \begin{equation}
                N=\mathcal{O}\left(\frac{n^2T^2}{\eta_L^2P^2_{C\mathrm{min}}}\log(\frac{nT}{\delta})\right),
    \end{equation}
    if the true error rates satisfy $\|\vec{p}\|_{\infty}\leq \min\left(\frac{2}{L_H K_{G_B}},\quad\left(\frac{\eta_L P_{C\mathrm{min}}}{4|\Gamma|C_c}\right)^{\frac{1}{2}}\right)$, and for all syndrome classes $C\in\mathcal{C}^*$, $P_C\geq \frac{1}{4}(K_H+ L^2_H K_{G_B})\|\vec p\,^*\|_{\infty}^2$, where constants $L_H, K_{G_B}, K_H, C_c$ are defined in \cref{equation: define K_{G_B},equation: define K_H,equation: define L_H,equation: define C_c},

\begin{proof}
    Directly follows from \cref{lemma: Physical error rate relative precision} and \cref{lemma: Relative error propagation to logical level}.

\end{proof}
For a Clifford circuit encoded in a family of stabilizer codes that have a nontrivial fault-tolerant threshold, when the physical error rates are well below the threshold, then the failing fault path is dominated by the minimum-weight paths. Let $d$ denote the circuit distance (distance of the spacetime code). Then the minimum weight is the minimum number of local channels required to generate a failing path and it is upper bounded by $t = \lfloor\frac{d+1}{2}\rfloor$. In this setting, the sample size needed to reach the same relative precision at the logical level as the physical level has an extra overhead that scales as $\mathcal{O}(d^2)$ instead of $\text{exp}(\frac{d}{2})$ from the logical level fidelity estimation/tomography directly.

However, \cref{theorem: logical sample complexity} does require the upper bound of physical error rate to drop as $\sim\frac{1}{n}$, or $\sim\frac{1}{d}$ when the min-weight fault paths dominate the logical error rate. This is because the correction term in the syndrome class error rate eventually dominates as the physical relative precision required increases with larger $n$. Therefore, if the scale of physical error rate is independent of the system size, then the quadratic overhead in the theorem holds within a finite-sized system, which is still practically relevant. See numerical results in \cref{fig: surface_logical_full_appendix}, where we tested on surface code of various sizes in the simple channel coding setting.

\section{Learning logical error rate under regional depolarizing channel}

\subsection{Logical error rate from regional depolarizing channel\label{appendix: lower bound logical error rate regional depolarizing}}

In this section, we investigate methods of estimating logical error rate for Pauli channels violating the condition in \cref{theorem: logical error rate learnability}. This is potentially possible given some assumptions about the channel. More specifically, on top of a background Pauli channel $\mathcal{N}_{b}$ that satisfies the learnability condition in \cref{theorem: logical error rate learnability}, we consider concatenating it with an additional regional depolarizing channel with error rate $p_{R}$:
\begin{equation}
    \mathcal{N}_R =(1-p_{R})\rho + \frac{p_{R}}{4^{|R|}-1}\sum_{e\in\mathcal{P}_{R}\backslash \{I\}} e\rho e^{\dagger}
\end{equation}
acting on a region $R\subseteq [n]$, where we define $\mathcal{P}_R=\{e\in\mathcal{P}_n|\text{supp}(e)\subseteq R\}$ as the set of $n$-qubit Pauli operators supported in $R$.

We denote the probability distribution of the Pauli error~$E$ in the background Pauli channel $\mathcal{N}_b$ and Pauli error~$F$ in the depolarizing channel $\mathcal{N}_R$ as $P_b$ and $P_R$ respectively. Denote the conditional logical error rate from the background Pauli channel $\mathcal{N}_b$ as $P(\text{fail}\mid e)$, where $e$ is a Pauli error from $\mathcal{N}_b$. We derive the following lemma:
\begin{applemma}[Logical error rate bound\label{lemma: logical error rate lower bound}]
     If the support $R$ of the depolarizing channel $\mathcal{N}_R$ contains some (dressed) logical operator~$L$, then the logical error rate $P_{\mathrm{fail}}$ is bounded by
\begin{equation}
    (1-p_R)P^{\text{bg}}_{\mathrm{fail}} + \left(\frac{1}{2} - e^{-\Omega(|R|)}\right) p_R \leq P_{\mathrm{fail}}\leq(1-p_R)P^{\text{bg}}_{\mathrm{fail}} + p_R.
\end{equation}
where $p_R$ is the depolarizing channel total error rate, and $P^{\text{bg}}_{\mathrm{fail}} = P(\text{fail}|F=I)$ is the logical error rate without the regional depolarizing channel $\mathcal{N}_R$.
\end{applemma}
\begin{proof}
    The logical error rate is given by
    \begin{equation}
        P_{\mathrm{fail}} = \sum_e P(\text{fail}\mid e)P_b(e).
    \end{equation}
    Denote the indicator function of an error $e\in\mathcal{P}_n$ being correctable as $\chi(e)$, where $\chi(e)=1$ when $e$ is correctable by a decoder and $\chi(e)=0$ otherwise. And note that for any given decoder, $\chi(e)+\chi(eL)\leq 1$ for any logical operator $L$. Denote the probability of error being corrected conditioned on the background error being $e$ and a nontrivial error $f$ in $\mathcal{N}_R$ happening by $P(\text{correct}\mid E=e, F\neq I)$. Since a logical operator $L$ is supported by the region $R$, we have
    \begin{align}
        P(\text{correct}\mid E=e,F\neq I)&=\frac{1}{p_R}\sum_{f\in\mathcal{P}_R\backslash \{I\}}P_R(f)\chi(ef)\\
        &=\frac{1}{p_R}\sum_{f\in\mathcal{P}_R\backslash \{L\}}P_R(fL)\chi(efL).
    \end{align}
    Moreover, we have $\forall f\in\mathcal{P}_R\backslash\{I,L\},\,  P_R(f)=P_R(fL)$, therefore
    \begin{align}
        2 P(\text{correct}\mid E=e,F\neq I)
        &=\frac{1}{p_R}\Bigl(2P_R(L)\chi(eL) + \sum_{f\in\mathcal{P}_R\backslash\{I,L\}}\left[P_R(f)\chi(ef) + P_R(fL)\chi(efL)\right]\Bigr)\\
        &=\frac{1}{p_R}\Bigl(2P_R(L)\chi(eL) + \sum_{f\in\mathcal{P}_R\backslash\{I,L\}}P_R(f)(\chi(ef)+\chi(efL))\Bigr)\\
        &\leq\frac{1}{p_R}\Bigl( 2P_R(L)\chi(eL) + \sum_{f\in\mathcal{P}_R\backslash\{I,L\}}P_R(f)\Bigr)\\
        &\leq 1 + \frac{2}{4^{|R|}-1}\\
        &=1 + e^{-\Omega(|R|)}.
    \end{align}
    As a result,
    \begin{align}
        P_{\mathrm{fail}} &= \sum_e P(\text{fail}\mid e)P_b(e)\\
        &=(1-p_R)P(\text{fail}|F=I) + p_R\sum_e P_b(e)P(\text{fail}\mid E=e,F\neq I)\\
        &\geq(1-p_R)P(\text{fail}|F=I) + \left(\frac{1}{2} - e^{-\Omega(|R|)}\right) p_R
    \end{align}

Therefore, it suffices to estimate the depolarizing channel total error rate to give a lower bound of the logical error rate. To obtain an upper bound, since
\begin{equation}
    P(\text{fail}\mid E=e,F\neq I) \leq 1
\end{equation}
we have
\begin{equation}
    P_{\mathrm{fail}}\leq(1-p_R)P(\text{fail}|F=I) + p_R
\end{equation}
and $P(\text{fail}|F=I)$ can be learned from the syndrome data post-selected on $F=I$ according to \cref{theorem: logical error rate learnability}.

\end{proof}

\subsection{Single-shot detection of regional depolarizing channel\label{appendix: single-shot detection}}
In this section, we consider a local sparse Pauli channel concatenated with a depolarizing channel over the region $R$ acting on a qLDPC code. We argue that it is possible to use syndrome data to detect rare high-weight Pauli error events, modeled as a depolarizing channel on some region $R$. In particular, letting $m$ be the number of measured stabilizer generators, we look at the average syndrome value $M=\frac{1}{m}\sum_a X_a$, where $X_a$ is the value of the stabilizer generator $S_a$. 

Recall that for two Pauli channels $\mathcal{N}_{\gamma_1}$ and $\mathcal{N}_{\gamma_2}$ we have the Pauli eigenvalue of $e\in\mathcal{P}_n$ for the channel $\mathcal{N}_{\gamma_1}\circ \mathcal{N}_{\gamma_2}$ given by $\Lambda(e)=\Lambda_{\gamma_1}(e)\Lambda_{\gamma_2}(e)$.
Denote as $\mathcal{G}_{R}$ the stabilizer generators that overlap with the region $R$ and let $g_R=|\mathcal{G}_{R}|$, conditioned on the a nontrivial error $F\neq I$ happened in $\mathcal{N}_R$, we have
\begin{align}\label{equation: exponentially small in R}
    \text{E}(X_a \mid F\neq I)&=\text{E}_b(X_a \mid F\neq I)\text{E}_R(X_a \mid F\neq I)\\
    &=\frac{\lambda_a}{p_R} \left(\sum_{e\in\mathcal{P}_R\backslash\{I\}} (-1)^{\langle e, S_a\rangle} \frac{p_R}{4^{|R|}-1} \right)\\
    &=-\frac{\lambda_a}{4^{|R|} - 1}\\
    &=-e^{-\Omega(|R|)}
\end{align}
On the other hand, for any stabilizer $X_a$ with support nonoverlapping with $R$, we have $E(X_a \mid F\neq I)=\Lambda(X_a)=\lambda_a$, where we denote $\lambda_a$ as the syndrome expectation values of $S_a$ without the added depolarizing channel $\mathcal{N}_R$. As a result, condition on $F\neq I$, the average syndrome up to an exponentially small correction is given by
\begin{equation}
    E(M \mid F\neq I) = -e^{-\Omega(|R|)} + \frac{1}{m}\sum_{a\mid S_a \notin \mathcal{G}_R} \lambda_a \approx (1-\frac{g_R}{m})\overline{\lambda}.
\end{equation}
where $g_R = |\mathcal{G}_R|$, and $\overline{\lambda}$ is the average syndrome expectation value due to the background local noise channel, i.e., without the regional depolarizing channel.
Therefore the shift of $M$ due to the nontrivial error from $\mathcal{N}_R$ is
\begin{equation}
    \Delta M=\frac{g_R}{m}\overline{\lambda}=\Theta\left(\frac{|R|}{n}\right).
\end{equation}
On the other hand, the variance of $M$
\begin{equation}
    \text{Var}(M|F\neq I)=\frac{1}{m^2}\sum_{a,b}\text{Cov}(X_a, X_b)=\frac{1}{m^2}\sum_{a,b}[\text{E}(X_a, X_b)-\text{E}(X_a)\text{E}(X_b)].
\end{equation}
Due to \cref{equation: exponentially small in R}, the expectation value is exponentially small in $R$ conditioned on a nontrivial error $F\neq I$ for $\mathcal{N}_R$. Therefore, the only contribution to the variance is the variance of each $X_a$ and the pairs of stabilizer generators that are either simultaneously flipped by some background local Pauli errors, or have identical and nontrivial components on $R$. Since for each stabilizer generator, there are only $\mathcal{O}(1)$ stabilizer generators that can form such a pair for a local-sparse background Pauli channel on a qLDPC code, we have
\begin{equation}
    \text{Var}(M|F\neq I) = \mathcal{O}\left(\frac{1}{m}\right)=\mathcal{O}\left(\frac{1}{n}\right).
\end{equation}
Similarly,
\begin{equation}
    \text{Var}(M|F=I) =\mathcal{O}\left(\frac{1}{n}\right).
\end{equation}

A simple criterion for detectability is therefore comparing the signal $\Delta M$ with the noise $\text{std}(M)$. And the border line is when $|R|=\Theta(\sqrt{n})$. However, when the distance satisfies $d=\omega(\sqrt{n})$, the logical error rate can directly be learned when $|R|=\mathcal{O}(\sqrt{n})$ as shown in \cref{theorem: logical error rate learnability}. Another learnable regime is $|R|=\omega(\sqrt{n})$, in which case the nontrivial error can be detected in real time. A lower and upper bound of the logical error rate can be estimated via \cref{lemma: logical error rate lower bound} by estimating $p_R$ from the detection method and estimating the background logical error rate using the post-selected syndrome data.

\section{\label{appendix: Intraclass uniformity}Intraclass uniformity constraints}

\begin{figure}[h!]
\centering
\begin{subfigure}{0.23\textwidth}
\phantomcaption
\stackinset{l}{2pt}{t}{2pt}{\captiontext*}
    {\includegraphics[width=\textwidth]{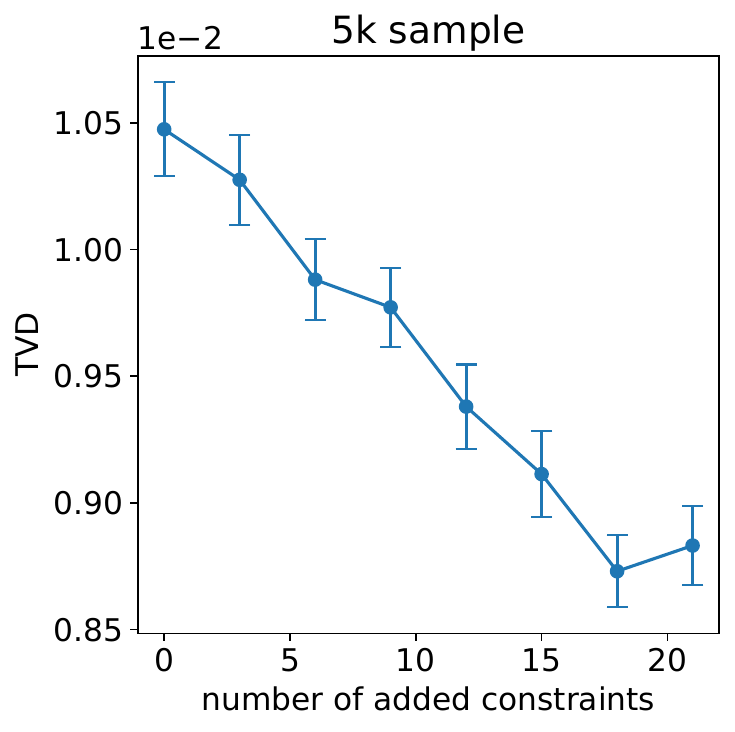}}
\end{subfigure}
\begin{subfigure}{0.23\textwidth}
\phantomcaption
\stackinset{l}{2pt}{t}{2pt}{\captiontext*}
    {\includegraphics[width=\textwidth]{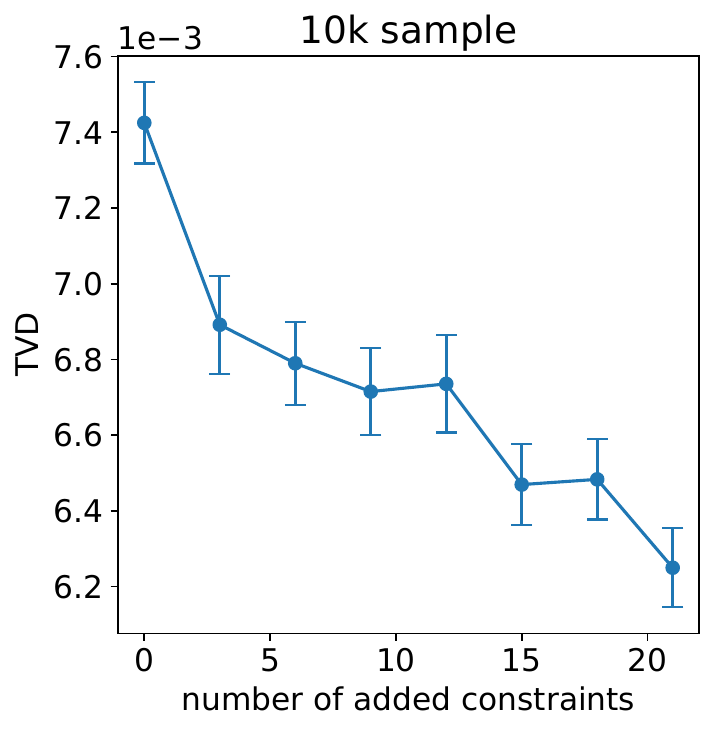}}
\end{subfigure}
\begin{subfigure}{0.23\textwidth}
\phantomcaption
\stackinset{l}{2pt}{t}{2pt}{\captiontext*}
    {\includegraphics[width=\textwidth]{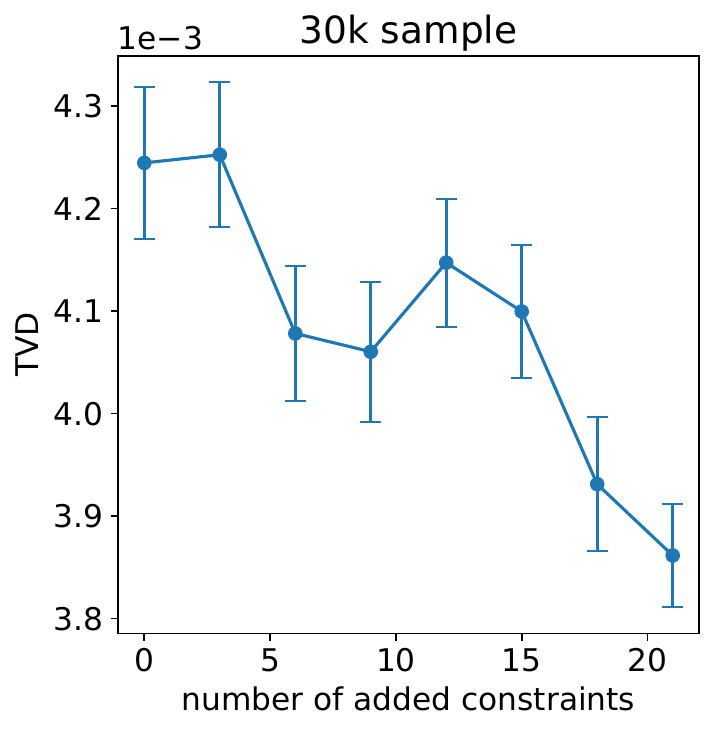}}
\end{subfigure}
\begin{subfigure}{0.23\textwidth}
\phantomcaption
\stackinset{l}{2pt}{t}{2pt}{\captiontext*}
    {\includegraphics[width=\textwidth]{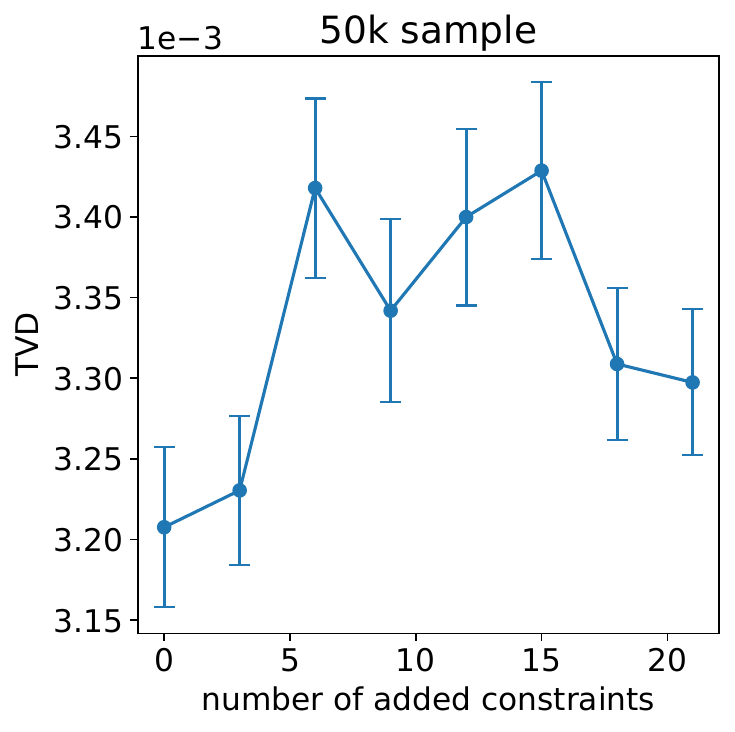}}
\end{subfigure}
\begin{subfigure}{0.23\textwidth}
\phantomcaption
\stackinset{l}{2pt}{t}{2pt}{\captiontext*}
    {\includegraphics[width=\textwidth]{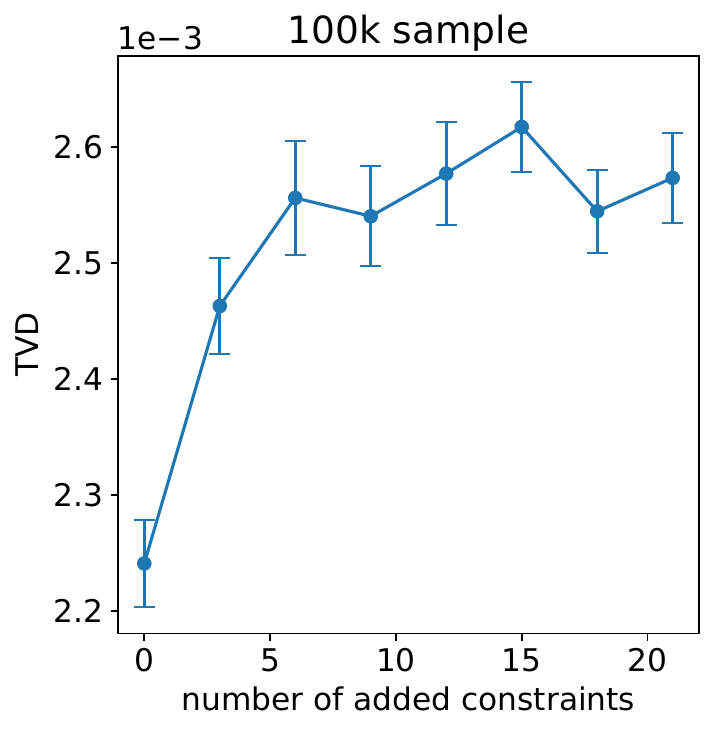}}
\end{subfigure}
\begin{subfigure}{0.23\textwidth}
\phantomcaption
\stackinset{l}{2pt}{t}{2pt}{\captiontext*}
    {\includegraphics[width=\textwidth]{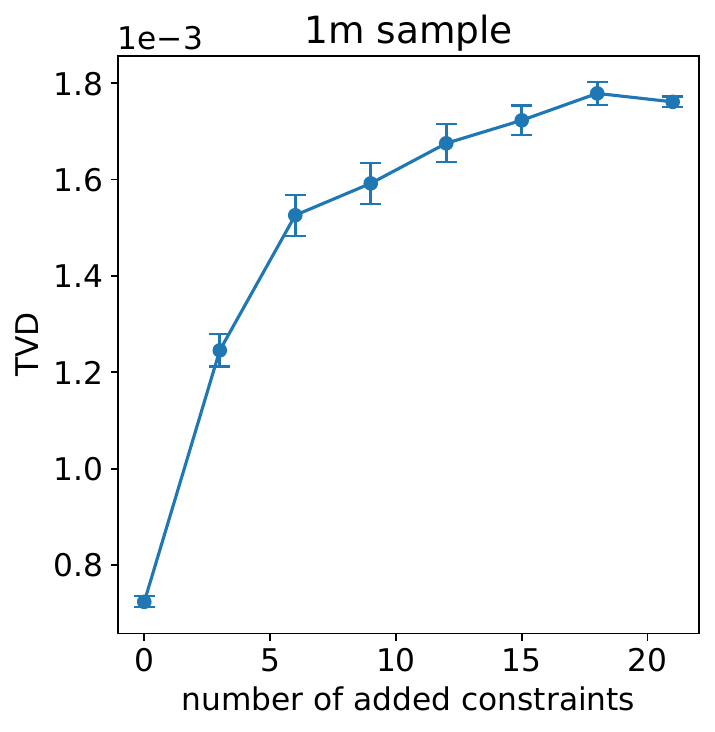}}
\end{subfigure}
\begin{subfigure}{0.23\textwidth}
\phantomcaption
\stackinset{l}{2pt}{t}{2pt}{\captiontext*}
    {\includegraphics[width=\textwidth]{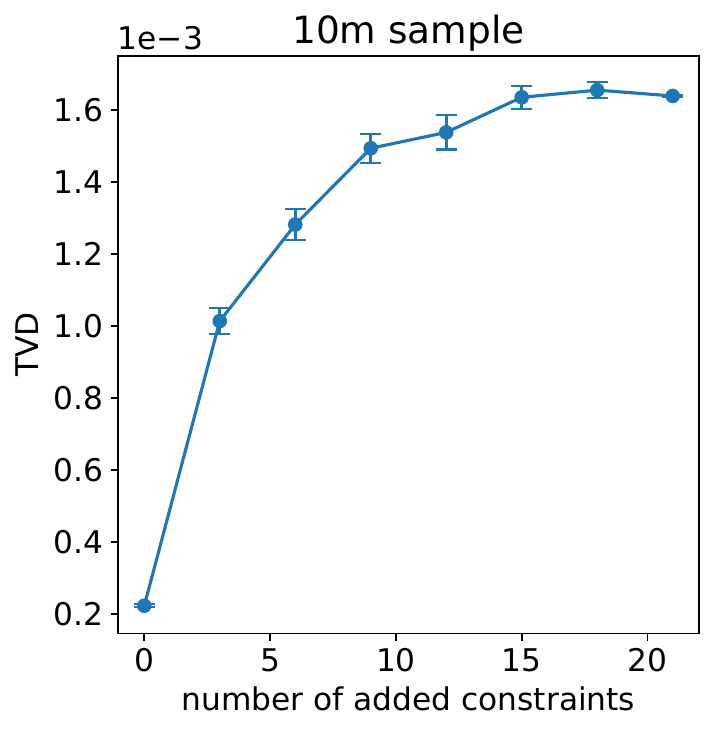}}
\end{subfigure}
\begin{subfigure}{0.23\textwidth}
\phantomcaption
\stackinset{l}{2pt}{t}{2pt}{\captiontext*}
    {\includegraphics[width=\textwidth]{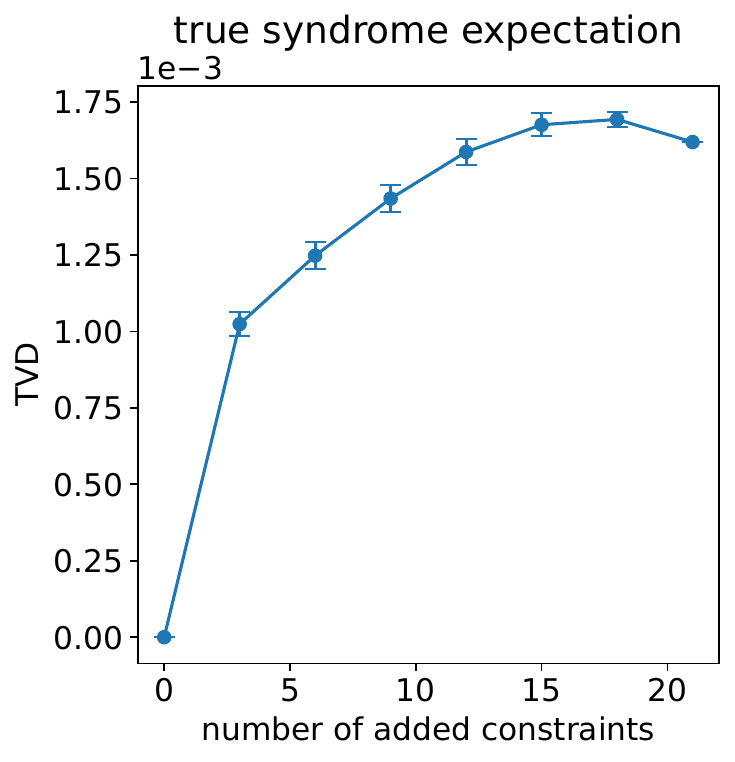}}
\end{subfigure}
\caption{Numerical simulation of the effect of extra uniformity constraints. The horizontal axis is the number of uniformity constraints acting on random error components. The vertical axis corresponds to the TVD of the learned error rate and the true noise. The uniform constraints benefit small sample sizes and large shot noise samples while degrading the quality of learning in the case of large sample sizes (weak shot noise). Panels (a)--(g) correspond to sample sizes $N=5{,}000$, $10{,}000$, $30{,}000$, $50{,}000$, $100{,}000$, $1{,}000{,}000$, and $10{,}000{,}000$, respectively; panel (h) uses exact syndrome expectation values.}
\label{fig: extra uniformity constraints numerics}
\end{figure}
In \cref{section: logical error rate and fidelity}, we add additional constraints to mitigate the overfitting due to shot noise for learning the experimental data, and we denote the learned result with $^*$. For each randomly selected qubit component, we add a new row to the $B$ matrix in \cref{equation: extra constraints}. If the $k$th error $e_k$ in $\mathcal{E}_{\Gamma}$ (corresponds to the $k$th column of $B$) is randomly selected for an additional row of $B$ (say the $i$th row), then the matrix elements of row $i$ are:
\begin{equation}
    B_{ij}=
    \begin{cases}
    1, & \text{if } i=k, \\
    \frac{-r(e_k)}{(|\mathcal{E}_{\Gamma}|-2)\,r(e_{j})} , & \text{else}.
\end{cases}
\label{equation: extra uniform constraints}
\end{equation}

Note that the constraints added here can be optimized depending on the strength of the shot noise. We did numerical simulations on the seven-qubit Steane code using single-qubit Pauli channels. Therefore, we have $21$ parameters corresponding to the $21$ Pauli error components. We see benefits in the large shot noise case, but the extra constraints will have a negative impact on the learning quality when the sample size exceeds a certain threshold. See \cref{fig: extra uniformity constraints numerics}.

\section{\label{appendix: more plots}Extra plots}
\begin{figure*}
    \centering
    \includegraphics[width=1\linewidth]{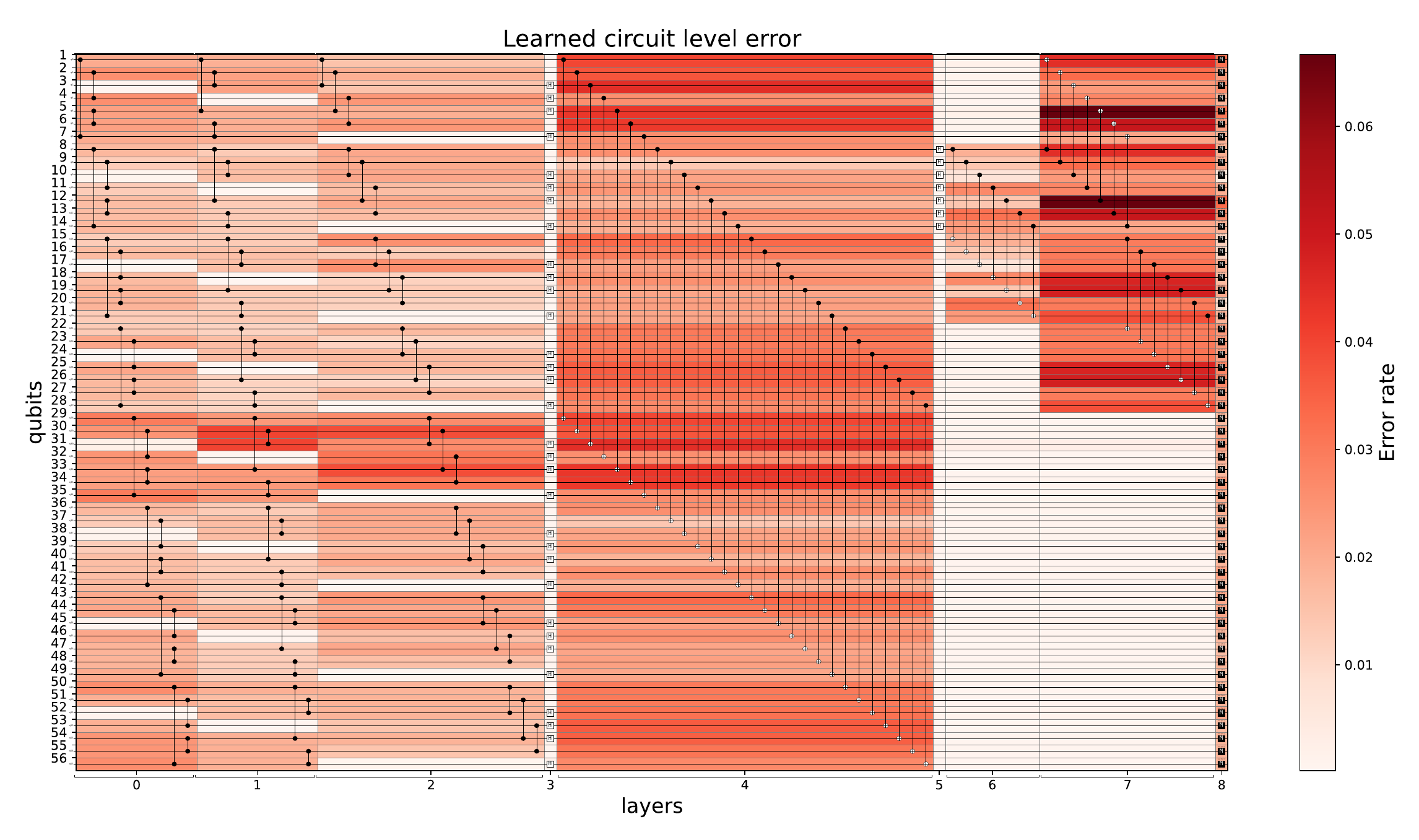}
    \caption{Physical Pauli error rate learning for the $\overline{ZZZZ}$ measurement circuits}
    \label{fig: logical GHZ physical learning appendix}
\end{figure*}
\begin{figure*}
    \centering
    \includegraphics[width=1\linewidth]{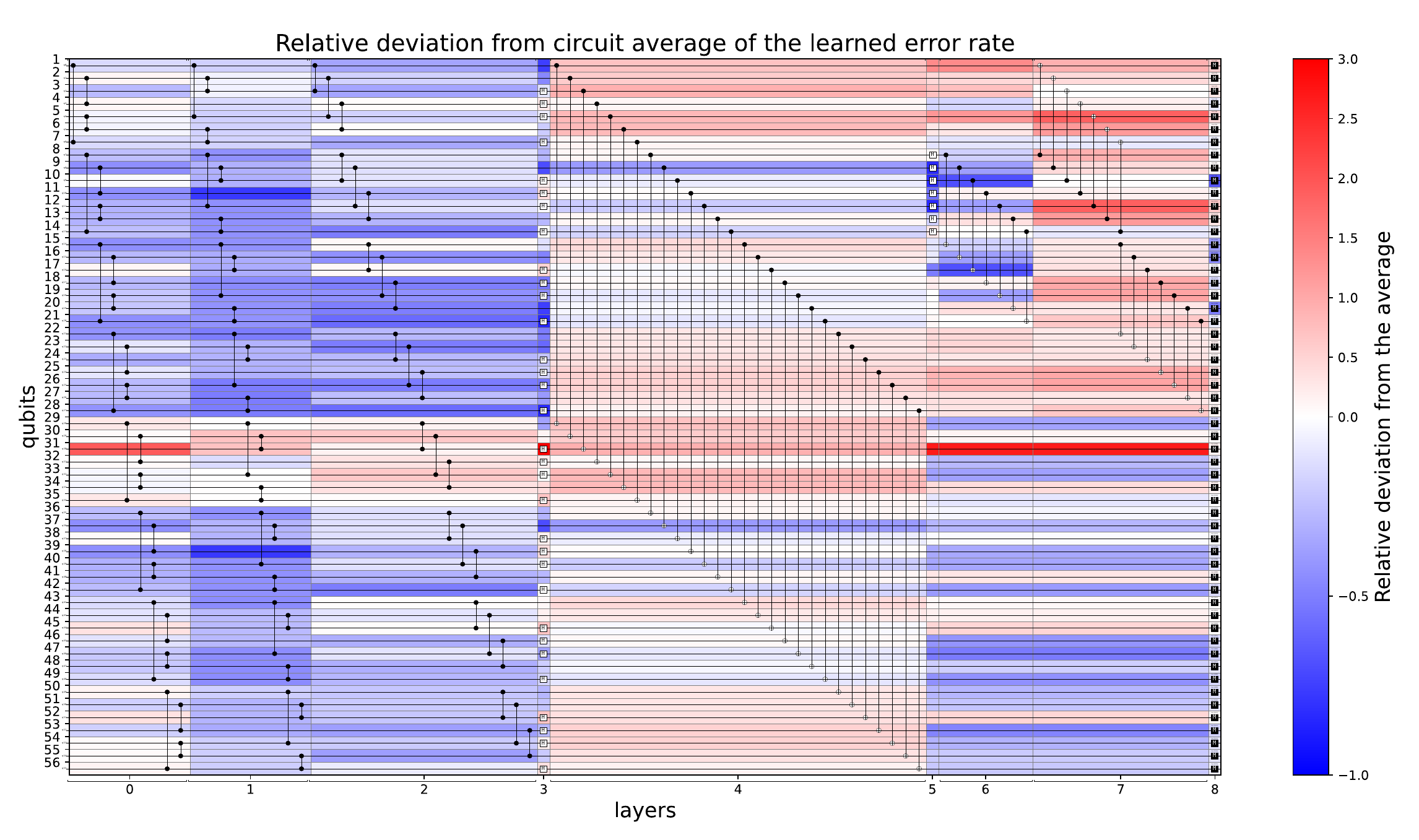}
    \caption{Relative deviation of gate error rates to the circuit average.}
    \label{fig: logical GHZ physical learning relative appendix}
\end{figure*}

\begin{figure*}[t]
\centering
\begin{subfigure}{0.3\textwidth}
    \phantomcaption
    \stackinset{l}{-2pt}{t}{-3pt}{\captiontext*}
    {\hspace{0.3cm}\vspace{1cm}\includegraphics[width=\textwidth]{plots/posts_selection_fidelity.pdf}}
    \label{fig: post_select_full}
\end{subfigure}
\hspace{0.5cm}
\begin{subfigure}{0.3\textwidth}
\phantomcaption
\stackinset{l}{-2pt}{t}{2pt}{\captiontext*}
    {\hspace{0.25cm}\includegraphics[width=\textwidth]{plots/partial_post_selection_flag_0.pdf}}
    \label{fig: flag_0_full}
\end{subfigure}
\hspace{0.5cm}
\begin{subfigure}{0.3\textwidth}
    \phantomcaption
    \stackinset{l}{-2pt}{t}{2pt}{\captiontext*}
    {\hspace{0.25cm}\includegraphics[width=\textwidth]{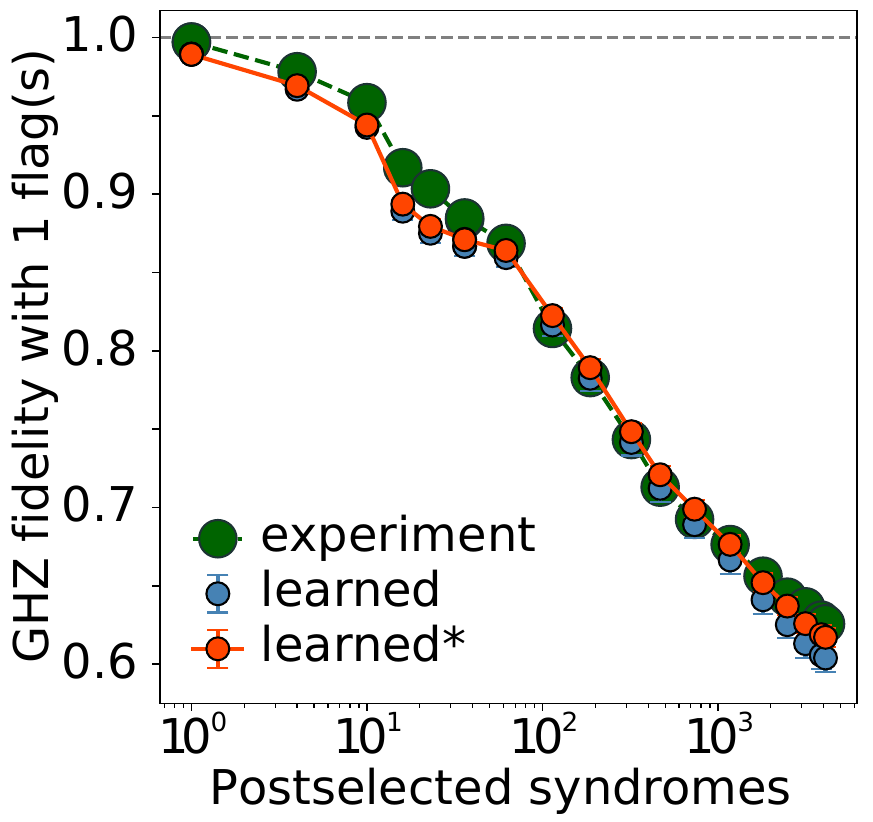}}
    \label{fig: flag_1_full}
\end{subfigure}
\begin{subfigure}{0.3\textwidth}
    \phantomcaption
    \stackinset{l}{-2pt}{t}{2pt}{\captiontext*}
    {\hspace{0.25cm}\includegraphics[width=\textwidth]{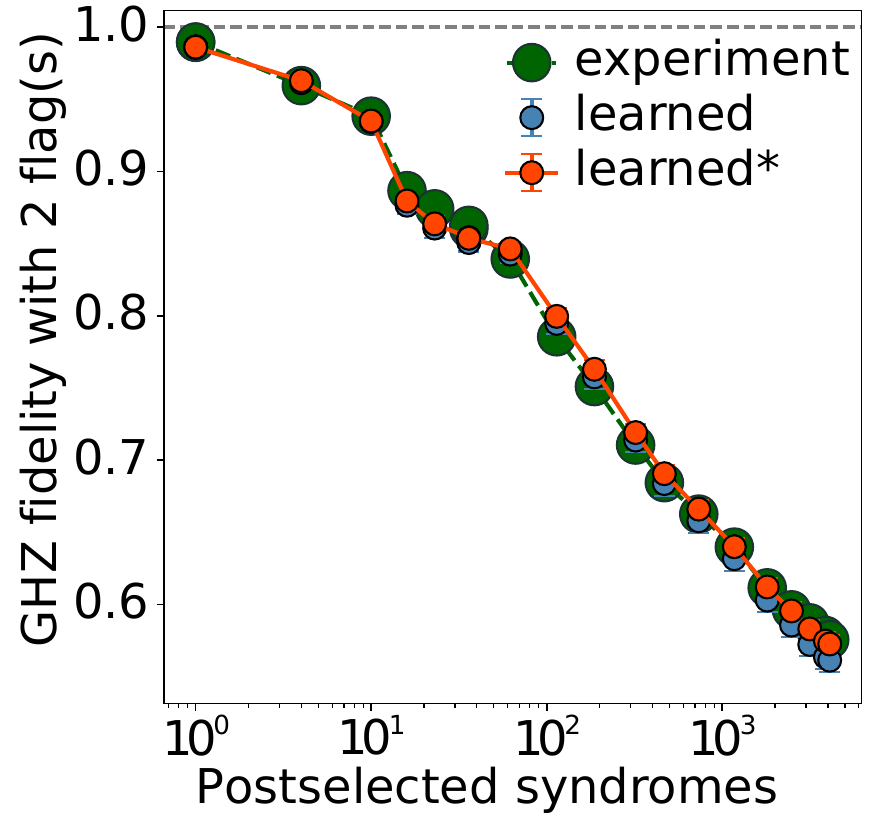}}
    \label{fig: flag_2_full}
\end{subfigure}
\hspace{0.5cm}
\begin{subfigure}{0.3\textwidth}
    \phantomcaption
    \stackinset{l}{-2pt}{t}{2pt}{\captiontext*}
    {\hspace{0.25cm}\includegraphics[width=\textwidth]{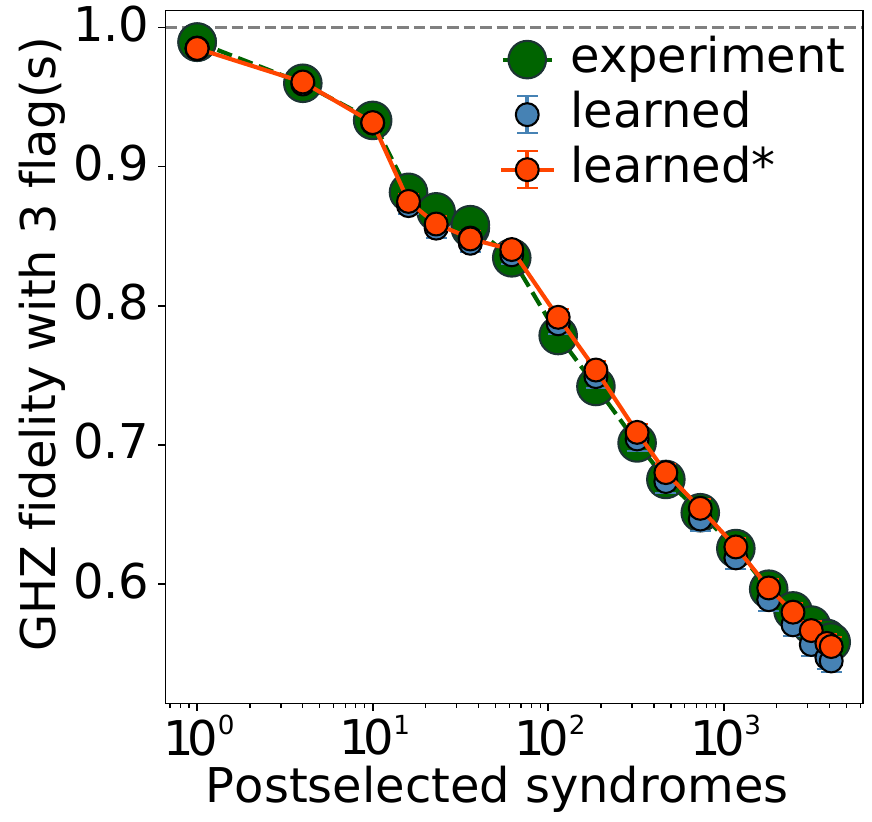}}
    \label{fig: flag_3_full}
\end{subfigure}
\hspace{0.5cm}
\begin{subfigure}{0.3\textwidth}
    \phantomcaption
    \stackinset{l}{-2pt}{t}{2pt}{\captiontext*}
    {\hspace{0.25cm}\includegraphics[width=\textwidth]{plots/partial_post_selection_flag_4.pdf}}
    \label{fig: flag_4_full}
\end{subfigure}
 \caption{Full version of \cref{fig: logical fidelity learning}. Comparison of logical direct fidelity estimation between experimental results and learned noise models from syndromes. (a) The error detecting results for the logical GHZ fidelity from experimental measurement, learned logical fidelity, and the learned fidelity with extra uniformity constraints, denoted with $^*$ (see \cref{appendix: Intraclass uniformity}). (b)-(f) The same comparison of logical fidelity with partial postselection on the syndrome is performed. Each subplot is conditioned on a different number of ancilla blocks that have a nontrivial syndrome (flag).}
\label{fig: logical fidelity learning (full)}
\end{figure*}

\begin{figure*}
    \centering
    \includegraphics[width=0.7\linewidth]{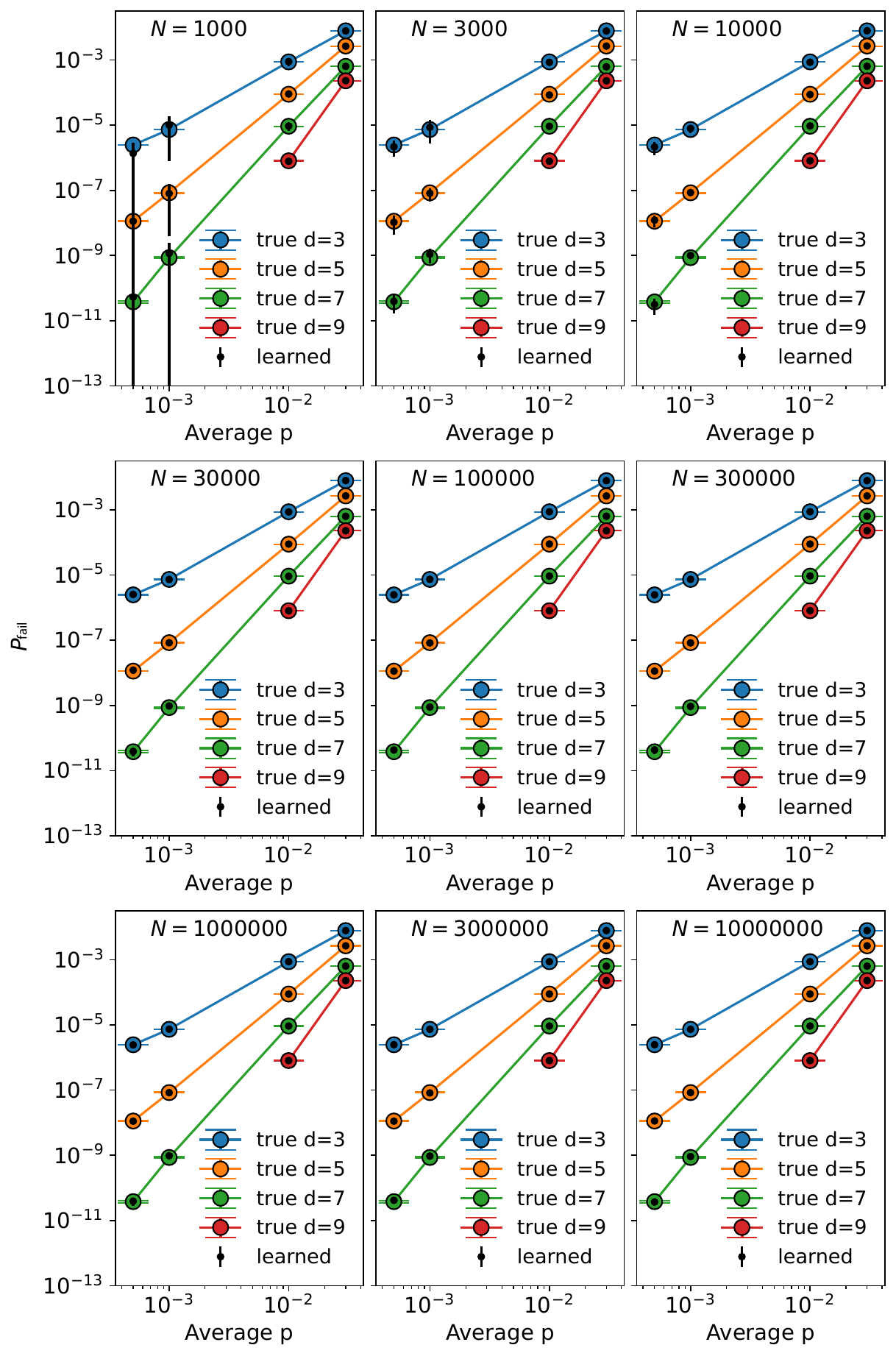}
    \caption{Surface code logical error rate learning in the code capacity setting for distances $d=3,5,7,9$. The learning is done by learning the physical error rate with different sample sizes $N$ of syndrome data.}
    \label{fig: surface_logical_full_appendix}
\end{figure*}

\end{appendix}
\end{document}